\documentclass[reqno]{amsart}
\usepackage[left=2.8cm,right=2.8cm,top=2.8cm,bottom=2.8cm]{geometry}
\usepackage{hyperref}
\usepackage{amsmath}
\usepackage{amssymb}
\usepackage{graphicx}
\usepackage{tikz}
\usetikzlibrary{calc}
\usepackage{xcolor}
\usetikzlibrary{snakes}
\usepackage{amsthm}

\numberwithin{equation}{section}

\newcommand{\EEE}{\color{black}}

\theoremstyle{plain}
\begingroup
\newtheorem{theorem}{Theorem}[section]
\newtheorem{lemma}[theorem]{Lemma}
\newtheorem{proposition}[theorem]{Proposition}

\endgroup

\theoremstyle{definition}
\begingroup
\newtheorem{definition}[theorem]{Definition}
\newtheorem{remark}[theorem]{Remark}

\endgroup

\usepackage[english]{babel}
\begin{document}

\title[Finite crystallization and Wulff shape emergence for ionic compounds] {Finite crystallization and Wulff shape emergence for ionic compounds in the square lattice}

\keywords{Ionic dimers, ground state, configurational energy minimization, crystallization, Wulff shape, square lattice, net charge}

\author{Manuel Friedrich}
\address[Manuel Friedrich]{Applied Mathematics M\"unster, University of M\"unster\\
Einsteinstrasse 62, 48149 M\"unster, Germany.}
\email{manuel.friedrich@uni-muenster.de}

\author{Leonard Kreutz}
\address[Leonard Kreutz]{Applied Mathematics M\"unster, University of M\"unster\\
Einsteinstrasse 62, 48149 M\"unster, Germany.}
\email{lkreutz@uni-muenster.de}

\date{}

\begin{abstract}

We present two-dimensional crystallization results in  the square lattice for finite particle systems consisting of two different atomic types. We identify energy minimizers of configurational energies featuring  two-body short-ranged particle interactions which favor some reference distance between different atomic types  and contain repulsive contributions for  atoms of the same type. We first prove that  ground states are connected subsets of the
square lattice with alternating arrangement of the two atomic types  in  the crystal lattice, 
 and address the emergence of a square macroscopic Wulff shape for an increasing number  of particles.  We then analyze the signed difference of the number of the two atomic types, the so-called net charge, for which we prove the sharp scaling ${\rm O}(n^{1/4})$ in terms of the particle number $n$. Afterwards, we investigate the model  under   prescribed net charge. We provide a characterization for the minimal energy and identify a critical net charge beyond which  crystallization in the square lattice fails. Finally, for this specific net charge we  prove a crystallization result and  identify a diamond-like Wulff-shape of energy minimizers  which illustrates  the sensitivity of the macroscopic geometry on the net charge.

\end{abstract}

\subjclass[2010]{82D25.} 
\maketitle

\section{Introduction}\label{section:introduction}

The question whether the ground states  of  particle systems for certain configurational energies  arrange themselves into  crystalline order is referred to as the {crystallization problem} \cite{Blanc}.  Due to its  paramount theoretical and applicative relevance, this mathematical issue has attracted a great deal of attention over the last decades and has led to various mathematically rigorous crystallization results  for ensembles consisting of \emph{one single} atomic type.  For particle systems with different types of atoms, however, rigorous results appear to be scarce.   The goal of this paper is to  contribute to these fundamental mathematical questions by presenting a study of  crystallization in the two-dimensional square lattice for finite   particle systems consisting of \emph{two different} atomic types.

Microscopically, crystallization  can be seen as the result of interatomic interactions governed by quantum mechanics. At zero or very low temperature, atomic interactions are expected to depend only on the geometry of the atomic arrangement. In this case, configurations can be identified  with the respective positions of identical atoms  $\lbrace x_1, \ldots, x_n \rbrace$.  Then the crystallization problem consists in considering the   minimization of a configurational energy $\mathcal{E}( \lbrace x_1,\ldots,x_n \rbrace )$ and in proving or disproving the periodicity of ground-state configurations of $\mathcal{E}$. 

Various such crystallization results for different choices of the energy $\mathcal{E}$   comprising classical interaction potentials have been derived over the last decades. Here, among the vast body of literature, we only mention some of the relevant works,   and refer the reader to the recent review \cite{Blanc} for a general perspective. Concerning results in one dimension we mention \cite{Gardner,Hamrick,Radi,Ventevogel}, and we refer to \cite{Lucia, HR, Mainini-Piovano, Mainini,  Radin,Wagner83}   for studies in dimension two  for  a finite number of identical particles. In particular, we highlight the work by {\sc Mainini, Piovano, and Stefanelli} \cite{Mainini-Piovano} where a comprehensive analysis of crystallization in the square lattice is performed. Besides crystallization, the authors also provide a fine characterization of ground-state geometries by proving the emergence of a square macroscopic Wulff shape for growing particle numbers. (We also refer to  \cite{Yuen, Davoli15, Davoli16, Mainini-Piovano-schmidt} for similar studies for different lattices.)  Under less restrictive assumptions on the potentials, various results have been obtained in the thermodynamic limit \cite{Betermin0,ELi, Smereka15, Theil}\EEE, i.e., as the number of particles tends to infinity.  In particular, we highlight the recent work by {\sc B\'etermin, De Luca, and Petrache} \cite{Betermin0} where, for a wide class of pair-interaction potentials, crystallization in the square lattice has been obtained.   \EEE The crystallization problem in three dimension seems to be very difficult and only few rigorous results \cite{Flateley1,Flateley2, cronut, Suto06} are  available.

For particle systems involving  different types of atoms, simulations are abundant, but   rigorous results seem to be limited to   \cite{Betermin,B43, FriedrichKreutzHexagonal,B195}.  To the best of our knowledge, the recent work \cite{FriedrichKreutzHexagonal} by the authors represents a first  rigorous mathematical crystallization result for two-dimensional \emph{dimers}, i.e., molecular compounds consisting of two atomic types. This result is  inspired by problems for systems of identical particles and follows the classical molecular-mechanical frame of configurational energy minimization: configurations of $n$ particles are identified with their respective \emph{positions} $\{x_1,\ldots,x_n\}  \subset  \mathbb{R}^{2}$ and additionally with their \emph{types} $\{q_1,\ldots,q_n\}  \in \{-1,1\}^n$. The goal is to determine minimizers of a corresponding  interaction energy $ \mathcal{E}(\{(x_1,q_1),\ldots,(x_n,q_n)\})$ and to characterize their geometry.

More precisely, the energy  $\mathcal{E} =  \mathcal{E}_{\mathrm{a}} + \mathcal{E}_{\mathrm{r}}$ is assumed to consist of two short-ranged two-body interaction potentials $\mathcal{E}_{\mathrm{a}}$ and  $\mathcal{E}_{\mathrm{r}}$, where $\mathcal{E}_{\mathrm{a}}$ represents the interactions between atoms of different type and  $\mathcal{E}_{\mathrm{r}}$ encodes the energy contributions of atoms having the  same type. The   potential $\mathcal{E}_{\mathrm{a}}$ is attractive-repulsive and  favors atoms sitting at some specific reference distance, whereas $\mathcal{E}_{\mathrm{r}}$  is a pure repulsive term. The main result of \cite{FriedrichKreutzHexagonal} is that, for specific quantitative assumptions on the potentials, global  minimizers of the configurational energy are essentially connected subsets of the regular hexagonal lattice with the two atomic types alternating. The first main goal of the present article is to show that   weaker short-ranged repulsive terms  favor crystallization in the square lattice, which illustrates  the sensitivity of the ground-state geometry on the precise assumptions on the potentials.

Let us mention that the  choice of the interaction potentials is motivated by the modelling of ions in ionic compounds.  In fact, one can interpret the two interaction energies as a (very simplified) model for  the bonding of ions capturing the essential features to describe the formation of ionic crystalline solids. The shape of  $\mathcal{E}_{\mathrm{a}}$ is due to electrostatic forces   between ions of opposite charge as well as a small additional force due to van der Waals interactions, resulting in a long ranged attraction. The short ranged repulsive force can be explained by the Pauli exclusion principle when a pair of ions comes close enough such that their outer electron shells  overlap. The balance between these forces leads to a potential energy landscape with a minimum energy when the nuclei are separated by a specific equilibrium distance. On the other hand, the repulsive  energy $\mathcal{E}_{\mathrm{r}}$ between ions of same charge models a simplified Coulomb repulsive force. \EEE We will often refer to the atomic types $\{q_1,\ldots,q_n\}  \in \{-1,1\}^n$ as positive or negative \emph{charges}. In this context, it is interesting to consider the \emph{net charge} $  \sum_{i=1}^n q_i$, i.e.,  the  (signed) difference of the number of the two atomic types.

In \cite{FriedrichKreutzHexagonal}, the \emph{free net charge problem} has been addressed, i.e., the net charge is not preassigned but a fundamental part of the minimization problem. It has been shown that $  \sum_{i=1}^n q_i$ is `almost neutral', with a deviation from zero of order   ${\rm O}(n^{1/4})$, where the scaling in terms of the particle number $n$ is sharp. Similar problems of charge distributions and net charge on Bravais lattices have been recently studied in \cite{Betermin2}. Besides extending the results of the free net charge problem to the present setting, the second main goal of this article is to complement the aforementioned analysis by studying the \emph{prescribed net charge problem}, i.e., the net charge is a given constraint in the minimization problem.    This corresponds to a model of a closed physical system containing a certain number of positively and negatively charged atoms.

We now give an overview of the main results of this article representing a comprehensive analysis of finite crystallization for dimers in the two-dimensional square lattice, see Section \ref{subsection:main results} for details.  
\begin{itemize}
\item[(1)]  We first  address the free net charge problem. Under suitable assumptions on the attractive and the repulsive potentials, we characterize the ground-state energy and geometry of  finite particle  configurations of ions in dimension two. In particular, we prove that each global minimizer of the configurational energy is essentially a connected subset of the square lattice with alternating arrangement of the two atomic types  in  the crystal lattice. This characterization holds except for possibly one atom at the boundary of the configuration. Similar to \cite{Mainini-Piovano}, we identify the emergence of a square macroscopic  Wulff shape for growing particle numbers. 
\item[(2)] We provide a fine asymptotic characterization for the net charge as the number of atoms $n$ grows.  More specifically,  we show that the fluctuation of the net charge around zero   can be at most of order ${\rm O}(n^{1/4})$, i.e., $|\sum_{i=1}^n q_i| \leq c n^{1/4}$ for some constant  $c>0$ independent of $n$. By providing an explicit construction we further prove that this scaling is sharp.
\item[(3)] We consider the prescribed net charge problem and provide a characterization of the minimal energy in dependence of the net charge. We identify a critical net charge,  the   so-called \emph{saturation net charge} $q_{\rm sat}$, which \EEE  corresponds to the case that all atoms of the less frequent atomic type are bonded to exactly four other atoms. By way of example,  see  Fig.~\ref{Fig: qnet bigger},  we show that $q_{\rm sat}$ is critical in the sense that beyond this specific net charge we cannot expect  that atoms arrange themselves in a regular lattice.
\item[(4)] We investigate the geometry of energy minimzers for prescribed net charge $q_{\rm sat}$ and show that also in this case optimal configurations are  (essentially)  subsets of the square lattice. Our interest in this specific case is twofold: (i) this problem together with the free net charge problem (which by (2) essentially corresponds to the problem with prescribed net charge  zero) constitute the extreme cases for which crystallization results can be shown. (ii) for $q_{\rm sat}$, we  identify a macroscopic diamond-like Wulff-shape, which  illustrates \EEE    the dependence of the macroscopic geometry  on  the prescribed net charge.   
\end{itemize}

Our general proof strategy for the free net charge problem follows the induction method on bond-graph layers developed in  \cite{HR, Mainini-Piovano, Mainini, Radin}. Let us mention that our problem is particularly related to  \cite{Mainini-Piovano}, where crystallization for identical particles in the square lattice has been investigated under three-body angular potentials. Actually, the ground-state energy in the present context coincides with the one obtained there.  A crucial point in the induction step is the derivation of a boundary energy estimate. The presence of repulsive instead of angular terms calls for a novel definition of the boundary energy,  complementing the approach in \cite{Mainini-Piovano} from a technical point of view, see Remark \ref{rem: boundary energy}.

To prove the emergence of a square Wulff shape and the sharp scaling ${\rm O}(n^{1/4})$ for the net charge,  we use the fact that in ground states the two atomic types are alternately arranged in the square lattice. This allows us to apply the $n^{3/4}$-law in  \cite{Mainini-Piovano} which states that ground states differ from
a square shape by at most ${\rm O}(n^{3/4})$ atoms, or equivalently, by  at most ${\rm O}(n^{1/4})$ in Hausdorff distance.

Concerning the proof of the results for the  prescribed net charge problem, in principle   we follow the same induction method on bond-graph layers as for the free net charge problem. The actual realization, however, is much more delicate. In fact, as a preliminary step for a crystallization result in the square lattice, a fine characterization of the saturation net charge $q_\mathrm{sat}$ is needed. Then, it turns out that the geometry of optimal configurations under prescribed net charge $q_\mathrm{sat}$  is quite flexible: by way of explicit constructions (see, e.g., Fig.~\ref{fig : subsetdiamond} and Fig.~\ref{fig : construction qsat}), we observe that optimal configurations are possibly not connected or regions at the boundary are not contained in the square lattice. Fine geometric arguments are necessary to ensure that these degenerate parts consist of a controlled number of atoms only.  For the  identification of the global diamond-like Wulff-shape,  we identify the charge of configurations with  the $\infty$-perimeter of specific  interpolations, and then, following an idea inspired by \cite{cicalese}, we apply the quantitative isoperimetric inequality to obtain a bound on the deviations from a diamond.

%
%
%

The article is organized as follows. In Section \ref{section:settings} we introduce the precise mathematical setting and present the main results about the free and prescribed net charge problem.  In Section \ref{section:upper bound} we construct explicitly some configurations in order to  provide sharp upper bounds for the ground-state energy and the  net charge. Moreover, we establish an upper bound for the saturation net charge $q_{\rm sat}$. These explicit constructions already give the right intuition for  the microscopic and macroscopic  geometry  of optimal configurations.  In Section \ref{sec : elementary}  we discuss elementary geometric properties of energy minimizers. In Section \ref{sec: ground states} we give the lower bound for the ground-state energy and provide a fine characterization of the geometry of ground states. Here, we also prove the $n^{1/4}$-law for the net charge of ground states. Finally,  Section \ref{sec: prescribed} is devoted to the prescribed net charge problem. We first provide a lower bound for $q_{\rm sat}$ matching the upper bound  derived \EEE in Section \ref{section:upper bound}.  Afterwards,  we  characterize \EEE the geometry of $q_{\rm sat}$-optimal configurations. \EEE

\section{Setting and main results}\label{section:settings}

 In this section we first introduce our model and give some basic definitions. Afterwards, we present our main
results.

\subsection{Configurations and interaction energy} We consider particle systems in two dimensions consisting of two different atomic types. We model their interaction by classical potentials in the frame of Molecular Mechanics \cite{Molecular,Friesecke-Theil15,Lewars}.
Let $n \in \mathbb{N}$. We indicate the \textit{configuration} of $n$ particles by
\begin{align*}
C_n=\{(x_1,q_1),\ldots,(x_{n},q_{n})\} \subset \left(\mathbb{R}^{2} \times \{-1,1\}\right)^n,
\end{align*}
identified with the respective \textit{atomic positions} $X_n=(x_1,\ldots,x_n) \in \mathbb{R}^{2n}$ together with their types $ Q_n  =(q_1,\ldots,q_n) \in \{-1, 1\}^n$. By referring to a model for ionic dimers, we will often call $Q_n$ the \textit{charges} of the atoms, $q=1$ representing \textit{cations} and $q=-1$ representing \textit{anions}. Our choice of the empirical potentials (see below) is indeed inspired by ions in ionic compounds, which are primarily held together by their electrostatic forces between the net negative charge of the anions and the net positive charge of the cations \cite{Pauling}.

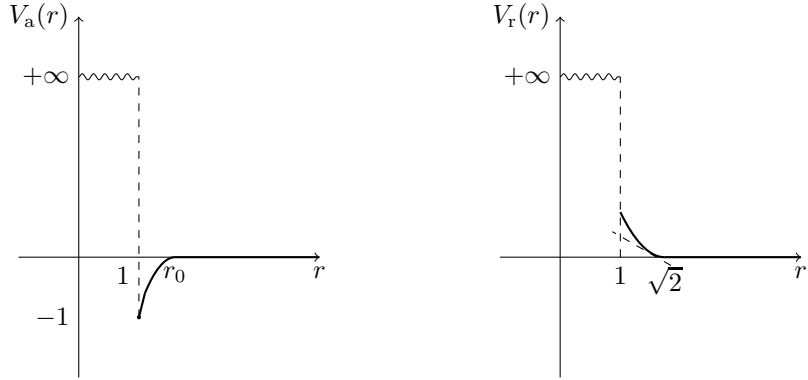
\begin{figure}[htp]
\centering
\begin{tikzpicture}[scale=0.8]
\draw[->](0,-2)--++(0,6) node[anchor= east] {$V_{\mathrm{a}}(r)$};
\draw[->](-1,0)--++(5,0) node[anchor =north] {$r$};
\draw[decorate, decoration={snake,amplitude=.4mm,segment length=1.5mm}] (0,3) node[anchor =east]{$+\infty$}--++(1,0);
\draw[dashed,thin](1,-1) --++(0,1) node[anchor =north east] {$ 1 $}--++(0,3);
\draw[fill=black](0,-1) node[anchor = east] {$-1$}++(1,0) circle(.025);
\draw[thick](1.6,0)--++(2.4,0);
\draw[thick](1,-1)--(1.1,-0.6);
\draw[thick](1.1,-0.6) parabola[bend at end] (1.6,0) node[anchor=north]{$r_0$};

\begin{scope}[shift={(8,0)}]
\draw[->](0,-2)--++(0,6) node[anchor= east] {$V_{\mathrm{r}}(r)$};
\draw[->](-1,0)--++(5,0) node[anchor =north] {$r$};
\draw[decorate, decoration={snake,amplitude=.4mm,segment length=1.5mm}] (0,3) node[anchor =east]{$+\infty$}--++(1,0);
\draw[dashed,thin](1,0) node[anchor =north] {$ 1 $}--++(0,3);
\draw[thick](1,.75) parabola[bend at end] ({sqrt(3)},0) node[anchor=north]{$\sqrt{2}$};
\draw[thick]({sqrt(3)},0)--(4,0);
\draw[dashed,thin]({sqrt(3)},-.075)++(150:-.125)--++(150:1.125);
\end{scope}

\end{tikzpicture}
\caption{The potentials $V_{\mathrm{a}}$ and $V_{\mathrm{r}}$.}
\label{FigurePotentials}
\end{figure}

 By following the setting in \cite{FriedrichKreutzHexagonal}, we  define the energy $\mathcal{E}:  (\mathbb{R}^2 \times \lbrace -1,1 \rbrace)^{n}$ $ \to \overline{\mathbb{R}}$ of a given configuration $\lbrace (x_1,q_1),\ldots,(x_n,q_n) \rbrace  \in (\mathbb{R}^2 \times \lbrace -1,1 \rbrace)^{n}$   by
\begin{align}\label{Energy}
\mathcal{E}(C_n) = \frac{1}{2}\sum_{\underset{q_i = q_j}{i \neq j}} V_{\mathrm{r}}(|x_i-x_j|) + \frac{1}{2}\sum_{\underset{q_i \neq q_j}{i \neq j}} V_{\mathrm{a}}(|x_i-x_j|),
\end{align}
where $V_{\mathrm{r}}, V_{\mathrm{a}} : [0,+\infty)\to \overline{\mathbb{R}}$ are a \emph{repulsive potential} and an \emph{attractive-repulsive potential}, respectively. The factor $1/2$  accounts for the fact that every contribution is counted twice in the sum.  The two potentials are pictured schematically in Fig.~\ref{FigurePotentials}.  Let $r_0 \in [1,(2\sin(\frac{\pi}{7}))^{-1})$  and note that $r_0 < \sqrt{2}$.  The  attractive-repulsive potential $V_{\mathrm{a}}$ satisfies
\begin{align*}
{\rm [i]}& \ \ V_{\mathrm{a}}(r) =+\infty \text{ for all $r < 1$},\\
{\rm [ii]}& \ \ V_{\mathrm{a}}(r)=-1 \text{ if and only if $r=1$ and $V_{\mathrm{a}}(r) >-1$ otherwise},\\
{\rm [iii]}& \ \  V_{\mathrm{a}}(r) \leq 0 \text{ for all } r\geq 1   \text{ with equality for all $r > r_0$.}
\end{align*}
The distance $r =1 $ represents the (unique) equilibrium distance of two atoms with opposite charge. The choice of $V_{\rm a}$ reflects a balance between a long-ranged Coulomb attraction and  the  short-ranged  Pauli repulsion acting  when a pair of ions comes too close to each other. Assumption [iii] restricts the interaction range and guarantees that the \emph{bond graph} is planar, see Section \ref{sec: notions}.

The repulsive potential $V_{\mathrm{r}}$ satisfies
\begin{align*}
{\rm [iv]} & \ \  \text{$V_{\mathrm{r}}(r) = +\infty$ for all $r<1$ and $0 \leq V_{\mathrm{r}}(r) <+\infty$ for all $r\geq 1$},\\
{\rm [v]} & \ \  \text{$V_{\mathrm{r}}$ is non-increasing and convex for $r \geq 1$,}\\
{\rm [vi]} & \ \   \text{$V_{\mathrm{r}}\left(2r_0 \sin\left(\frac{\pi}{5} \right) \right) > 6$,}  \\
{\rm [vii]} & \ \  \text{$V_{\mathrm{r}}(r) = 0 $  if and only if  $r \geq \sqrt{2}$.}
\end{align*}
 The natural assumption [v] is satisfied for example for repulsive Coulomb interactions. We emphasize that some quantitative requirements of the form [vi] and [vii] are necessary to obtain a crystallization
result in the square lattice. Other quantitative assumptions on the repulsive potential will favor, e.g., that the atoms arrange themselves in a hexagonal lattice, see \cite{FriedrichKreutzHexagonal}.

Finally, we require the following \emph{slope conditions} \EEE
\begin{itemize}
\item[{\rm[viii]} ] $ \displaystyle
V'_{\mathrm{r},-}(\sqrt{2}) < -\frac{16}{\sqrt{2}\pi},\EEE  \ \ \ \ \ \ \    \frac{1}{r-1}(V_{\mathrm{a}}(r)- V_{\mathrm{a}}(1)) > -\frac{1}{2} V'_{\mathrm{r},+}(1)   \ \ \ \text{for all } r \in (1,r_0],
$
\end{itemize}
 where the functions $V'_{\mathrm{r},-}$, $V'_{\mathrm{r},+}$ denote the left and right derivative, respectively. (They exist
due to  the  convexity of $V_{\rm r}$.) These conditions are reminiscent of the soft-interaction assumption by
{\sc Radin} \cite{Radin} and the slope condition for an angular potential by  {\sc Mainini, Piovano, and Stefanelli} \cite{Mainini-Piovano}.  Assumption [viii] is only needed in Lemma \ref{LemmaBoundaryEnergy} where the energy contribution of atoms at the boundary of the configuration  is  estimated. For a more detailed discussion on the assumptions on $V_{\rm a}$ and $V_{\rm r}$ we refer the reader to \cite[Section 2.1]{FriedrichKreutzHexagonal}.

\subsection{Basic notions}\label{sec: notions} In this  subsection  we collect some basic notions. Consider a configuration $C_n \in \left(\mathbb{R}^2 \times \{-1,1\}\right)^{n}$ with finite energy consisting of the positions $X_n= (x_1,\ldots,x_n) \in \mathbb{R}^{2n}$ and the charges $Q_n=(q_1,\ldots,q_n)\in \{-1,1\}^n$.  

\textbf{Neighborhood, bonds, angles:}  For each $x_i\in \mathbb{R}^2$, $i \in \{1,\ldots,n\}$, we define   the \textit{neighborhood} by
\begin{align}\label{eq: neighborhood}
\mathcal{N}(x_i) = \left(X_n \setminus \{x_i\}\right) \cap \lbrace x \in \mathbb{R}^2: \ |x-x_i| \le r_0 \rbrace,
\end{align}
 where $r_0$ is defined in [iii].   If $x_j \in \mathcal{N}(x_i)$, we say that $x_i$ and $x_j$ are \emph{bonded}. We will say that $x_i$ is $k$-bonded if $\# \mathcal{N}(x_i) = k$. Given $x_j,x_k \in \mathcal{N}(x_i)$, we define the \textit{bond-angle} between $x_j,x_i,x_k$ as the angle between the two vectors $x_k-x_i$ and $x_j-x_i$. (We choose anti-clockwise orientation for
definiteness.) In general, we say that it is an angle at $x_i$.

\EEE

\textbf{Bond graph:}  The set of   atomic positions $X_n \subset \mathbb{R}^{2n}$ together with the set of  bonds $\{\{x_i,x_j\}: x_j \in \mathcal{N}(x_i)\}$   forms a graph which we call the  \textit{bond graph}.  Since for  configurations with finite energy  there holds $\mathrm{dist}(x_i,X_n \setminus \{x_i\})\geq 1$ and $x_j \in \mathcal{N}(x_i)$ only if $|x_i-x_j| \leq r_0<\sqrt{2}$,     their bond   graph is a planar.  Indeed, given a quadrangle with all sides and one diagonal in $[1,r_0]$, the second diagonal is at least $\sqrt{2} >r_0$.  If no ambiguity arises, the number of bonds in the bond graph will be denoted by $b$, i.e.,
$$b= \# \{\{x_i,x_j\}: x_j \in \mathcal{N}(x_i)\}.$$ \EEE 
We say a configuration is \emph{connected} if each two atoms are joinable through a simple path in the bond graph. In a similar fashion, we speak of connected components of a configuration.  Any simple cycle of the bond graph is a \textit{polygon}.

 \textbf{Acyclic bonds:} A bond is called \textit{acyclic} if it is not contained in any simple cycle of the bond graph. Among acyclic bonds we distinguish between \textit{flags} and \textit{bridges}. We say that an acyclic bond is a bridge if it is contained in some simple path connecting two vertices which are included in two distinct cycles.  All other acyclic bonds are  called flags, see Fig.~\ref{FigureFlagsBridges}.

 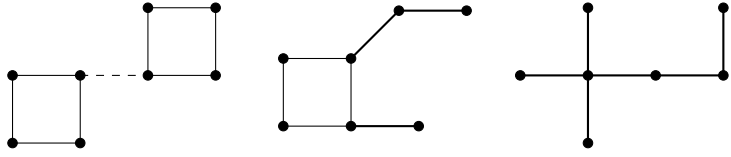
\begin{figure}[htp]
\centering
\begin{tikzpicture}[scale=0.9]
\foreach \j in {0,1,2,3}{
\draw(0,0)++(90*\j:.5)++(90*\j+90:-.5)--++(90*\j+90:1);
\draw[fill=black](0,0)++(90*\j:.5)++(90*\j+90:-.5) circle(.07);
\draw[fill=black](2,1)++(90*\j:.5)++(90*\j+90:-.5) circle(.07);

\draw(2,1)++(90*\j:.5)++(90*\j+90:-.5)--++(90*\j+90:1);
\begin{scope}[shift={(4,.25)}]
\draw(0,0)++(90*\j:.5)++(90*\j+90:-.5)--++(90*\j+90:1);
\draw[fill=black](0,0)++(90*\j:.5)++(90*\j+90:-.5) circle(.07);
\end{scope}

}

\draw[thick](4.5,-.25)--++(1,0);

\draw[thick](4.5,.75)--++(45:1)--++(0:1);
\draw[fill=black](4.5,.75)++(45:1)++(0:1)circle(.07);
\draw[fill=black](4.5,.75)++(45:1)circle(.07);
\draw[fill=black](4.5,-.25)++(1,0)circle(.07);
\draw[dashed](.5,.5)--++(1,0);

\draw[thick](7,.5)--++(3,0)++(0,1)--++(0,-1)++(-2,0)++(0,1)--++(0,-2);
\draw[fill=black](7,.5)circle(.07);
\draw[fill=black](7,.5)++(1,0)circle(.07);
\draw[fill=black](7,.5)++(2,0)circle(.07);
\draw[fill=black](7,.5)++(3,0)circle(.07);
\draw[fill=black](7,.5)++(1,1)circle(.07);
\draw[fill=black](7,.5)++(1,-1)circle(.07);
\draw[fill=black](7,0.5)++(3,1)circle(.07);
\end{tikzpicture}
\caption{Examples of flags (bold) and a bridge (dashed).}
\label{FigureFlagsBridges}
\end{figure}

 \textbf{Defects:}  By elementary polygons we denote  polygons which do \EEE not contain any non-acyclic bonds in its interior region.  An elementary polygon \EEE  in the bond graph which is not a  square is called  a \emph{defect}.  We  introduce the  \textit{excess of edges}   $\eta(C_n)$  by
\begin{align}\label{Excess}
\eta(C_n)= \sum_{j\geq 4} (j-4)   f_j,  
\end{align}
 where $f_j$ denotes the number of  elementary  polygons with $j$ vertices in the bond graph. The excess of edges is a tool to quantify the number of {defects} in the bond graph. \EEE { Note that the summation in (\ref{Excess}) runs over $j\geq 4$. This is due to the fact that we use this definition only for configurations whose bond graph contains only $k$-gons with $k \geq 4$,   cf.~Lemma \ref{RemarkPolygon}  below.  If it is clear from the context, we omit the dependence on $C_n$ and write $\eta=\eta(C_n)$.  
 
 In the following, we  frequently  refer to $C_n$ instead of $X_n$ when speaking about its bond graph, acyclic bonds, or connectedness properties.
 
\textbf{Charges:} We say that a configuration satisfying
\begin{align}\label{SamechargeNeighbourhood}
 \mathcal{N}(x_i) \cap \left\{x_j \in X_n: q_j=q_i\right\} = \emptyset  \text{ for all } i \in \{1,\ldots,n\}
\end{align}
 has \textit{alternating charge distribution}. A configuration is called  \emph{repulsion-free}  if $|x_i - x_j| \ge \sqrt{2}$ for all $x_i \neq x_j$ with $q_i = q_j$. The \emph{net charge} of a configuration is defined  as the (signed) difference of the number of the two atomic types, i.e.,   
 \begin{align} \label{DefinitionNetCharge}
  \mathcal{Q}(C_n) := \sum_{i=1}^n q_i. 
 \end{align} 
We note that all possible net charges are given by $\mathcal{Q}_{\rm net}(n) := (2\mathbb{Z} + n\,{\rm mod} 2) \cap [-n,n]$. If $\mathcal{Q}(C_n)>0$, we say that the atoms with charge $+1$ represent the \emph{majority phase} and the atoms with charge  $-1$  the \emph{minority phase}. We use a corresponding denomination if  $\mathcal{Q}(C_n)<0$.  We denote by 
\begin{align}\label{def : Cnpm}
X_n^+:=\{x_i \in X_n : q_i= + 1\}, \ \ \ \ \ \ \ X_n^-:=\{x_i \in X_n : q_i= - 1\}
\end{align}
the \textit{positively} and \textit{negatively} \textit{charged phase}, respectively.

\textbf{Ground states:} 
For $q_{\rm net} \in \mathcal{Q}_{\rm net}(n)$ we define 
\begin{align}\label{def : mnq}
\mathcal{E}^n_{\rm min}(q_{\rm net}):= \min\left\{\mathcal{E}(C_n): \,  C_n \subset ( \mathbb{R}^{2} \times \{ -1, 1\})^{n}, \  \mathcal{Q}(C_n) = q_{\rm net}\right\}.
\end{align}
A configuration $C_n$ with  $\mathcal{Q}(C_n)=q_{\rm net}$ and $\mathcal{E}(C_n) = \mathcal{E}^n_{\rm min}(q_{\rm net})$ is called a \emph{$q_{\rm net}$-optimal configuration}.  We will often simply call $C_n$ an \emph{optimal configuration}. Moreover, we call  $C_n$  a \textit{ground state} if  $\mathcal{E}(C_n) = \min_{q_{\rm net} \in \mathcal{Q}_{\rm net}(n)} \mathcal{E}^n_{\rm min}(q_{\rm net})$. In other words, $C_n$ is a ground state if and only if
\begin{align*}
\mathcal{E}(C_n) \leq \mathcal{E}(C_n')
\end{align*}
for all $C'_n \subset ( \mathbb{R}^{2} \times \{ -1, 1\})^{n}$.

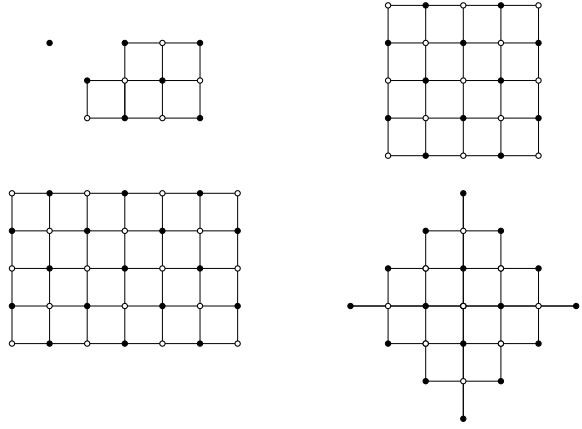
\begin{figure}[h]
\centering
\begin{tikzpicture}[scale=0.5]

\begin{scope}[rotate=90]

\draw[step=1cm,very thin] (1,0) grid (5,6);
\foreach \j in {0,...,2}{
\foreach \i in {1,...,2}{
\draw[fill=black](2*\i,2*\j) circle(.07);
\draw[fill=black](2*\i+1,2*\j+1) circle(.07);
\draw[fill=white](2*\i+1,2*\j) circle(.07);
\draw[fill=white](2*\i,2*\j+1) circle(.07);
}
}

\foreach \j in {0,...,2}{
\draw[fill=black](1,2*\j+1) circle(.07);
}

\foreach \j in {0,...,1}{
\draw[fill=black](2*\j+2,6) circle(.07);
\draw[fill=white](2*\j+3,6) circle(.07);
}

\foreach \j in {0,...,3}{
\draw[fill=white](1,2*\j) circle(.07);
}

\draw[step=1cm,very thin] (7,1) grid (9,3);
\draw[step=1cm,very thin] (7,3) grid (8,4);
\foreach \j in {0,...,1}{
\foreach \i in {0,...,1}{
\draw[fill=black](7,1)++(2*\i,2*\j) circle(.07);
}}

\foreach \j in {0,...,1}{
\draw[fill=white](7,2)++(0,2*\j) circle(.07);
\draw[fill=white](8,1)++(0,2*\j) circle(.07);
\draw[fill=black](8,2)++(0,2*\j) circle(.07);
}

\draw[fill=white](9,2) circle(.07);

\draw[fill=black](9,5) circle(.07);
\begin{scope}[shift={(2,-6)}]

\foreach \j in {0,...,2}{
\draw[ultra thin](\j,-3+\j)--++(0,6-2*\j);
\draw[ultra thin](-\j,-3+\j)--++(0,6-2*\j);
\draw[ultra thin](\j-3,\j)--++(6-2*\j,0);
\draw[ultra thin](\j-3,-\j)--++(6-2*\j,0);
}

\foreach \j in {1,3}{
\pgfmathsetmacro\y{\j-1}
\foreach \l in {0,...,\y}{
\foreach \k in {0,...,3}{
\begin{scope}[rotate = \k*90]
\draw[fill=black](0,0)++(\j,0)++(-\l,\l) circle(.07);
\end{scope}
}
}}

\foreach \j in {0,2}{
\pgfmathsetmacro\y{\j-1}
\foreach \l in {0,...,\y}{
\foreach \k in {0,...,3}{
\begin{scope}[rotate = \k*90]
\draw[fill=white](0,0)++(\j,0)++(-\l,\l) circle(.07);
\end{scope}
}
}}
\end{scope}

\begin{scope}[shift={(6,-8)}]
\clip(-.5,-.5) rectangle (4.5,4.5);
\draw[ultra thin] (0,0) grid (4,4);
\foreach \j in {0,1,2}{
\foreach \i in {0,1,2}{
\draw[fill=white](0,0)++(2*\j,2*\i) circle(.07);
\draw[fill=black](1,0)++(2*\j,2*\i) circle(.07);
\draw[fill=white](1,1)++(2*\j,2*\i) circle(.07);
\draw[fill=black](0,1)++(2*\j,2*\i) circle(.07);
}}

\end{scope}
\end{scope}
\end{tikzpicture}
\caption{Some configurations that are subset of the square lattice. The top configuration on the right  is $S_n$, $25 \le n\le 35$, whereas the bottom configuration on the right is $D_n$, $25\le n \le 40$. The configurations have alternating charge distribution where black indicates $q= + 1$ and white indicates $q=-1$.}
\label{Fig:Squarelattice}
\end{figure}

 \textbf{Subsets of the square lattice:}   We denote  by \EEE $\mathbb{Z}^2\subset \mathbb{R}^2$  the \emph{square lattice}.  We define special subsets of the square lattice representing the Wulff-shapes of optimal configurations. \EEE  For $n \in \mathbb{N}$   we let
 \begin{align}\label{def : Sn}
 S_n := \left\{x \in \mathbb{Z}^2 : x_1,x_2 \geq 0, \   |x|_\infty  \leq  \lfloor\sqrt{n}-1  \rfloor \right\}
 \end{align}
 the \textit{square} of sidelength  $\lfloor \sqrt{n}-1\rfloor$,  where by $\lfloor t\rfloor$ we denote the integer part of $t\in \mathbb{R}$.    By
 \begin{align} \label{def :Dn}
 D_n := \left\{x \in \mathbb{Z}^2 :  |x|_1  \leq \lfloor (-1+ \sqrt{2n-1})/2 \rfloor\right\} 
 \end{align}
 we denote the \textit{diamond} of radius $\lfloor (-1+ \sqrt{2n-1})/2 \rfloor$.  $S_n$ and $D_n$ represent the largest squares and diamonds, respectively,  whose number of atoms is  less or equal to $n$.     In Fig.~\ref{Fig:Squarelattice} some subsets of the square lattice are depicted. 

Observe that these  configurations can be chosen to have   alternating charge distribution.   By assumptions [ii], [iii], and [vii], the energy of such configurations $C_n$  satisfies        $\mathcal{E}(C_n) = -b$ since all atoms of the same charge have at least distance $\sqrt{2}$ and all atoms of $X_n$ are bonded only to atoms of opposite charge   with bonds of unit length.

\textbf{Equilibrated atoms:} We say  that an atom $x \in X_n$ is \emph{equilibrated} if all bond-angles at $x$ lie in $\lbrace \frac{\pi}{2},\pi,\frac{3\pi}{2} \rbrace$. By  $\mathcal{A}(X_n)$ \EEE we denote the atoms which are \emph{not} equilibrated.  Note that, if $\mathcal{A}(X_n)=\emptyset$ and $X_n$ is connected, then $X_n$ is a subset of the square lattice  $\mathbb{Z}^2$.

\subsection{Main results}\label{subsection:main results}  In this  subsection \EEE we state our main results.  We will first address the free net charge problem, and characterize the energy and geometry of ground states. In particular, we will  prove a rigorous planar crystallization result in the spirit of \cite{FriedrichKreutzHexagonal, HR,Mainini-Piovano,Mainini,Radin}  and  the emergence of a square macroscopic Wulff shape (cf.\ \cite{Yuen, Davoli15, Davoli16, Mainini-Piovano, Schmidt-disk}).  Then we will characterize the net charge of ground-state configurations.

Afterwards, we   change the perspective and   study $q_{\rm net}$-optimal configurations under prescribed net charge $q_{\rm net}$. We  give an estimate for the   minimal energy $\mathcal{E}^n_{\rm min}(q_{\rm net})$  and   identify a specific net charge, the \emph{saturation net charge} $q^n_{\rm sat}$, which corresponds to the smallest net  charge where all atoms of the minority phase are $4$-bonded. Finally, we  prove the emergence of a diamond-like Wulff-shape for $q^n_{\rm sat}$-optimal configurations which  reflects  the sensitivity of the Wulff-shape on  the prescribed net  charge. 

 \noindent \textbf{Free net charge:}   Our first result characterizes the energy of ground states.  For $n \in \mathbb{N}$, we introduce the function
\begin{align}\label{def: beta}
\beta(n):= 2n -2\sqrt{n}.
\end{align}

\begin{theorem}[Ground-state energy]\label{TheoremGroundstates}  Let $n \in \mathbb{N}$. Ground states are connected and have alternating charge distribution. They do not contain any bridges. There holds 
\begin{align} \label{Energygroundstates}
\mathcal{E}(X_n)=-b= -\lfloor\beta(n)\rfloor.
\end{align}
\end{theorem}

\begin{remark}\label{rem: repulsionsfree}
{\normalfont 
In view of assumptions ${\rm [ii]}$ and [vii], we have that $\mathcal{E}(C_n) \geq -b$ with equality if and only if the configuration is repulsion-free and all bonds have unit length. In particular, Theorem \ref{TheoremGroundstates}   implies that ground states satisfy both properties.
}
\end{remark}

The next result states that ground states are essentially subsets of the square lattice and that a square Wulff-shape emerges as $n \to \infty$. Without further notice, all following statements regarding the geometry of ground states hold up to
isometry.  Recall the definition of a square of sidelength   $\lfloor \sqrt{n}-1\rfloor$   in  (\ref{def : Sn}).

\EEE

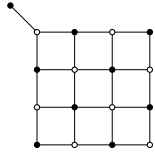
\begin{figure}[h]
\centering
\begin{tikzpicture}[scale=0.5]
\draw[very thin](0,3)--++(135:1);
\draw[fill=black](0,3)++(135:1) circle(.07);
\draw[step=1cm,very thin] (0,0) grid (3,3);
\foreach \j in {0,1}{
\foreach \i in {0,1}{
\draw[fill=black] (2*\j,2*\i) circle(.07);
\draw[fill=black] (1+2*\j,1+2*\i) circle(.07);
\draw[fill=white] (1+2*\j,2*\i) circle(.07);
\draw[fill=white] (2*\j,1+2*\i) circle(.07);
}
}
\end{tikzpicture}
\caption{A ground state for  $n=17$  that is not a subset of the square lattice.}
\label{Fig:flexible}
\end{figure}

\begin{theorem}[Characterization of ground states]\label{theorem : GeometryGroundstatewithout}  Let $n \in \mathbb{N}$ and let $C_n$ be a ground state.

\begin{itemize}
 \item[(a)]   (Crystallization)  Except for possibly one  atom,  $C_n$ is a   subset of the square lattice.
 \item[(b)]   (Wulff-shape) For a universal constant $c>0$ independent of $n$, there holds
$$\min_{a \in \mathbb{R}^2}    \# \big(X_n \triangle (a+ S_n) \big) \le cn^{3/4},$$ 
where $\triangle$ denotes the symmetric difference of sets. 
\end{itemize}

\end{theorem}

  We point out that, if a ground state contains a flag,  it is possibly not a subset of the square lattice,  see Fig.~\ref{Fig:flexible}. Our next result addresses the net charge \eqref{DefinitionNetCharge} of ground-state configurations.

\begin{theorem}[Net charge of ground states]\label{theorem : charge} 
The following properties for the net charge hold:
\begin{itemize}
 \item[(a)]  (Net charge control) \EEE There  exists \EEE a universal constant $c>0$ such that for all $n \in \mathbb{N}$ and all ground states $C_n$ the net \EEE charge satisfies $|\mathcal{Q}(C_n)| \le cn^{1/4}$.

\item[(b)]  (Sharpness of the $n^{1/4}$-scaling) There exists  an increasing  sequence of integers $(n_j)_j$ and ground states $(C_{n_j})_j$ such that 
$$\liminf_{j \to +\infty} n_j^{-1/4}|\mathcal{Q}(C_{n_j})|>0. $$
\end{itemize}
\end{theorem}

 The sharp scaling $n^{1/4}$ for the upper bound of the net charge has also been identified in a related model where ground states are subsets of the hexagonal lattice, see \cite[Theorem 2.5]{FriedrichKreutzHexagonal}. The three theorems are proved in Section \ref{sec: ground states}. Explicit constructions for the upper bound of the  ground-state  energy and Theorem \ref{theorem : charge}(b)   are given in Subsection \ref{subsection:special subsets} and Subsection \ref{subsection:sqare rectangle}, respectively.


\noindent \textbf{Prescribed net charge:} We now change the perspective and suppose that the net charge $q_{\rm net}$ is prescribed.  We first characterize the energy of $q_{\rm net}$-optimal configurations $\mathcal{E}^n_{\rm min}(q_{\rm net})$, see \eqref{def : mnq}.  To this end, we introduce the \emph{saturation net charge} 
\begin{align}\label{eq: qsat-def}
q^n_{\rm sat} = \min\left\{q_\mathrm{net}: \,  q_\mathrm{net}\in \mathcal{Q}_\mathrm{net}(n)\cap [0,n], \  \mathcal{E}^n_\mathrm{min}(q_\mathrm{net})=-2n+2q_\mathrm{net} \right\}. 
\end{align}
The definition corresponds to the smallest net charge for which all atoms of the minority phase are $4$-bonded.  Configurations with  the latter  property will be called \emph{saturated} in the following.    The saturation net charge can be characterized as follows.
\begin{proposition}[Characterization of $q_\mathrm{sat}^n$]\label{theorem : qsat} 
There holds $q_\mathrm{sat}^n =\phi(n)$, where  $\phi(0) = 0$, $\phi(1)=1$, and 
\begin{align}\label{def: phi}
\phi(n) := \begin{cases} 
2\left\lfloor-\frac{1}{2}+\frac{1}{2}\sqrt{2n-5}\right\rfloor+3 &\text{if } n \text{ odd}, n \geq 3, \\
2\left\lfloor\frac{1}{2}\sqrt{2n-4}\right\rfloor+2 &\text{if } n \text{ even},n \geq 2.
\end{cases}
\end{align}
\end{proposition}
 We refer to  (\ref{eq: different repr}) for an equivalent characterization. The representation \eqref{def: phi} shows the  scaling $q_{\rm sat}^n \sim n^{1/2}$, which corresponds to the scaling of the number of boundary atoms of optimal configurations.

By $h^+ := \max(h, 0)$ we denote the positive part of a function $h$.  The energy of $q_{\rm net}$-optimal configurations can be characterized as follows.

\begin{theorem}[Energy of $q_{\rm net}$-optimal configurations]\label{theorem: min-en2}
 For all $n \in \mathbb{N}$ and all $q_{\rm net} \in Q_{\rm net}(n)$ there holds
\begin{align}\label{eq: charge energy}
-2n  + 2|q_{\rm net}|  \le  \mathcal{E}^n_{\rm min}(q_{\rm net})  \le    -2n  + 2|q_{\rm net}|   +  4 (q^n_{\rm sat} - |q_{\rm net}|)^+.
 \end{align}

 \end{theorem}
 We point out that the upper bound in \eqref{eq: charge energy} is consistent with Theorem \ref{TheoremGroundstates}, i.e., with $\mathcal{E}^n_{\rm min}(q_{\rm net}) \geq -\lfloor\beta(n)\rfloor$ for all $q_{\rm net}^n$. To see this, it suffices to note that   $2q_{\rm sat}^n \ge -\lfloor -2 \sqrt{n}\rfloor$.  The result states that for $|q_{\rm net}| \ge q^n_{\rm sat}$, the minimal energy is exactly   $-2n + 2|q_{\rm net}|$.  This corresponds to the case that  optimal configurations are saturated, i.e.,  each atom of the minority phase is $4$-bonded. In this sense, $q^n_{\rm sat}$ can be understood as a critical net charge.

\begin{figure}[htp]
 \begin{tikzpicture}[scale=.5]
 
     \draw[fill=black](9.2,2.1) circle(.07);
     
          \draw[fill=black](8.4,-1.8) circle(.07);

 \begin{scope}[rotate=5]
   \draw[ultra thin](7,0)--++(1,0);
   \draw[ultra thin](0,1)--++(-1,0);
 \draw[ultra thin](0,0) grid(7,1);
 \foreach \j in {0,1,2,3}{
 \draw[fill=black](2*\j,0) circle(.07);
  \draw[fill=white](2*\j,1) circle(.07);
   \draw[fill=white](2*\j+1,0) circle(.07);
    \draw[fill=black](2*\j+1,1) circle(.07);
     \draw[fill=black](2*\j+1,-1) circle(.07);
          \draw[fill=black](2*\j,2) circle(.07);
               \draw[ultra thin](2*\j+1,-1)--++(0,1);
          \draw[ultra thin](2*\j,2)--++(0,-1);
 }
  \draw[fill=black](7,0)++(1,0)circle(.07);
   \draw[fill=black](0,1)++(-1,0)circle(.07);
 \end{scope}
 
 \begin{scope}[shift={(12,0.025)}]
 \draw[ultra thin](-1,-1)--++(0,2);
  \draw[ultra thin](-2,0)--++(2,0);
    \draw[fill=white](-1,0) circle(.07);
 \foreach \j in {0,1,2,3}{
  \draw[fill=black](-1,0)++(90*\j:1) circle(.07);
 }
  \end{scope}
 
  \begin{scope}[shift={(13,.2)},rotate=10]
 \draw[ultra thin](0,-1)--++(0,2);
  \draw[ultra thin](-1,0)--++(2,0);
    \draw[fill=white](0,0) circle(.07);
 \foreach \j in {0,1,2,3}{
  \draw[fill=black](90*\j:1) circle(.07);
 }
 
 \end{scope}
 
   \begin{scope}[shift={(14.92,.7)},rotate=20]
 \draw[ultra thin](0,-1)--++(0,2);
  \draw[ultra thin](-1,0)--++(2,0);
    \draw[fill=white](0,0) circle(.07);
 \foreach \j in {0,1,2,3}{
  \draw[fill=black](90*\j:1) circle(.07);
 }
 
 \end{scope}
 
   \begin{scope}[shift={(23,1)},rotate=30]
 \draw[ultra thin](0,-1)--++(0,2);
  \draw[ultra thin](-1,0)--++(2,0);
    \draw[fill=white](0,0) circle(.07);
 \foreach \j in {0,1,2,3}{
  \draw[fill=black](90*\j:1) circle(.07);
 }
 
 \end{scope}
 
 \begin{scope}[shift={(18,-1.5)}]
\foreach \j in {0,...,3}{
\draw[ultra thin]({sqrt(2)},0)++(15:1)--++(15+90*\j:1);
\draw[ultra thin]({sqrt(2)},0)++(165:1)--++(75+90*\j:1);
\draw[ultra thin](0,0)++(60:{sqrt(2)})++(45:1)--++(45+90*\j:1);
}

\draw[fill=white]({sqrt(2)},0)++(15:1) circle(.07);
\draw[fill=black]({sqrt(2)},0)++(15:2) circle(.07);
\draw[fill=black]({sqrt(2)},0)++(-30:{sqrt(2)}) circle(.07);

\draw[fill=white]({sqrt(2)},0)++(165:1) circle(.07);
\draw[fill=black]({sqrt(2)},0)++(165:2) circle(.07);
\draw[fill=black]({sqrt(2)},0)++(30:{-sqrt(2)}) circle(.07);

\draw[fill=white](0,0)++(60:{sqrt(2)})++(45:1) circle(.07);
\draw[fill=black](0,0)++(60:{sqrt(2)})++(45:2) circle(.07);

\draw[fill=black]({sqrt(2)},0)++(60:{sqrt(2)})++(135:2) circle(.07);

\draw[fill=black]({sqrt(2)},0)++(60:{sqrt(2)}) circle(.07);

\draw[fill=black](0,0)++(60:{sqrt(2)}) circle(.07);
\draw[fill=black](0:{sqrt(2)}) circle(.07);

\end{scope}

 \end{tikzpicture}
 \caption{\emph{One} $q_\mathrm{net}$-optimal configuration for $q_\mathrm{net} > q_\mathrm{sat}^n$,  where $n=58$, $q_\mathrm{sat}^n = 12$, $q_\mathrm{net} = 43 - 15 = 28$. }
 \label{Fig: qnet bigger}
 \end{figure}
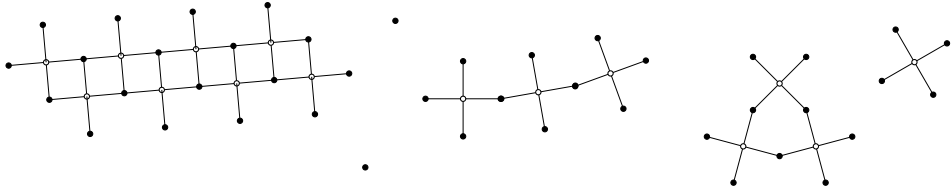

It turns out that $q^n_{\rm sat}$ is not only a critical net charge in terms   of \EEE the minimal energy $\mathcal{E}^n_{\rm min}(q_{\rm net})$, but also from a geometrical viewpoint: $q^n_{\rm sat}$ is critical in the sense that beyond $q^n_{\rm sat}$ no crystallization can be expected, cf.~Remark~\ref{remark : fragmentation}.  Note, however, that   $q_\mathrm{sat}^n$-optimal configurations crystallize,  cf.~Theorem  \ref{theorem : GeometryGroundstate}.

\begin{remark}[Fragmentation for $q_{\rm net} >q_\mathrm{sat}^n$] \label{remark : fragmentation} For $q_{\rm net} > q_\mathrm{sat}^n$, $q_{\rm net}$-optimal configurations may be completely degenerate, e.g., may consist of many connected components,  see Fig.~\ref{Fig: qnet bigger}.  Their characterizing property is that atoms of the minority phase are $4$-bonded with bond angles $\frac{\pi}{2}$, cf.\ Lemma \ref{lemma : Charge Energybound}(a)(i),(iv).  
\end{remark}

%
  Recall the definition of diamonds in  (\ref{def :Dn}).


 \begin{theorem}[Crystallization and Wulff-shape for $q^n_{\rm sat}$-optimal configurations]\label{theorem : GeometryGroundstate}  Let $n \in \mathbb{N}$ and let $C_n$ be a $q^n_{\rm sat}$-optimal configuration. 
 \begin{itemize}
 \item[(a)]  (Crystallization) There exists a universal constant  $n_0 \in \mathbb{N}$  independent of $n$ such that $C_n$ is a subset of the square lattice except for at most $n_0$ atoms. 
\item[(b)] (Wulff-shape) For a universal constant $c>0$ independent of $n$, there holds
$$\min_{a \in \mathbb{R}^2}    \# \big(X_n \triangle (a+ D_n) \big) \le cn^{3/4}.$$ 

\end{itemize}

\end{theorem}

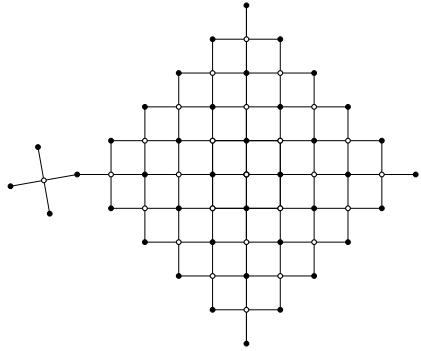
\begin{figure}[htp]
\centering
\begin{tikzpicture}[scale=0.45]
\begin{scope}[rotate=90]
\draw[ultra thin](0,5)--++(100:2);
\draw[ultra thin](0,5)++(100:1)++(190:1)--++(10:2);

\foreach \j in {0,...,2}{
\draw[ultra thin](\j,-2+\j)--++(0,4-2*\j);
\draw[ultra thin](-\j,-2+\j)--++(0,4-2*\j);
\draw[ultra thin](\j-2,\j)--++(4-2*\j,0);
\draw[ultra thin](\j-2,-\j)--++(4-2*\j,0);
\draw[fill=black](0,5)++(100:1)++(10+90*\j:1)circle(.07);
}

\draw[fill=white](0,5)++(100:1) circle(.07);

\foreach \j in {0,...,5}{
\draw[ultra thin](\j,-5+\j)--++(0,10-2*\j);
\draw[ultra thin](-\j,-5+\j)--++(0,10-2*\j);
\draw[ultra thin](\j-5,\j)--++(10-2*\j,0);
\draw[ultra thin](\j-5,-\j)--++(10-2*\j,0);
}

\foreach \j in {1,3,5}{
\pgfmathsetmacro\y{\j-1}
\foreach \l in {0,...,\y}{
\foreach \k in {0,...,3}{
\begin{scope}[rotate = \k*90]
\draw[fill=black](0,0)++(\j,0)++(-\l,\l) circle(.07);
\end{scope}
}
}}

\foreach \j in {0,2,4}{
\pgfmathsetmacro\y{\j-1}
\foreach \l in {0,...,\y}{
\foreach \k in {0,...,3}{
\begin{scope}[rotate = \k*90]
\draw[fill=white](0,0)++(\j,0)++(-\l,\l) circle(.07);
\end{scope}
}
}}
\end{scope}

\end{tikzpicture}
\caption{A $q_\mathrm{sat}^n$-optimal configuration for which four  atoms are not subset of the square lattice.  (Note that $q_\mathrm{sat}^n=13$ for $n=65$.)}
\label{fig : subsetdiamond}
\end{figure}
 
\emph{Complete} \EEE crystallization cannot be expected for certain values of $n$, as shown for example in Fig.~\ref{fig : subsetdiamond}.  Comparing this result to the geometry of ground states identified in Theorem \ref{theorem : GeometryGroundstatewithout}, we observe that the geometry of the Wulff-shape   and therefore the global geometry of  optimal \EEE configurations  is very sensitive to the prescribed net charge.  Proposition \ref{theorem : qsat}, Theorem \ref{theorem: min-en2}, and Theorem \ref{theorem : GeometryGroundstate} are proved in Section \ref{sec: prescribed}. The upper bound for $q_{\rm sat}^n$ is constructed explicitly in  Subsection \EEE \ref{subsection : upper bound qsat}.

\section{Constructions  of special subsets of the square lattice}\label{section:upper bound}
 This section is devoted to explicit constructions of sub-configurations of the square lattice  with alternating
charge distribution.   In Subsection \ref{subsection:special subsets} we exhibit candidates for the ground-state energy which will already give the upper bound in \eqref{Energygroundstates}.    In Subsection \ref{subsection:sqare rectangle} we construct ground-state configurations with  net charge  of order $n^{1/4}$, which will establish Theorem \ref{theorem : charge}(b). Finally,  in Subsection \ref{subsection : upper bound qsat} we define configurations  with net  charge $\phi(n)$, see \eqref{def: phi},    for which the atoms of the minority phase are $4$-bonded. This yields    the upper bound for $q_\mathrm{sat}^n$ in Proposition \ref{theorem : qsat}. We defer the lower bound on the ground-state energy and  the upper bound on the net charge to Section \ref{sec: ground states}. The lower bound for $q_\mathrm{sat}^n$ is addressed in  Subsection \ref{subsection : calculation qsat}.

 \subsection{Upper bound on the ground-state energy} \label{subsection:special subsets}
This subsection is devoted to  an  explicit construction of  configurations $C_n$  which \EEE maximize   the number of bonds and that are subsets of the square lattice. These configurations provide a reference energy value for every $n$, namely $\mathcal{E}(C_n)=-\lfloor\beta(n)\rfloor$,  see \eqref{def: beta}. 

By the special geometry of the square lattice, it is quite natural to give an interpretation of the two terms appearing in $\beta$. The leading order term of the energy is given by $-2n$ which corresponds to the bulk part of the energy. Its value is due to the fact that every interior atom  is bonded  to  four other  atoms of  opposite  charge. Furthermore, the repulsive term in the energy is zero for such configurations since the distance of two atoms with the same charge is bigger than or equal  to  $\sqrt{2}$. The additional lower order correction term is due to the fact that  atoms \EEE on the boundary of the ground-state  configuration \EEE do not have four  neighbors. 

The construction follows \cite{Mainini-Piovano} and is illustrated in Fig.~\ref{FigureDaisy}.   If  $n=k^2$, $k\in \mathbb{N}$,   we arrange the atoms on the lattice points of the square $S_{n}$ (cf.\ (\ref{def : Sn})). Then for $n=k^2+m+1$  with $0\leq m\leq 2k-1$  we proceed as follows:  for $ 0\leq m \leq k-1$ we recursively construct $X_{k^2+m+1}$ by adding the point with coordinates $(m,k)$ to   $X_{k^2+m}$. For $k\leq m \leq 2k-1$ we construct $X_{k^2+m+1}$ by adding the point with coordinates  $(k,m-k)$  to  $X_{k^2+m}$.  Since  the bond graphs only contain cycles of even length, we can choose corresponding charges such that the resulting configurations $C_n$ have alternating charge distribution. One can check that $\mathcal{Q}(C_n) \in \{-1,0,1\}$ for all $n \in \mathbb{N}$.

\begin{figure}[h]
\centering
\begin{tikzpicture}[scale=0.5]

\foreach \j in {1,...,4}{
\pgfmathtruncatemacro\result{\j+3}
\pgfmathtruncatemacro\resulti{\j-1}
\draw[dotted](\j-1,3)--++(0,1) node[anchor=south]{$\resulti$};
\draw[dotted](3,\j-1)--++(1,0) node[anchor=west]{$\result$};
}

 \foreach \j in {1,...,3}{
\draw[dotted](\j-1,4)--++(1,0);
\draw[dotted](4,\j-1)--++(0,1);
}

\draw[step=1cm,very thin] (0,0) grid (3,3);
\foreach \j in {0,...,1}{
\foreach \i in {0,...,1}{
\draw[fill=black](2*\i,2*\j) circle(.07);
\draw[fill=black](2*\i+1,2*\j+1) circle(.07);
\draw[fill=white](2*\i+1,2*\j) circle(.07);
\draw[fill=white](2*\i,2*\j+1) circle(.07);
}

\draw[fill=black](2*\j,4) circle(.07);
\draw[fill=black](4,2*\j) circle(.07);
\draw[fill=white](4,2*\j+1) circle(.07);
\draw[fill=white](2*\j+1,4) circle(.07);
}

\end{tikzpicture}
\caption{The construction for    $n=4^2 + m+1$ with $0 \le m \le 7$.   }
\label{FigureDaisy}
\end{figure}
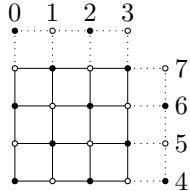

\begin{proposition}\label{PropositionDaisy} For all $n \in \mathbb{N}$,  let $C_n$ be the configuration introduced above. Then  there holds
\begin{align*}
\mathcal{E}(C_n) = -\lfloor \beta(n)\rfloor =  -\lfloor 2n - 2  \sqrt{n}  \rfloor.
\end{align*}
\end{proposition}
\begin{proof} The proof follows as in \cite[Proposition 4.3]{Mainini-Piovano},   additionally observing that  all atoms in the bond graph are bonded to particles of opposite charge only   and that for such configurations, up to neglecting the charge, our energy coincides with the one considered in \cite{Mainini-Piovano}.  
\end{proof}

\subsection{Squares with additional trapezoid} \label{subsection:sqare rectangle}

Recall that the above configurations have net charge in $\{-1,0,1\}$. Starting with a square and attaching a trapezoid in a suitable way, we can also construct configurations with energy $-\lfloor\beta(n)\rfloor$ having net charge of order $n^{1/4}$.
The construction is inspired by related ideas \cite{Davoli15, FriedrichKreutzHexagonal} used in connection to the derivation of the so-called  $n^{3/4}$-law.

We choose $k= 5l^2+7l+3$,   $l \in 2\mathbb{N}$,   and $n=k^2+1$. We construct a configuration $C_n$ as follows. We start from the square $S_{(k-l)^2}$.  Since $k-l$ is odd, the net charge of $S_{(k-l)^2}$  can be chosen as $+1$.  We add a new atom to the bond graph in such a way that it gets bonded to the second up-most among the rightmost atoms. Then we add descendingly  atoms along the right side of the square  in such a way that they are bonded to the atom in the previous step and one atom of the square. In this way, we add $(k-l-2)$ atoms. Since $k-l$ is odd, we have added $\lceil (k-l)/2 \rceil-1$ atoms of charge $+1$ and $\lfloor (k-l)/2 \rfloor-1$ atoms of charge $-1$. Next, we add a new column starting from the second up-most among the atoms of the previous column.  In this way, we add  $(k-l-4)$ atoms. We repeat this  procedure  until we have added $2l+1$ columns of atoms. This corresponds to 
\begin{align*}
\sum_{j=1}^{2l+1}(k-l-2j)=2kl-l^2+1
\end{align*}
added atoms. Note that in each column the number of added atoms of charge $+1$ exceeds the number of added atoms of charge $-1$ exactly by one,  and that the resulting configuration consists of  $n=k^2+1$  atoms.   The construction is sketched in Fig.~\ref{Fig:trapezoid}.
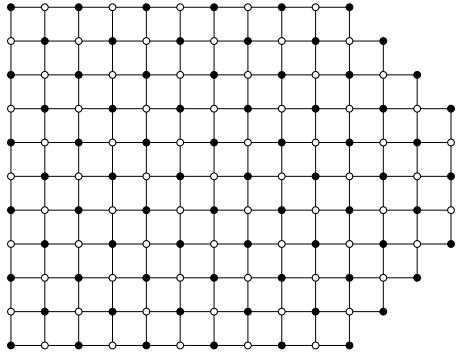
\begin{figure}[htp]
\centering
\begin{tikzpicture}[scale=0.45]
\draw[very thin](10,5)--++(3,0);
\draw[very thin](10,6)--++(3,0);
\draw[very thin](10,4)--++(3,0);
\foreach \j in {0,...,2}{
\draw[very thin](10,\j+1)--++(\j+1,0);
\draw[very thin](10,9-\j)--++(\j+1,0);
}

\foreach \j in {0,...,2}{
\draw[very thin](11+\j,\j+1)--++(0,8-2*\j);
}

\draw[fill=black](10,10) circle(.1);
\draw[step=1cm,very thin] (0,0) grid (10,10);
\foreach \j in {0,...,4}{
\draw[fill=black](10,2*\j) circle(.1);
\draw[fill=white](10,2*\j+1) circle(.1);

\draw[fill=black](2*\j,10) circle(.1);
\draw[fill=white](2*\j+1,10) circle(.1);
\foreach \i in {0,...,4}{
\draw[fill=black](2*\i,2*\j) circle(.1);
\draw[fill=black](2*\i+1,2*\j+1) circle(.1);
\draw[fill=white](2*\i+1,2*\j) circle(.1);
\draw[fill=white](2*\i,2*\j+1) circle(.1);
}
}

\foreach \j in {0,...,2}{
\pgfmathtruncatemacro\result{5-\j}
\pgfmathtruncatemacro\results{4-\j}
\foreach \i in {1,...,\result}{
\draw[fill=black](11+\j,2*\i+\j-1) circle(.1);
}
\foreach \i in {1,...,\results}{
\draw[fill=white](11+\j,2*\i+\j) circle(.1);
}
}

\end{tikzpicture}
\caption{Construction of a square with an additional trapezoid.  Three columns of atoms  have been added where the first and the last atom in the added columns  have charge $+1$. Thus, the net charge is  $4$.    Note that $k=12$ and $l=1$ is actually not admissible but chosen here for illustration purposes.}
\label{Fig:trapezoid}
\end{figure}

 We now determine the energy of the configurations. Recall the definition of $\beta$ in \eqref{def: beta}.  \EEE We observe that in a column where we  add \EEE  $m$  \EEE atoms we add exactly $2m-1$ bonds to the bond graph. Consequently, in view of Proposition \ref{PropositionDaisy}, the energy of $C_n$ is given by
\begin{align*}
\mathcal{E}(C_n) &= \mathcal{E}( S_{(k-l)^2} )-\sum_{j=1}^{2l+1}(2(k-l-2j)-1) = -2(k-l)^2+2(k-l)-4kl+2l^2-2 +2l+1 \\&=-2k^2 + 2\sqrt{k^2} -1= -\lfloor \beta(k^2)\rfloor -1= -\lfloor \beta(n)\rfloor,
\end{align*}
 where in the last step we used that $n =k^2 + 1$.  \EEE We now determine the net charge of the configuration. Recall that the configuration  $S_{(k-l)^2}$  has net charge equal to $+1$. As explained above, in each column the number of added atoms of charge $+1$ exceeds the number of added atoms of charge $-1$ by exactly one, i.e., $\mathcal{Q}(C_n) = 2l+2$.

  We are now in the position to give the proof of Theorem \ref{theorem : charge}(b). To this end, consider the sequence of integers $n_l = (5l^2+7l+3)^2+1$, $l \in 2\mathbb{N}$, and the configurations $C_{n_l}$ constructed above. Note that $4l \geq  n_l^{1/4}$. Thus, we obtain
$
\mathcal{Q}(C_{n_l}) = 2l+2 \geq \frac{1}{2}n_l^{1/4}.
$
This yields
\begin{align*}
\liminf_{l\to +\infty} n_l^{-1/4}|\mathcal{Q}(C_{n_l})| >0.
\end{align*}
The statement follows once we know that the ground state energy equals exactly $-\lfloor \beta(n)\rfloor$ for all $n \in \mathbb{N}$. This will be proven in  Subsection  \ref{subsection:energetic characterization}.

\subsection{Upper bound on $q_\mathrm{sat}^n$}\label{subsection : upper bound qsat} 

We  construct  sub-configurations of the square lattice satisfying $\mathcal{E}(C_n) = -2n +2|\mathcal{Q}(C_n)|$ and $|\mathcal{Q}(C_n)|= \phi(n)$. This shows $q_\mathrm{sat}^n \le \phi(n)$, see \eqref{eq: qsat-def}--\eqref{def: phi}.   The main idea of the construction is to place atoms of the minority phase only at sites whose $\mathbb{Z}^2$ neighborhood, that is the set of points on $\mathbb{Z}^2$ with distance $1$ to the point, is already occupied by four points of the majority phase. This leads to configurations whose global geometry is reminiscent of a diamond. \EEE For an illustration of the construction we refer to Fig.~\ref{fig : construction qsat}.

For $1 \leq n \leq 4$ we arrange $ n$ atoms of charge $+1$ such that their mutual distance is bigger or equal to $\sqrt{2}$.
Next, we provide the construction for $n =1 + 2k^2+2k$, $k \in \mathbb{N}$. In this case, we define  $X_n=D_n$, see \eqref{def :Dn},  and $q_{i}=(-1)^{i_1+i_2+k}$, $i=(i_1,i_2) \in  D_n  $. Now for $n = 1 +2 k^2+2k+m$, with $1\leq m \leq 4k+3$, we recursively construct $X_{ 1 +2 k^2+2k+m+1}$ by adding one atom to the configuration $X_{1 +2 k^2+2k+m}$. For $m=1$ we add an atom of positive charge at position $(-(k+1),1)$. Let $2\leq m \leq 2k+2$. For $m$ even, we add  an atom of positive charge at position $(-(k+1)+m/2,1+m/2)$.  For $m$ odd, we add an atom of negative charge at position  $(-(k+1)-\lfloor -m/2 \rfloor, -\lfloor- m/2 \rfloor)$.   Now  let $2k+3 \leq m \leq 4k+3$. For $m$ odd, we add an atom  of positive charge at position  $( -\lfloor - m/2 \rfloor -(k+1),2k+3-\lfloor m/2 \rfloor)$.  For $m$ even,   we add an atom of negative charge at position $( m/2 -(k+2),2k+3- m/2)$. 
%
%

\begin{figure}[htp]
\centering
\begin{tikzpicture}[scale=0.5]

\foreach \j in {0,...,5}{
\draw[ultra thin](\j,-5+\j)--++(0,10-2*\j);
\draw[ultra thin](-\j,-5+\j)--++(0,10-2*\j);
\draw[ultra thin](\j-5,\j)--++(10-2*\j,0);
\draw[ultra thin](\j-5,-\j)--++(10-2*\j,0);
}

\foreach \j in {1,3,5}{
\pgfmathsetmacro\y{\j-1}
\foreach \l in {0,...,\y}{
\foreach \k in {0,...,3}{
\begin{scope}[rotate = \k*90]
\draw[fill=black](0,0)++(\j,0)++(-\l,\l) circle(.07);
\end{scope}
}
}}

\foreach \j in {0,2,4}{
\pgfmathsetmacro\y{\j-1}
\foreach \l in {0,...,\y}{
\foreach \k in {0,...,3}{
\begin{scope}[rotate = \k*90]
\draw[fill=white](0,0)++(\j,0)++(-\l,\l) circle(.07);
\end{scope}
}
}}

\draw[fill=black](-6,1) circle(.07);
\draw(-6,1) node[anchor=east]{\small $1$};

\draw[fill=black](-5,2) circle(.07);
\draw(-5,2) node[anchor=east]{\small $2$};

\draw[fill=white](-5,1) circle(.07);
\draw(-5,1) node[anchor=east]{\small $3$};

\draw[fill=black](-4,3) circle(.07);
\draw(-4,3) node[anchor=east]{\small $4$};

\draw[fill=white](-4,2) circle(.07);
\draw(-4,2) node[anchor=east]{\small $5$};

\draw[fill=black](-3,4) circle(.07);
\draw(-3,4) node[anchor=east]{\small $6$};

\draw[fill=white](-3,3) circle(.07);
\draw(-3,3) node[anchor=east]{\small $7$};

\draw[fill=black](-2,5) circle(.07);
\draw(-2,5) node[anchor=east]{\small $8$};

\draw[fill=white](-2,4) circle(.07);
\draw(-2,4) node[anchor=east]{\small $9$};

\draw[fill=black](-1,6) circle(.07);
\draw(-1,6) node[anchor=east]{\small $10$};

\draw[fill=white](-1,5) circle(.07);
\draw(-1,5) node[anchor=east]{\small $11$};

\draw[fill=black](0,7) circle(.07);
\draw(0,7) node[anchor=east]{\small $12$};

\draw[fill=white](0,6) circle(.07);
\draw(0,6) node[anchor=east]{\small $14$};

\begin{scope}[rotate=-90]
\draw[fill=black](-6,1) circle(.07);
\draw(-6,1) node[anchor=east]{\small $13$};

\draw[fill=black](-5,2) circle(.07);
\draw(-5,2) node[anchor=east]{\small $ 15$};

\draw[fill=white](-5,1) circle(.07);
\draw(-5,1) node[anchor=east]{\small $16$};

\draw[fill=black](-4,3) circle(.07);
\draw(-4,3) node[anchor=east]{\small $17$};

\draw[fill=white](-4,2) circle(.07);
\draw(-4,2) node[anchor=east]{\small $18$};

\draw[fill=black](-3,4) circle(.07);
\draw(-3,4) node[anchor=east]{\small $19$};

\draw[fill=white](-3,3) circle(.07);
\draw(-3,3) node[anchor=east]{\small $20$};

\draw[fill=black](-2,5) circle(.07);
\draw(-2,5) node[anchor=east]{\small $21$};

\draw[fill=white](-2,4) circle(.07);
\draw(-2,4) node[anchor=east]{\small $22$};

\draw[fill=black](-1,6) circle(.07);
\draw(-1,6) node[anchor=east]{\small $23$};

\end{scope}
\end{tikzpicture}
\caption{The construction of $\phi(n)$-optimal configurations $C_n$ where $n=1 + 2k^2 + 2k +m$ for  $k=5$ and  $1 \le m \le 23$.} 
\label{fig : construction qsat}
\end{figure}
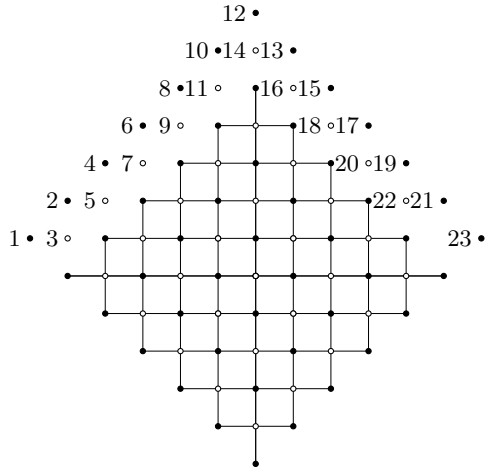

\begin{proposition}[Upper bound for $q_\mathrm{sat}^n$] \label{proposition : upper bound} Let $n \in \mathbb{N}$. Then
$q_\mathrm{sat}^n \leq \phi(n)$.
\end{proposition}
\begin{proof} We  check that for every $n \in \mathbb{N}$  the constructed  configuration $C_n$  satisfies $\mathcal{E}(C_n)=-2n +2\mathcal{Q}(C_n)$ and  $\mathcal{Q}(C_n)=\phi(n)$. This implies $q_\mathrm{sat}^n \leq \phi(n)$.

 Let us first  confirm   that, for all $n \in \mathbb{N}$,  $C_n$  satisfies $\mathcal{E}(C_n)= -2n + 2\mathcal{Q}(C_n)$.   Recall definition \eqref{def : Cnpm}.  Since  $\# X_n^+ + \# X_n^-=n$,  $\# X_n^+ - \# X_n^-=\mathcal{Q}(C_n)$, and $X_n^-\leq X_n^+$  we get  
\begin{align}\label{eq: good for later} 
 \# X_n^- = \frac{1}{2}(n-\mathcal{Q}(C_n)).
\end{align}
The construction shows that each atom with negative charge has exactly four neighbors of positive charge since we   place atoms of negative charge only on lattice sites whose $\mathbb{Z}^2$-neighbourhood  is already occupied by four  atoms   of positive charge. Moreover, we observe that the distance of atoms with the same charge is at least $\sqrt{2}$. Thus, by [ii], [vii], and \eqref{eq: good for later} we get
$\mathcal{E}(C_n)= -4  \# X_n^- = -2n +2\mathcal{Q}(C_n)$. \EEE

 It remains to prove that $\mathcal{Q}(C_n)=\phi(n)$. To this end,  it is convenient to  use a different representation of $\phi$ given by
\begin{align*}
\phi(n) = \begin{cases}
2 +  2k  + 0  & \ \ \ \text{if } n=1+2k^2+2k+m, \ \ \  1\leq m \leq 2k+2 , \ m \text{ odd},\\
2+ 2k +  1 &\ \ \ \text{if } n=1+2k^2+2k+m,  \ \ \   1\leq m \leq 2k+2, \ m \text{ even},\\
3+ 2k + 1  &\ \ \ \text{if } n=1+2k^2+2k+m, \ \ \   2k+3\leq m \leq 4k+ 4 \  m \text{ odd},\\
3+ 2k + 0 &\ \ \ \text{if } n=1+2k^2+2k+m, \ \ \  2k+ 3\leq m \leq 4k+4, \ m \text{ even},\\
\end{cases}
\end{align*} 
see \eqref{eq: different repr} in  Subsection \EEE \ref{subsection : calculation qsat} below.  
 Let $n =1 + 2k^2+2k = 1 + 2(k-1)^2+2(k-1) + 4(k-1)+4 $.  In this case,  $C_n$ is a union of $2k+1$  rows  for each of which the number of atoms of charge $+1$ exceeds the number of atoms of charge $-1$ by  exactly  one. We therefore have $\mathcal{Q}(C_n)=1+2k=\phi(n)$. 

For $n =1 + 2k^2+2k +m$ with $ m \in \{1,2\}$, we add $m$ atoms of charge $+1$  to the configuration $C_{1+2k^2+2k}$.  Consequently, we get $\mathcal{Q}(C_n) = \mathcal{Q}(C_{1+2k^2+2k})+ m = 1+2k+m= \phi(n)$. 

For $3 \leq m \leq 2k+  2  $, we add alternatingly first an atom of charge $-1$ and then an atom of charge $+1$. We therefore obtain $\mathcal{Q}(C_n) = \mathcal{Q}(C_{1+2k^2+2k}) + 1 = 1+2k +1  = \phi(n)$ if $m$ is odd and $\mathcal{Q}(C_n) = \mathcal{Q}(C_{1+2k^2+2k}) + 2= 1+2k+2 = \phi(n)$ if $m$ is even.  For $2k+3 \leq m \leq 4k+3$, we add alternatingly first an atom of charge $+1$ and then an atom of charge $-1$. We therefore obtain $\mathcal{Q}(C_n) = \mathcal{Q}(C_{1+2k^2+4k+2}) + 1 = 1+2k +3  = \phi(n)$ if $m$ is odd and $\mathcal{Q}(C_n) = \mathcal{Q}(C_{1+2k^2+4k+2}) =  2k +3 \EEE = \phi(n)$ if $m$ is even.
\end{proof}

\section{Elementary properties of optimal configurations}\label{sec : elementary}
In this section we prove elementary geometric properties  of   optimal configurations.  For their definition, we refer to the paragraph below   \eqref{def : mnq}. 

\begin{lemma}\label{LemmaNeighborhood} 
 Let $C_n$ be an optimal configuration.  Then $C_n$ has alternating charge distribution,  all bond angles satisfy  
\begin{align}\label{ineq: bondangles}
\tfrac{2\pi}{5} \leq \theta \leq \tfrac{8\pi}{5},
\end{align} 
 \EEE  and
\begin{align}\label{NeighbourhoodCard}
\#\mathcal{N}(x_i) \leq 4 \text{ for all } i \in \{1,\ldots,n\}.
\end{align}

\end{lemma}

\begin{proof}  In this proof, we will use the following convention: we say that  we \textit{relocate} $(x,q) \in C_n$,   and write $C_n-\{(x,q)\}$, by considering the configuration $(C_n\cup (x+\tau,q))\setminus (x,q) $, where $\tau \in \mathbb{R}^2$ is chosen such that
\begin{align*}
\mathrm{dist}(X_n,x+\tau)\geq \sqrt{2}.
\end{align*}
Since  $C_n$ is a $q_{\rm net}$-optimal configuration, there  holds  $\mathcal{E}(C_n)<+\infty$. Thus,  by [i] and [iv] 
\begin{align}\label{eq: distance}
\mathrm{dist}(x_i,  X_n \EEE \setminus\{x_i\}) \geq 1 \text{ for all }i\in \{1,\ldots,  n \EEE \}.
\end{align} 
 For brevity, we define $\mathcal{N}_{\rm rep}(x_i) =   \mathcal{N}(x_i) \cap \left\{x_j \in   X_n: q_j=q_i\right\}$  for all $i \in \lbrace 1 ,\ldots, n \rbrace$. We give the proof of the statement in two \EEE steps. \EEE

\noindent \emph{Step 1: $\# \mathcal{N}_{\rm rep}(x_i)  =0 \text{ for all } i \in \{1,\ldots,n\}$.} \EEE First, \eqref{eq: distance} \EEE and $r_0 < (2\sin(\frac{\pi}{7}))^{-1}$  entail by an elementary  argument that \EEE
\begin{align}\label{eq: six neighbors}
\#\mathcal{N}(x_i) \leq 6 \text{ for all } i \in \{1,\ldots,n\}.
\end{align}
 Indeed, if $\#\mathcal{N}(x_i) \ge 7$, two neighbors of $x_i$ would necessarily have distance smaller than $1$.  
Now assume by contradiction that $\# \mathcal{N}_{\rm rep}(x_i) \geq 1$. Note that every bond between  atoms  of different charge contributes at least $-1$ to the energy by $[\mathrm{ii}]$.  This along with  $V_{\rm r} \ge 0$,  see [iv], \EEE  and the fact that   the energy per neighbor of same charge  exceeds  $6$ (see  $[\mathrm{iv}],[\mathrm{v}]$,   and   $[\mathrm{vi}]$)  allows us to relocate $(x_i,q_i)$:   by \eqref{eq: six neighbors}   we  obtain 
 \begin{align*}
\mathcal{E}(  C_n - \{(x_{i},q_i)\} ) &<   \mathcal{E}(C_n) + \#\left(\mathcal{N}(x_i) \setminus \mathcal{N}_{\rm rep}(x_i)\right) - 6\#\mathcal{N}_{\rm rep}(x_i)\\& \le \mathcal{E}(C_n) + 5 - 6    < \EEE  \mathcal{E}(C_n).
\end{align*} 
 This contradicts the fact that $C_n$ is  a $q_{\rm net}$-optimal configuration. Thus, Step 1 is proved.

 \noindent \emph{Step 2: $\tfrac{2\pi}{5} \leq \theta \leq \tfrac{8\pi}{5}.$ } Assume by contradiction that there exists a bond angle such that $\theta<\tfrac{2\pi}{5}$. Denote by $x_1,x_0,x_2$, $x_1,x_2 \in \mathcal{N}(x_0)$,  the three atoms forming the bond angle $\theta$.
 By Step 1 we   get  that $q_0=-q_1=-q_2$.
 Since $\sin(x)$ is   increasing  for $x\in [0,\frac{\pi}{2}]$, we have  $|x_{1}-x_{2}| \leq  2r_0 \sin\left(\theta/2\right)\leq 2r_0 \sin\left(\pi/5\right)$. By using  $[\mathrm{ii}]$,  $[\mathrm{v}]$,  [vi], \EEE  $V_{\rm r} \ge 0$ (cf. $[\mathrm{iv}]$),  and \eqref{eq: six neighbors}  we get 
 \begin{align*}
 \mathcal{E}(C_n-\{ (x_{1},q_{1})\})&\leq  \mathcal{E}(C_n)-V_{\mathrm{r}}(2r_0  \sin\left(\pi/5\right)) +  \# \big(\mathcal{N}(x_{1}) \setminus \mathcal{N}_{\rm rep} (x_{1}) \EEE \big) \EEE  \\&<  \mathcal{E}(C_n)-6 +  6    =  \mathcal{E}(C_n).
\end{align*}
If a bond angle $\theta$ satisfies $\theta>\tfrac{8\pi}{5}$, consider the opposite angle formed by the three points, i.e., $2\pi-\theta<\tfrac{2\pi}{5}$.
In both cases, this contradicts the fact that $C_n$ is   an optimal configuration and concludes the proof of  Step 2.\EEE
 \EEE

 \noindent\emph{Step 3: $\#\mathcal{N}(x_i) \leq 4$ for all $i \in \{1,\ldots,n\}$.}  Assume by contradiction that there exists $i \in \{1,\ldots,n\}$ such that $\#\mathcal{N}(x_i) \geq 5$. By Step 1 we may suppose that there exist  \EEE  $\{x_0,\ldots,x_4\}\subset \mathcal{N}(x_{i})$ with $q_j =-q_i$ $j=0,\ldots,4$. We let $\theta_j \in [0,2\pi)$ be the angle between $x_j,x_i,x_{j+1}$.  (Here and in the following the indices have to be understood modulo $4$.) \EEE We can choose $j_0 \in \{0,\ldots,4\}$ such that
\begin{align*}
\theta_{j_0} \leq \frac{1}{5}\sum_{j=0}^4 \theta_{j} = \frac{2\pi}{5}.
\end{align*}

 Now  \eqref{ineq: bondangles} and \eqref{NeighbourhoodCard} follow from Step 2 and 3, respectively.  The property of alternating charge (see \eqref{SamechargeNeighbourhood})  follows from  Step 1.  This   concludes \EEE the proof of Lemma \ref{LemmaNeighborhood}.
\end{proof}
 Recall the definition of non-equilibrated atoms $\mathcal{A}(X_n)$ at the end of Subsection \ref{sec: notions}.

\begin{lemma}[Bond-angles and polygons]\label{RemarkPolygon} Let $C_n$ be a configuration.
\begin{enumerate}
\item[(a)] If $C_n$ is  an optimal  configuration, then the bond graph consists  only of polygons of even length.  For each polygon  $P \subset X_n$ satisfying   $\mathcal{A}(X_n)\cap P \neq \emptyset$ there holds
\begin{align*}
\#(\mathcal{A}(X_n)\cap P) \geq 2.
\end{align*}
\item[(b)] If $C_n$ is such that $\mathcal{E}(C_n) =-b$, then all bonds have unit length, the configuration is repulsion-free, and all bond angles $\theta$ satisfy
\begin{align}\label{eq: lower theta}
\tfrac{\pi}{2} \leq \theta \le \tfrac{3\pi}{2}.
\end{align}
 In particular, all squares in the bond graph are regular.
\end{enumerate}
\end{lemma}

\begin{proof}

\noindent \emph{Proof of (a).} Let $C_n$ be  an optimal  configuration.  By  Lemma \ref{LemmaNeighborhood} the configuration has alternating charge distribution. Hence, the bond graph consists only of elementary polygons of even length. Let $P=\{x_1,\ldots,x_k\} \subset X_n $ be a polygon,  and denote the interior angles by $\{\theta_1,\ldots,\theta_k\}$.  Assume  by contradiction  that  $\#(\mathcal{A}(X_n)\cap P)=1$. Without loss of generality $x_1 \in \mathcal{A}(X_n)$, i.e., $\theta_1 \notin \frac{\pi}{2}\mathbb{N}$. We have $\theta_1 =\pi(k-2)- \sum_{i=2}^k \theta_i$. Now by assumption $\theta_i \in \frac{\pi}{2}  \mathbb{N}  $ for all $i \in \{2,\ldots,k\}$ and so the right hand side is  an  integer multiple   of $\frac{\pi}{2}$. This contradicts the fact that $x_1 \in \mathcal{A}(X_n)$.

\noindent\emph{Proof of (b):}  Given a configuration $C_n$ such that $\mathcal{E}(C_n)=-b$, all bonds are necessarily of unit length and the configuration is repulsion-free,  see Remark \ref{rem: repulsionsfree}.    We show  that all bond angles $\theta$ satisfy (\ref{eq: lower theta}).
In fact, suppose that $x_1$, $x_0$, $x_2$ form the angle $\theta$.  Observe that $C_n$ has  alternating charge distribution since  it is repulsion-free. As  $x_1$ and  $x_2$ are neighbors of $x_0$, we  thus  have $q_1 = q_2$.  Then we get $|x_1 - x_0| = |x_2 - x_0| =1$  as well as  $|x_1-x_2| \ge \sqrt{2}$ by [vii].   A simple geometric argument  yields $\pi/2\le\theta \le 3\pi/2$.  By (\ref{eq: lower theta}) and the fact that the four interior angles of a square sum up to $2\pi$ we deduce that all squares are regular. 
\end{proof}


\EEE

 We now investigate the relation of  net \EEE charge and energy.  Without restriction we consider   configurations with non-negative  net \EEE charge in order to simplify notation.

 \begin{lemma}[Net charge controls energy] \label{lemma : Charge Energybound}      Let $C_n$ be a   configuration  with $\mathcal{Q}(C_n) \ge 0$.
 
 \noindent (a) There holds  
\begin{align}\label{ineq : charge}
\mathcal{E}(C_n)\geq -2n +2\mathcal{Q}(C_n),
\end{align}
with equality only if
\begin{itemize}
\item[(i)] All atoms  with charge $-1$  are $4$-bonded,  and all bonds are of unit length, 
\item[(ii)]  The configuration is repulsion-free,
\item[(iii)]  $\mathcal{Q}(C_n) \ge q_\mathrm{sat}^n$, 
\item[(iv)] All $4$-bonded atoms have only bond angles  $\frac{\pi}{2}$,
\item[(v)]  If $x_i$  satisfies  $\#\mathcal{N}(x_i)=3$, then $x_i \notin \mathcal{A}(X_n)$.
\item[(vi)] For all connected components $C_m  \subset C_n$  there holds $\mathcal{Q}(C_m) \geq q_\mathrm{sat}^m$.  
\end{itemize}
 (b) Conversely, if properties (i) and (ii) are satisfied, then equality holds in \eqref{ineq : charge}. 
\end{lemma}
\begin{proof}
 We first show \eqref{ineq : charge}. To this end, it is not restrictive to assume that $C_n$ is an optimal configuration.  Recall \eqref{def : Cnpm}.   Since  $\# X_n^+ + \# X_n^-=n$ and  $\# X_n^+ - \# X_n^-=\mathcal{Q}(C_n)$, we get  $\min\lbrace \# X_n^+,  \# X_n^- \rbrace = (n-|\mathcal{Q}(C_n)|)/2$. As   $C_n$ is   an   optimal configuration,  \eqref{NeighbourhoodCard} yields $\mathcal{N}(x_i) \leq 4$ for all $i=1,\ldots,n$. We therefore obtain  by [ii] and $V_{\rm r} \ge 0$ 
\begin{align}\label{ineq : charge2}
\mathcal{E}(C_n) \geq -b \geq -4 \,   \min\lbrace \# X_n^+,  \# X_n^- \rbrace  = -2n + 2|\mathcal{Q}(C_n)|. 
\end{align}
 This shows (\ref{ineq : charge}) since by assumption $|\mathcal{Q}(C_n)| = \mathcal{Q}(C_n)$.   In a similar fashion, to see (b), it suffices to note  that (i) and (ii) imply that all inequalities in \eqref{ineq : charge2} are actually equalities. \EEE

 Now we assume that equality holds in \eqref{ineq : charge}. (In particular, this implies that $C_n$ is  an optimal configuration  and satisfies $\mathcal{E}(C_n) = -b$.)   We  confirm (i)--(vi).

\noindent\emph{Proof of (i), (ii):}  First, by Lemma \ref{RemarkPolygon}(b) we get that all bonds have unit length and that the configuration is  repulsion-free.   Suppose  now   that there exists an atom  with charge $-1$  such that $\mathcal{N}(x_i)<4$.  This implies  strict inequality in the second inequality in (\ref{ineq : charge2}). This contradicts  the equality in  \eqref{ineq : charge}.

 \noindent\emph{Proof of (iii):} The inequality follows directly from definition \eqref{eq: qsat-def}. 
 
\noindent\emph{Proof of (iv), (v):}   Now, (iv) follows from  \eqref{eq: lower theta} and the fact that  the bond angles at each atom sum up to $2\pi$.  To see (v), let   $x_0$ be a $3$-bonded atom. By (i) it necessarily has charge $+1$ and  it is only   bonded to atoms with charge $-1$. These atoms  are $4$-bonded by (i). Denote the bond angles at $x_0$ by $\theta_1,\theta_2,\theta_3$. Without loss of generality, suppose that $\pi/2\leq\theta_1 \leq 2\pi/3 \leq 3\pi/4$, see \eqref{eq: lower theta}. Assume by contradiction that $\theta_1 > \pi/2$. Denote by $x_1,x_2$ the two $4$-bonded atoms whose bonds enclose $\theta_1$. Denote by $z_1,z_2$ the two atoms $z_1 \in \mathcal{N}(x_1) \setminus \lbrace x_0 \rbrace  $, $z_2 \in \mathcal{N}(x_2) \setminus \lbrace x_0 \rbrace $ that have minimal distance to each other. Observe that $z_1$ and $z_2$ have charge $+1$.  We proceed to show that  $z_1 \neq z_2$ and $|z_1-z_2| <\sqrt{2}$ which contradicts (ii). By  (iv)  and the fact that all bonds are of unit length  we get  $|z_1-x_0| =|z_2-x_0|=\sqrt{2}$. Moreover,   the angle enclosed by  $z_1 $, $x_0$, $z_2$  is equal to $\theta_1- 2 \pi/4$,   where $0<\theta_1- 2 \pi/4  \leq \pi/4$. Thus, we obtain 
\begin{align*}
  0<  |z_1-z_2| = 2\sqrt{2} \sin\left(\frac{1}{2}\left(\theta_1 -\frac{\pi}{2}\right)\right) \leq 2\sqrt{2}\sin \left(\frac{\pi}{8}\right) <\sqrt{2}.
\end{align*}
 This yields  $z_1 \neq z_2$ and contradicts (ii). Therefore, we have shown  \EEE  $\theta_1=\pi/2$.  One of the remaining two bond angles, say $\theta_2$, then also satisfies  $\theta_2 \leq 3\pi/4$.  By the same argument we have that $\theta_2 =\pi/2$. Now $\theta_3=\pi$ and the claim is proved.

\noindent\emph{Proof of (vi):} Assume that there exist two connected components $C_m$ and $C_{n-m}$ with no bonds between them.  We have that 
\begin{align}\label{eq : charge additivity}
\mathcal{Q}(C_m)+\mathcal{Q}(C_{n-m})=\mathcal{Q}(C_n).
\end{align}
By  \eqref{ineq : charge2} \EEE applied to $C_m$ and $C_{n-m}$ there holds
\begin{align}\label{ineq : CmCn-m}
 \mathcal{E}(C_m)\geq -2m +2|\mathcal{Q}(C_m)|, \ \ \ \  \quad \mathcal{E}(C_{n-m}) \geq -2(n-m) +2|\mathcal{Q}(C_{n-m})|,
\end{align}
and therefore
\begin{align}\label{ineq : (vi)}
\begin{split}
-2n +2|\mathcal{Q}(C_n)|&=\mathcal{E}(C_n) \geq -2n  +2( |\mathcal{Q}(C_m)|  + |\mathcal{Q}(C_{n-m})|) \EEE \ge -2n  +2|\mathcal{Q}(C_n)|.
\end{split}
\end{align}
We now have equality everywhere in (\ref{ineq : (vi)}) and therefore necessarily   also  in (\ref{ineq : CmCn-m}). 
Moreover,   in view of \eqref{eq : charge additivity}, equality also implies that  there  holds $\mathcal{Q}(C_m)\ge 0$ and $\mathcal{Q}(C_{n-m}) \ge 0$.  
 By equality in (\ref{ineq : CmCn-m}) and (iii)  (applied on $C_m$ and $C_{n-m}$)  we get $\mathcal{Q}(C_m) \geq q_\mathrm{sat}^m$ and $\mathcal{Q}(C_{n-m}) \geq q_\mathrm{sat}^{n-m}$. \end{proof}

\section{Characterization of ground states}\label{sec: ground states}

 This section is devoted to the proofs of Theorem \ref{TheoremGroundstates}, Theorem \ref{theorem : GeometryGroundstatewithout}, and Theorem \ref{theorem : charge}. \EEE 

\subsection{Boundary energy}\label{section:boundary energy}

 This subsection is devoted to the concept of boundary energy and a corresponding
estimate which will be instrumental for the characterization of ground states and their energy in  Subsections \ref{subsection:energetic characterization} and \ref{subsection:geometric characterization}.  It is convenient to first introduce an auxiliary energy, sorted by the contributions of single atoms.  We then define the so-called  \EEE boundary energy   in terms of this auxiliary energy.

 \textbf{Reduced energy:}   Let $C_n$ be a configuration \EEE and let $x_i \in X_n$. \EEE We set
\begin{align} \label{def:cell energy}
V_{ \rm atom \EEE }(x_i) := \frac{1}{2} \sum_{x_j \in \mathcal{N}(x_i)} V_{\mathrm{a}}(|x_i-x_j|)+ \frac{1}{8}\underset{x_j \neq x_k}{\sum_{x_j,x_k \in \mathcal{N}(x_i)}} V_{\mathrm{r}}(|x_j-x_k|),
\end{align}
 where $\mathcal{N}(x_i)$ is defined in \eqref{eq: neighborhood}. \EEE
We define  the \emph{reduced energy} \EEE $\mathcal{R} : (\mathbb{R}^2\times \{-1,1\})^n \to \overline{\mathbb{R}}$ by 
\begin{align}\label{def:renormalized energy}
\mathcal{R}(C_n):= \sum_{x_i \in X_n} V_{\rm atom}(x_i).
\end{align}

\begin{lemma}[Relation of $\mathcal{E}$ and $\mathcal{R}$]\label{lemma: relation} If $C_n$ has alternating charge distribution and satisfies $\# \mathcal{N}(x_i) \le 4$ for all $i=1,\ldots,n$, then 
\begin{align*}
\mathcal{E}(C_n) \geq \mathcal{R}(C_n)
\end{align*}
with equality if $C_n$ is repulsion-free. 
\end{lemma}

\begin{proof}
 Since $C_n$ has alternating charge distribution, the contributions of $V_{\rm a}$ in \eqref{Energy} and \eqref{def:renormalized energy} coincide. In the case that $C_n$ is repulsion-free, the contributions of $V_{\rm r}$ vanish and thus there indeed holds $\mathcal{E}(C_n) = \mathcal{R}(C_n)$. In the general case, we observe that in \eqref{Energy}  each  pair  $\lbrace x_i,x_j \rbrace$ with $q_i=q_j$  contributes $V_{\rm r}(|x_i-x_j|)$ to the energy. In \eqref{def:renormalized energy}, however,  $\lbrace x_i,x_j \rbrace$ contributes at most $V_{\rm r}(|x_i-x_j|)$ since the pair appears at most eight times in the sum. This is due to double counting and the fact that each atom can be bonded to at most  four  other atoms. \EEE
\end{proof}

 By  Lemma \ref{LemmaNeighborhood}, ground states $C_n$ satisfy the assumptions of the lemma. Later we will also see that ground states are repulsion-free which will imply  $\mathcal{E}(C_n) = \mathcal{R}(C_n)$. \EEE


 \textbf{Boundary atoms, boundary energy:} Within the bond graph, we say that an atom   is a \textit{boundary atom} if it is not contained in the interior region of any simple cycle. Otherwise, we call it \emph{bulk atom}.  We denote the union of the boundary atoms   by $\partial X_n$ and let $d = \# \partial X_n$. \EEE A \textit{boundary bond} is a bond containing a boundary atom. All other bonds are called \textit{bulk bonds}. Given $C_n$, we define its \textit{bulk}, denoted by $C_n^{\mathrm{bulk}}$, as the sub-configuration obtained by dropping all boundary atoms (and the corresponding charges). Similarly, the particle positions are indicated by $X_n^{\mathrm{bulk}}$. \EEE With the above definition, we have that the bulk is an \EEE $(n-d)$-atom configuration.   There are  two contributions to the energy of $C_n$, namely $\mathcal{R}^{\mathrm{bnd}}$ and $\mathcal{R}^{\mathrm{bulk}}$, defined by \EEE
\begin{align}\label{eq: bnd2}
\begin{split}
\mathcal{R}^{\mathrm{bulk}}(C_n) & := \mathcal{R}(C_n^{\mathrm{bulk}}), \\ 
\mathcal{R}^{\mathrm{bnd}}(C_n) & := \mathcal{R}(C_n) - \mathcal{R}^{\mathrm{bulk}}(C_n).
\end{split}
\end{align}
 
\begin{remark}[Boundary energy]\label{rem: boundary energy}
{\normalfont
Since $C_n^{\mathrm{bulk}}$ does not contain boundary bonds of  $C_n$,  each boundary bond  $\lbrace x_i,x_j \rbrace$ with $q_i \neq q_j$  contributes $V_{\rm a}(|x_i-x_j|)$ to $\mathcal{R}^{\mathrm{bnd}}(C_n)$, and each boundary bond  $\lbrace x_i,x_j \rbrace$ with $q_i = q_j$  contributes at least $\frac{1}{4} V_{\rm r}(|x_i-x_j|)$ to $\mathcal{R}^{\mathrm{bnd}}(C_n)$. Note that $\mathcal{R}^{\mathrm{bnd}}(C_n)$ contains also pair interactions of certain bulk bonds, namely if the corresponding bulk atoms  are neighbors of   the same boundary atom, see Fig.~\ref{Fig:ConcaveAngle}. We point out that, when the boundary energy in \eqref{eq: bnd2} is defined  with $\mathcal{E}$ in place of $\mathcal{R}$ (see, e.g., \cite{FriedrichKreutzHexagonal, Mainini-Piovano, Mainini}), such pair interactions do not contribute to the boundary energy.  Our definition, slightly different in comparison to  \cite{FriedrichKreutzHexagonal},  is necessary from a technical point of view  in order to derive the `correct' boundary energy estimate  (\ref{BoundaryEnergyEstimate}) and to obtain (\ref{eq: 2})-(\ref{eq: 3}) as necessary conditions for equality in (\ref{BoundaryEnergyEstimate}). 
}
\end{remark} 

\EEE 

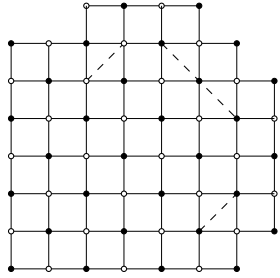
\begin{figure}[hbt!]
\begin{tikzpicture}[scale=0.5]
\draw[ultra thin](0,0) grid (6,6);
\draw[ultra thin](7,1)--++(0,4);
\draw[ultra thin](2,7)--++(3,0);
\draw[ultra thin](7,5)--++(-1,0);
\draw[fill=black](7,5)circle(.07);
\draw[dashed](4,6)--++(2,-2);

\draw[dashed](5,1)--++(1,1);
\draw[dashed](2,5)--++(1,1);
\foreach \k in {0,1}{
\draw[fill=white](2*\k+2,7) circle(.07);
\draw[ultra thin](2*\k+2,7) --++(0,-1);
\draw[fill=black](2*\k+3,7) circle(.07);
\draw[ultra thin](2*\k+3,7) --++(0,-1);

\draw[fill=black](7,2*\k+1) circle(.07);
\draw[ultra thin](7,2*\k+1) --++(-1,0);
\draw[fill=white](7,2*\k+2) circle(.07);
\draw[ultra thin](7,2*\k+2)--++(-1,0);
}

\clip(-.5,-.5) rectangle (6.5,6.5);
\foreach \k in {0,1}{
\foreach \j in {0,...,3}{
\foreach \i in {0,...,3}{
\draw[fill=black](2*\j+\k,2*\i+\k) circle(.07);
\draw[fill=white](2*\j+\k-1,2*\i+\k) circle(.07);
}
}
}

\end{tikzpicture}
\caption{The dashed bonds  contribute to $\mathcal{R}^\mathrm{bnd}$. They would not contribute to the boundary energy if it was defined with $\mathcal{E}$ in place of $\mathcal{R}$. }
\label{Fig:ConcaveAngle}
\end{figure}

%

   \textbf{Maximal polygon:} We introduce an additional notion in the case that $C_n$ is connected and does not contain acyclic bonds. In this case, the bond graph is  delimited by a simple cycle which we call the \emph{maximal polygon}.  We denote the atoms of the maximal polygon  by $\{x_1,\ldots,x_d\}$   and  the interior angle   at  $x_i \in \partial X_n$ by $\theta_i$.  Moreover, for  $2 \le k \le 4$,  we indicate by
\begin{align*}
&I_k=\{  x_i \in \partial X_n: \EEE \#\mathcal{N}(x_i)=k\}
\end{align*}
the set of  $k$-bonded  boundary atoms.  For ground states $C_n$ there holds  $\# I_2 + \# I_3 +\#I_4 = d$ by \eqref{NeighbourhoodCard}.   

 We now provide an estimate for the boundary energy  $\mathcal{R}^{\rm bnd}$.  Its proof is inspired by \cite[Lemma 3.1]{Mainini-Piovano}. The precise estimates, however, deviate significantly from   \cite{Mainini-Piovano}   due to  the presence of the repulsive potential $V_{\rm r}$ instead of an angular potential.  We defer the proof to Appendix \ref{appendix}.

\begin{lemma}\label{LemmaBoundaryEnergy}Let $n \geq 4$ and let $C_n$ be a connected ground state with no acyclic bonds.   Then
\begin{align}\label{BoundaryEnergyEstimate}
\mathcal{R}^{\mathrm{bnd}}(C_n) \geq -2d +4
\end{align}
with equality only if the following conditions are satisfied:
\begin{align}
& \text{All boundary bonds are of unit length},\label{eq: 1}\\
&\#I_2 + 2\#I_3 +3 \#I_4 = 2d -4,\label{eq: 2}\\
&\theta_i = \frac{\pi}{2} \ \text{ for } x_i \in I_2, \ \ \ \  \theta_i = \pi \ \text{ for } x_i \in I_3, \ \ \ \  \theta_i = \frac{3\pi}{2} \ \text{ for } x_i \in I_4.\label{eq: 3}
\end{align}
\end{lemma}

\begin{remark}\label{rem: bdy}
{\normalfont  
 Observe by  \eqref{eq: 2} that equality in \eqref{BoundaryEnergyEstimate} implies that $2d-4$ bonds contribute to the boundary energy. Thus, equality in \eqref{BoundaryEnergyEstimate} together with [ii] and [vii] imply that for   all boundary atoms $x_i$ one has $\min\lbrace |x_i  - x_j|: \  j \in \lbrace 1, \ldots, n \rbrace, j\neq i, \ q_j = q_i \rbrace \ge \sqrt{2}$.  Furthermore, note that \eqref{eq: 1}-\eqref{eq: 3} together with \eqref{def:renormalized energy}-\eqref{eq: bnd2} imply that for ground-state configurations $C_n$ such that there holds equality in \eqref{BoundaryEnergyEstimate} the atoms
 \begin{align}\label{eq: remark-equ} 
 \partial X_n \cup \bigcup\nolimits_{i=1}^d \mathcal{N}(x_i)
 \end{align}
 are subset of the same square lattice.  We point out that each $x_i \in I_k,k=3,4$, has exactly $k-2$ neighbors in $X_n^\mathrm{bulk}$ and that the interior bond angles are exactly $\pi/2$.}
\end{remark}

 Recall the excess of edges $\eta= \sum_{j\geq 4} (j-4)   f_j$, introduced in \eqref{Excess}, where $f_j$ denotes the number
of  elementary  polygons with $j$ vertices in the bond graph. Clearly, $\eta = 0$ if   and   only if the bond graph consists of squares only. We also recall that $b$ denotes the number of bonds in the bond graph.

\begin{lemma}[Cardinality of the bulk]\label{LemmaBoundaryestimate}  Let $C_n$ be a connected ground state with no acyclic bonds.   Then
\begin{align*}
n-d = 2b + 4 - 3n +\eta.
\end{align*}
\end{lemma}
\begin{proof}Let $f_j$ be the number elementary $j$-gons in the bond graph and let $f$ be the number of elementary polygons in the bond graph.  From Lemma \ref{RemarkPolygon}(a) we obtain  
\begin{align*}
\sum_{j\geq 4} j f_j = 2b-d, 
\end{align*} 
since by the summation on the left all bonds contained in the maximal polygon are counted
only once whereas all other bonds are counted twice.  By (\ref{Excess}) we get $4f = 2b-d -\eta$.  This along \EEE with Euler's formula $n-b+f=1$ (omitting the exterior face) yields $n-d = 2b+4 - 3n +\eta$.
\end{proof}

\subsection{Energetic characterization}\label{subsection:energetic characterization}
 We start with the proof of Theorem \ref{TheoremGroundstates}. \EEE We use the following properties of the function $\beta$ defined in \eqref{def: beta}. \EEE

\begin{lemma}\label{LemmaPropertiesbeta} The function $\beta : \mathbb{N} \to \mathbb{R}$ satisfies

\begin{itemize}
\item[(i)] $\lfloor\beta(n-m)\rfloor + \lfloor \beta(m)\rfloor +1 < \lfloor\beta(n)\rfloor$ for all $n \geq 4$ and $1\leq m,n-m \leq n$, 
\item[(ii)] $\lfloor \beta(n)\rfloor +1 \leq \lfloor\beta(n+1)\rfloor$ for all $n\geq 1$.
\end{itemize}
\end{lemma}
\begin{proof}
 See \cite[Proposition 4.1, Proposition 4.2]{Mainini-Piovano}. 
\end{proof}

\begin{lemma}\label{lemma: Wurzel}
Let $j,n,m \in \mathbb{N}$ and let $x \in \mathbb{R}$ satisfy
\begin{align*}
2- \frac{3}{2}n+ \frac{m}{2}\geq x \geq -2n +j +2\sqrt{-2x+4-3n +m}.
\end{align*}
Then $x \geq -2n +j-4 + 2\sqrt{m+n+8-2j}$.
\end{lemma}

\begin{proof} The proof is elementary: we note that the function
\begin{align*}
x \mapsto x+2n -j -2\sqrt{-2x+4-3n +m}
\end{align*}
is strictly increasing and vanishes for $x = -2n +j-4 + 2\sqrt{m+n+8-2j} $.
\end{proof}

\begin{proof}[Proof of  Theorem \ref{TheoremGroundstates}] We start by noting that every ground state $C_n$ has alternating charge distribution   and every atom has at most four neighbors   by  Lemma \ref{LemmaNeighborhood}.   By Proposition \ref{PropositionDaisy}  and  Lemma \ref{lemma: relation} the ground-state energy satisfies
\begin{align}\label{Groundstateineq1}
\mathcal{R}(C_n) \leq \mathcal{E}(C_n) \leq -\lfloor \beta(n)\rfloor.
\end{align}
We prove that, if $C_n$ is a ground state, then $C_n$ is connected,  does not contain bridges, and satisfies   
 \begin{align}\label{eq : GroundstateineqR}
 \mathcal{R}(C_n) = -b = -\lfloor\beta(n)\rfloor.
 \end{align}
 The   same statement   then  holds also for $\mathcal{E}(C_n)$.  In fact,  due to \eqref{Groundstateineq1}-\eqref{eq : GroundstateineqR}, \EEE we have  
$$-\lfloor \beta(n)\rfloor =  \mathcal{R}(C_n) \leq \mathcal{E}(C_n) \leq -\lfloor\beta(n)\rfloor.$$
 This implies $\mathcal{E}(C_n)= \mathcal{R}(C_n)=-b = -\lfloor\beta(n)\rfloor$.

  Let $n \ge 4$.   We proceed by induction.   Suppose that the statement has been proven for all $m < n$ (for $1 \le m \le 3$ this is elementary).  We first show connectedness of the ground state  and the non-existence of bridges (Step 1). Afterwards, we prove  the energy equality (Step 2).

\emph{Step 1:  $C_n$ is connected and does not contain bridges.}  Assume by contradiction that $C_n$ consists of two sub-configurations $C_m$ and $C_{n-m}$ that are connected by at most one bond. The energy contribution of this bond, if it exists, is greater or equal to $-1$. Apart from that,  we can estimate the sum of the energy contributions of both components separately. In both cases,  as  $m,n-m <n$ $n\geq 4$, using  the induction assumption, Lemma \ref{LemmaPropertiesbeta}(i), and $V_{\rm r} \ge 0$ (cf. [iv]),   we get
\begin{align*}
\mathcal{R}(C_n)  \ge \mathcal{R}(C_m) + \mathcal{R}(C_{n-m})  - 1  \ge  -\lfloor \beta(m)\rfloor - \lfloor \beta(n-m)\rfloor -1 > -\lfloor \beta(n)\rfloor.
\end{align*}
This contradicts (\ref{Groundstateineq1}) and shows that $C_n$ is connected and its bond graph does not contain any bridges.

 \emph{Step 2: Energy equality $\mathcal{R}(C_n) = -b = -\lfloor \beta(n) \rfloor$.}     We divide the proof into three steps. We first treat the case that $C_n$ contains acyclic bonds  (Step 2.1).   Afterwards, we consider only configurations $C_n$ without acyclic bonds and show $\mathcal{R}(C_n) = -b$ (Step 2.2) as well as $\mathcal{R}(C_n) = -\lfloor \beta(n) \rfloor$ (Step 2.3). 

\noindent\emph{Step 2.1: $C_n$ contains acyclic bonds.}  By Step 1, $C_n$ does not contain bridges.  If there exist flags, we can find an atom $x_i$ such that removing $x_i$ removes exactly one flag.  We can count the energy contribution of  this  flag by at least $-1$ and we estimate the energy of the rest of the configuration   by induction. By Lemma \ref{LemmaPropertiesbeta}(ii) we get  
\begin{align*}
\mathcal{R}(C_n) \ge  -1 + \mathcal{R}(C_n \setminus \lbrace (x_i,q_i)\rbrace)   \geq -1 -\lfloor \beta(n-1)\rfloor \geq -\lfloor \beta(n)\rfloor.
\end{align*}
 Equality also shows that $C_n \setminus \lbrace (x_i,q_i)\rbrace$ has $\lfloor \beta(n-1) \rfloor$ bonds by induction and $C_n$ has  $\lfloor \beta(n-1) \rfloor+1 = \lfloor \beta(n) \rfloor$ bonds.

\noindent\emph{Step 2.2: $\mathcal{R}(C_n) = -b$ for connected $C_n$ with no acyclic bonds.} Assume by contradiction that $\mathcal{R}(C_n) > -b$, i.e.,  there exist $x_1,x_2 \in X_n$ such that $q_1=q_2$ and $|x_1-x_2|<\sqrt{2}$  or there exists a bond between $x_1,x_2\in X_n$  such that  $q_1=-q_2$ and $|x_1-x_2| >1$. Now if $x_1 \in \partial X_n$ or $x_2 \in \partial X_n$ we have by   (\ref{BoundaryEnergyEstimate}),  \eqref{eq: 1},  and Remark \ref{rem: bdy}  
\begin{align*}
\mathcal{R}^{\mathrm{bnd}}(C_n) > -2d + 4.
\end{align*}
Moreover, by \eqref{eq: bnd2}, the  induction hypothesis, and \eqref{eq : GroundstateineqR} there holds \EEE
\begin{align*}
\mathcal{R}^{\mathrm{bulk}}(C_n)  =  \mathcal{R}(C_n^{\mathrm{bulk}})  \geq -\lfloor \beta(n-d)\rfloor.
\end{align*}
 On the other hand, if  $x_1, x_2 \notin \partial X_n$,  by Lemma \ref{LemmaBoundaryEnergy} we get $\mathcal{R}^{\mathrm{bnd}}(C_n) \geq -2d + 4$.  By \eqref{eq: bnd2} and the induction assumption we obtain  
\begin{align*}
 \mathcal{R}^{\mathrm{bulk}}(C_n) = \mathcal{R}(C_n^{\mathrm{bulk}}) >  -\lfloor \beta(n-d)\rfloor.
\end{align*}
In both cases,  by \eqref{eq: bnd2} and  \eqref{def: beta} there holds  $\mathcal{R}(C_n) > -\lfloor 2n -2\sqrt{(n-d)}\rfloor +4$. Since the right hand side is an integer, we obtain
\begin{align}\label{Integerinequality1}
-(\lfloor -\mathcal{R}(C_n)\rfloor +1) \geq -2n +2\sqrt{(n-d)} + 4.
\end{align}
 In a similar fashion, the assumption $\mathcal{R}(C_n) > -b$  implies $-\lfloor-\mathcal{R}(C_n)\rfloor -1   \geq -b$. Now by Lemma \ref{LemmaBoundaryestimate}   we obtain $n-d \geq 2 (\lfloor-\mathcal{R}(C_n)\rfloor+1) + 4 -3n$. This along with   (\ref{Integerinequality1})  yields
\begin{align*}
-(\lfloor -\mathcal{R}(C_n)\rfloor +1) \geq -2n +2\sqrt{2(\lfloor-\mathcal{R}(C_n)\rfloor+1) + 4 -3n} + 4.
\end{align*}
 Note that  $-  (\lfloor -\mathcal{R}(C_n)\rfloor +1) \leq 2- 3n/2  $. 
 By using Lemma \ref{lemma: Wurzel} with $j =4,m=0$, and $x=  -  (\lfloor -\mathcal{R}(C_n)\rfloor +1)$ we obtain
\begin{align*}
-(\lfloor -\mathcal{R}(C_n)\rfloor +1) \geq -2n + 2\sqrt{n}.
\end{align*}
 Since  the left hand side is an integer, we get  by \eqref{def: beta} 
\begin{align*}
\mathcal{R}(C_n) > -(\lfloor -\mathcal{R}(C_n)\rfloor +1) \geq -\lfloor 2n - 2\sqrt{n}\rfloor  = - \lfloor \beta(n) \rfloor. 
\end{align*}
 This contradicts  \eqref{Groundstateineq1}.

\noindent\emph{Step 2.3: $\mathcal{R}(C_n) = -\lfloor \beta(n) \rfloor$ for connected $C_n$ with no acyclic bonds.} 
  Due to (\ref{Groundstateineq1}), it suffices to prove $\mathcal{R}(C_n) \geq -\lfloor \beta(n)\rfloor$.  We again   proceed by induction.   By Lemma \ref{LemmaBoundaryEnergy}, \eqref{eq: bnd2},  and the induction hypothesis  we obtain 
\begin{align*}
\mathcal{R}^{\mathrm{bnd}}(C_n) \geq -2d+4, \ \ \ \ \ \  \mathcal{R}^{\mathrm{bulk}}(C_n)= \mathcal{R}(C_n^{\mathrm{bulk}}) \geq -\lfloor \beta(n-d)\rfloor.
\end{align*}
 By   \eqref{def: beta}  and \eqref{eq: bnd2} there holds  
$ \mathcal{R}(C_n) \geq -2n +2\sqrt{n-d} + 4$. By Lemma \ref{LemmaBoundaryestimate} and  Step 2.2   we obtain $n-d \geq  - 2 \mathcal{R}(C_n)  + 4 -3n$.
This yields
\begin{align*}
\mathcal{R}(C_n) \geq  -2n + 2\sqrt{-2\mathcal{R}(C_n)-3n +4}  + 4.
\end{align*}
By applying Lemma \ref{lemma: Wurzel} with $j=4$, $m=0$,  and $x = \mathcal{R}(C_n)$   we obtain $\mathcal{R}(C_n) \geq -\beta(n)$. Finally, since $\mathcal{R}(C_n)$ is an integer due to Step 2.2,  we conclude $\mathcal{R}(C_n) \ge -\lfloor \beta(n) \rfloor$. \EEE 
\end{proof}

\subsection{Geometric Characterization}\label{subsection:geometric characterization}

 We now proceed with the proof of Theorem \ref{theorem : GeometryGroundstatewithout}. To this end, we first prove two lemmas about flags and  non-equilibrated atoms. 

\begin{lemma}[Flags]\label{lemma:flags} Let $n \geq 4$ and let $C_n$ be a ground state. Then the bond graph of $C_n$ contains at most one flag.
\end{lemma}
\begin{proof} Assume by contradiction that there exist  at least two  flags.  We can choose two flags such that removing the two flags increases the energy of the configuration at most by $2$.  By applying Thereom \ref{TheoremGroundstates} to the sub-configuration after removing the two flags, \EEE we obtain $\mathcal{E}(C_n) \geq -2 - \lfloor \beta(n-2)\rfloor$. Therefore, we get  by \eqref{def: beta} 
\begin{align*}
\mathcal{E}(C_n)\geq  -2  - 2(n-2) + 2\sqrt{n-2} &= -2n + 2\sqrt{n} -2\sqrt{n} + 2 +  2\sqrt{n-2} \\&= -2n + 2 \sqrt{n} +2 - \frac{4}{\sqrt{n-2}+ \sqrt{n}}.
\end{align*}
We have that $2 - \frac{4}{\sqrt{n-2}+ \sqrt{n}} \geq 1$ for all $n \geq 6$. With the above estimate this implies $\mathcal{E}(C_n) > -\lfloor \beta(n) \rfloor$.  This gives a contradiction to Theorem \ref{TheoremGroundstates} in the cases $n \geq 6$.  For $n=4,5$ it   can be checked directly that  $-2 - \lfloor \beta(n-2)\rfloor > -\lfloor \beta(n)\rfloor$. 
\end{proof}

 Recall the definition of  non-equilibrated atoms in Subsection \ref{sec: notions}. By $\mathcal{A}_{\rm bulk}(X_n) \subset \mathcal{A}(X_n)$ we additionally denote the bulk atoms which are not equilibrated. 

\begin{lemma}[Non-equilibrated atoms] \label{lemma : non-eq-atoms} Let $C_n$ be a ground state and assume that its bond graph does not contain any acyclic bonds. If $\mathcal{A}(X_n)\neq \emptyset$, then $\#\mathcal{A}(X_n) \geq 2$ and $\eta \geq 2$. If $ \mathcal{A}_{\mathrm{bulk}}(X_n) \neq \emptyset$, then $\#\mathcal{A}(X_n) \geq 2$ and $\eta \geq 4$.
\end{lemma}

\begin{proof} We split the proof into several steps. First, we prove that $\theta \in [\frac{\pi}{2},\frac{3\pi}{2}]$ for all bond angles $\theta$. Secondly, we show that,  if $\mathcal{A}(X_n) \neq \emptyset$, then $\#\mathcal{A}(X_n) \geq 2$. Then we prove  $\eta \geq 2$  if $\mathcal{A}(X_n) \neq \emptyset$. Finally, we confirm  that, if $\mathcal{A}_{\mathrm{bulk}}(X_n) \neq \emptyset$, then $\eta \geq 4$. These statements show the thesis of the lemma.

\noindent\emph{Step 1: $\theta \in [\frac{\pi}{2},\frac{3\pi}{2}]$ for all bond angles $\theta$.} By Theorem \ref{TheoremGroundstates}  there holds  $\mathcal{E}(C_n)=-b$.  Thus, \EEE Lemma \ref{RemarkPolygon}(b)  implies \EEE that $\theta \in [\frac{\pi}{2},\frac{3\pi}{2}]$ for all bond angles $\theta$.\EEE

\noindent\emph{Step 2: If $\mathcal{A}(X_n) \neq \emptyset$, then $\#\mathcal{A}(X_n) \geq 2$.}
Since the bond graph of $C_n$ does not contain any acyclic bonds and $\mathcal{A}(X_n)\neq \emptyset$, there holds $\mathcal{A}(X_n) \cap P \neq \emptyset$ for some polygon $P$ in the bond graph. By Lemma \ref{RemarkPolygon}(a) it follows that $\#(\mathcal{A}(X_n)\cap P) \geq 2$.\EEE

\noindent\emph{Step  3: If $\mathcal{A}(X_n) \neq \emptyset$, then $\eta \geq 2$.}  
Assume by contradiction that $\eta=0$. By the definition of $\eta$ this implies that all elementary polygons in the bond graph are squares. Furthermore, since the bond graph of $C_n$ does not contain any acyclic bonds, all atoms are contained in some elementary polygon, i.e., in some square. Lastly, we note that, due to Theorem \ref{TheoremGroundstates}, there holds $\mathcal{E}(C_n)=-b$ and therefore by Lemma \ref{RemarkPolygon}(b) all squares are regular. In particular, all their atoms are equilibrated. This contradicts the fact that $\mathcal{A}(X_n) \neq \emptyset$.  
 
\noindent\emph{Step 4: If $\mathcal{A}_{\mathrm{bulk}}(X_n) \neq \emptyset$, then $\eta \geq 4$.} Let  $x \in \mathcal{A}_{\rm bulk}(X_n)$  and denote by $\theta_1, \ldots,\theta_k$ the bond angles at $x$. By \eqref{NeighbourhoodCard} we have that $k\leq 4$.  We have $k \in \{2,3,4\}$.    We conclude once we have shown  that there exist at least two bond angles at $x$ that are not equilibrated. In fact, since no acyclic bonds are present in the bond graph and $x \in \mathcal{A}_{\rm bulk}(X_n)$, all the bond angles at $x$ belong to different elementary polygons.  Then,  there exist at least two non squares in the bond graph and hence $\eta \geq 4$. 

It thus remains to prove that there exist $i_1,i_2 \in \{1,\ldots,k\}$,  $i_1 \neq i_2$,  such that $\theta_{i_1},\theta_{i_2} \notin \{\frac{\pi}{2},\pi,\frac{3\pi}{2}\}$. Recall $x \in \mathcal{A}(X_n)$.  Without restriction we suppose that $\theta_1 \notin \{\frac{\pi}{2},\pi,\frac{3\pi}{2}\}$. Now assume that $ \theta_j \in \{\frac{\pi}{2},\pi,\frac{3\pi}{2}\}$ for all $j = 2,\ldots, k$. Since the bond angles at $x$ need to sum to $2\pi$, we obtain
\begin{align*}
k'\frac{\pi}{2} + \theta_1 = \sum_{i=1}^k \theta_i = 2\pi \ \ \ \ \ \ \ \text{ with } \  k'= \sum_{j=1}^3\#\lbrace i \in \{2,\ldots,k\} : \theta_i= j\pi/2\rbrace.
\end{align*}
By  Step 1   we get   $\theta_1 \in \frac{\pi}{2}  \mathbb{N}  \cap [\frac{\pi}{2},\frac{3\pi}{2}]$.  This contradicts the fact that $\theta_1 \notin \{\frac{\pi}{2},\pi,\frac{3\pi}{2}\}$.
\end{proof}

 We are now in the position to prove Theorem \ref{theorem : GeometryGroundstatewithout}(a).

\begin{proof}[Proof of Theorem \ref{theorem : GeometryGroundstatewithout}(a)] We prove the statement by induction.   For $=1,2,3$ the statement is clearly true, and for $n=4$ it follows from Lemma  \ref{RemarkPolygon}(b).   Let $n \ge 5$.  By Theorem \ref{TheoremGroundstates} and Lemma \ref{lemma:flags}  the bond graph of ground states does not contain bridges and at most one flag. Therefore,  in view of Lemma \ref{LemmaPropertiesbeta}(ii),  up to removing a flag, it is not restrictive to  assume that $C_n$ is a ground state with no acyclic bonds. We will use the following fact several times: by Remark \ref{rem: repulsionsfree}, $C_n$ is repulsion-free. Thus Lemma \ref{lemma: relation} implies
\begin{align}\label{eq: several times}
\mathcal{R}(C_n^{\mathrm{bulk}})   = \mathcal{E}(C_n^{\mathrm{bulk}}),
\end{align}
 where $C_n^{\mathrm{bulk}}$ is defined before \eqref{eq: bnd2}. 
 \EEE By the definition of $\mathcal{A}(X_n)$ and due to the fact that ground states are connected (cf.\ Theorem \ref{TheoremGroundstates}), it suffices to prove that $\mathcal{A}(X_n)=\emptyset$. 
We divide the proof into two steps. We first prove that $\mathcal{A}(X_n) \setminus \mathcal{A}_{\mathrm{bulk}}(X_n) = \emptyset$ and secondly we show that $\mathcal{A}_{\mathrm{bulk}}(X_n) = \emptyset$.

\noindent\emph{Step 1:  $\mathcal{A}(X_n) \setminus \mathcal{A}_{\mathrm{bulk}}(X_n) = \emptyset$.} Assume by contradiction that $\mathcal{A}(X_n) \setminus \mathcal{A}_{\mathrm{bulk}}(X_n) \neq \emptyset$. By Lemma \ref{LemmaBoundaryEnergy} this implies  $\mathcal{R}^{\mathrm{bnd}}(C_n) > -2d + 4$. Then  \eqref{eq: several times}  and Theorem \ref{TheoremGroundstates} (applied for $C_n^{\mathrm{bulk}}$)  yield \EEE
\begin{align*}
\mathcal{R}^{\mathrm{bulk}}(C_n)  = \mathcal{R}(C_n^{\mathrm{bulk}})    = \mathcal{E}(C_n^{\mathrm{bulk}}) \geq -\lfloor \beta(n-d)\rfloor.
\end{align*}
 This along with $\mathcal{R}^{\mathrm{bnd}}(C_n) > -2d + 4$,   \eqref{eq: bnd2},   and again Lemma \ref{lemma: relation} gives    
\begin{align}\label{eq: prev-inequ}
\mathcal{E}(C_n) \geq \mathcal{R}(C_n) > -2d + 4 -\lfloor 2(n-d) - 2\sqrt{n-d}\rfloor = -\lfloor 2n -2\sqrt{n-d}\rfloor +4.
\end{align} 
By Theorem \ref{TheoremGroundstates}, $\mathcal{E}(C_n)$ is an integer and consequently we derive  
\begin{align*}
\mathcal{E}(C_n)  \ge  -\lfloor 2n - 2\sqrt{n-d}\rfloor +5 \geq -2n +2\sqrt{n-d} + 5.
\end{align*}
Now by  using Lemma \ref{LemmaBoundaryestimate} and $\mathcal{E}(C_n) =-b$ (see Theorem \ref{TheoremGroundstates}) we obtain
\begin{align*}
\mathcal{E}(C_n) \geq -2n + 2\sqrt{-2\mathcal{E}(C_n)+4-3n+\eta} +5.
\end{align*}
 Observe that   $\eta \geq 2$ by Lemma \ref{lemma : non-eq-atoms}. By  applying Lemma \ref{lemma: Wurzel} with $j=5, m= \eta$ and $x = \mathcal{E}(C_n)$ we conclude $\mathcal{E}(C_n) \geq - \beta(n)+1 > -\lfloor \beta(n)\rfloor$. This contradicts Theorem \ref{TheoremGroundstates}.
 
\noindent\emph{Step 2: $\mathcal{A}_{\mathrm{bulk}}(X_n)=\emptyset$.} Assume by contradiction that $\mathcal{A}_{\mathrm{bulk}}(X_n) \neq\emptyset$. Note that  $\eta \ge 4$ by  Lemma \ref{lemma : non-eq-atoms} and thus  $n \ge 8$. There are two cases two consider: a) $n-d \leq 3$ and b) $n-d \geq 4$.  

\noindent\emph{Proof for $n-d \le 3$.}  Using Lemma \ref{LemmaBoundaryestimate} we have that 
\begin{align}\label{ineq : case a}
2b+ 4 - 3n + \eta =n-d \leq 3.
\end{align}
By Theorem \ref{TheoremGroundstates} we have that $\mathcal{E}(C_n) = -b  = -\lfloor 2n - 2\sqrt{n}\rfloor  $. This together with (\ref{ineq : case a}) leads to 
\begin{align}\label{ineq : case a 2}
-2n -\lfloor -2\sqrt{n}\rfloor=  -\lfloor 2n - 2\sqrt{n}\rfloor =  \mathcal{E}(C_n) \geq -\frac{3}{2}n + \frac{1}{2}(1+\eta)  \geq -\frac{3}{2}n + \frac{5}{2}, 
\end{align}
 where the last step follows from $\eta \ge 4$.  Inequality (\ref{ineq : case a 2}) is violated for all  $n \geq 8$. This yields a contradiction and  concludes the proof in this case.

\noindent \emph{Proof for $n-d \ge 4$.}  
 We first observe that $C_n^{\mathrm{bulk}}$ is clearly not a subset of the square lattice if $\mathcal{A}(X_n^\mathrm{bulk}) \neq \emptyset$.  In addition, we notice that $\mathcal{A}_{\mathrm{bulk}}(X_n) \setminus\mathcal{A}(X_n^\mathrm{bulk}) \neq \emptyset$ implies that $C_n^{\mathrm{bulk}}$ is not subset of the square lattice. To see this, assume that there exists $x \in \mathcal{A}_{\mathrm{bulk}}(X_n) \setminus\mathcal{A}(X_n^\mathrm{bulk})$ and denote by $y \in \partial X_n$ one of its neighbors. If $y$ is $4$-bonded, \eqref{eq: remark-equ} and $x \in \mathcal{A}(X_n)$ yield that $C_n^\mathrm{bulk}$ is not a subset of the square lattice.   On the other hand, if $y$ is $3$-bonded, \eqref{eq: remark-equ} together with  a simple geometric consideration implies that $C_n$ is not repulsion-free, which contradicts Remark \ref{rem: repulsionsfree}.

First, assume that the bond graph of $C_n^{\mathrm{bulk}}$ does not contain a flag.   By induction hypothesis, ground states with less than $n$ atoms  without flags are subsets of the square lattice. Thus, by the assumption $\mathcal{A}_{\mathrm{bulk}}(X_n) \neq\emptyset$, $C_n^{\mathrm{bulk}}$ cannot be a ground state.  This along with \eqref{eq: bnd2} and  \eqref{eq: several times}  \EEE yields  \EEE
\begin{align}\label{eq: strici}
 \mathcal{R}^{\mathrm{bulk}}(C_n) = \mathcal{R}(C_n^{\mathrm{bulk}}) = \mathcal{E}(C_n^{\mathrm{bulk}}) > -\lfloor \beta(n-d)\rfloor.
\end{align}
By Lemma \ref{LemmaBoundaryEnergy} we  also have $\mathcal{R}^{\mathrm{bnd}}(C_n) \geq -2d + 4$. The two inequalities  together with \eqref{eq: bnd2}  lead to the strict inequality \eqref{eq: prev-inequ}, and we may proceed exactly as in Step 1 of the proof to  obtain a contradiction to the fact that $C_n$ is a ground state.

 Now assume that the bond graph of $C_n^{\mathrm{bulk}}$  contains a flag. Without restriction we can assume that $C_n^{\mathrm{bulk}}$  is a ground state since otherwise the strict inequality \eqref{eq: strici} holds, and we obtain a contradiction exactly as before. Since $n-d \geq 4$, the bond graph of $C_n^{\mathrm{bulk}}$ contains exactly one flag by Lemma \ref{lemma:flags}. \EEE   There holds  $\#\mathcal{A}(X_n) \geq 2$ by Lemma \ref{lemma : non-eq-atoms} and $\mathcal{A}(X_n)\setminus \mathcal{A}_{\mathrm{bulk}}(X_n) =\emptyset$  by Step 1. Therefore, $\#\mathcal{A}_{\mathrm{bulk}}(X_n) \geq 2$. After removing the flag from $C_n^{\mathrm{bulk}}$, \EEE we get a configuration $C_{n-d-1}$ which does not contain flags,  but there still holds  $\mathcal{A}(X_n) \cap X_{n-d-1} \neq \EEE \emptyset$.   By induction hypothesis, $C_{n-d-1}$ can thus not be a ground state and we get $\mathcal{E}(C_{n-d-1})  > \EEE -\lfloor \beta(n-d-1)\rfloor$. By  counting the contribution of the flag to the energy by at least $-1$  we obtain 
 $$\mathcal{R}(C_n^{\mathrm{bulk}}) = \mathcal{E}(C_n^{\mathrm{bulk}}) > -1 -\lfloor \beta(n-d-1)\rfloor,$$
where we also used \eqref{eq: several times}. This along with  $\mathcal{R}^{\mathrm{bnd}}(C_n) \geq -2d + 4$ (see Lemma \ref{LemmaBoundaryEnergy}) and \eqref{eq: bnd2}  gives \EEE 
\begin{align*}
\mathcal{E}(C_n)  > -1 - 2d + 4 -\lfloor \beta(n-d-1)\rfloor =   - \lfloor 2n  - 2\sqrt{n-d-1}\rfloor +5. 
\end{align*}
Since $\mathcal{E}(C_n) = - b$ is an integer by Theorem \ref{TheoremGroundstates}, \EEE we obtain $\mathcal{E}(C_n) \geq  -2n +6 + 2\sqrt{n-d-1}$. Now using Lemma \ref{LemmaBoundaryestimate} and $\mathcal{E}(C_n) =  -b$ \EEE we obtain
\begin{align*}
\mathcal{E}(C_n) \geq -2n +6 + 2\sqrt{-2\mathcal{E}(C_n) +3 -3n +\eta}.
\end{align*}
 Recall that  $\eta \geq 4$  by Lemma \ref{lemma : non-eq-atoms}. By \EEE applying Lemma \ref{lemma: Wurzel} with $j=6, m=\eta-1$ and $x =\mathcal{E}(C_n)$ we obtain
\begin{align*}
\mathcal{E}(C_n) &\geq -2n +2\sqrt{\eta +n -5} +2 \geq -2n + 2\sqrt{n-1} +2 \\&= -2n +2\sqrt{n} -2\sqrt{n}+2\sqrt{n-1}+2 = -2n + 2\sqrt{n} + 2 - \frac{2}{\sqrt{n}+\sqrt{n-1}}.
\end{align*}
 Recall that  $n \geq 8$. Therefore, \EEE we get $2 - \frac{2}{\sqrt{n}+\sqrt{n-1}} \geq 1$ and thus
\begin{align*}
\mathcal{E}(C_n) \geq -2n + 2\sqrt{n}+1 > -\lfloor 2n -2\sqrt{n}\rfloor.
\end{align*}
This  contradicts \EEE the fact that $C_n$ is a ground state and concludes the proof. \EEE
\end{proof}

 We now address the proof of Theorem \ref{theorem : GeometryGroundstatewithout}(b). To this end, we use the following result.

\begin{theorem}[Deviation from Wulff-shape]\label{theorem:deviation} Let  $n \in \mathbb{N}$   and let $C_n$ be a ground state with no acyclic bonds. Then, possibly after translation we find two squares $S_{k_1^2} \subset \mathbb{Z}^2$ and $  S_{k_2^2} \subset \mathbb{Z}^2$,  $k_1,k_2 \in \mathbb{N}$,  with $S_{k_1^2} \subset X_n \subset S_{k_2^2}$ such that
\begin{align*}
0< k_2-k_1 \leq cn^{\frac{1}{4}},
\end{align*}
where $c>0$ is a universal constant independent of $n$ and $C_n$. 
\end{theorem}

\begin{proof} 
Let $n \ge 4$.  From Theorem \ref{theorem : GeometryGroundstatewithout}(a) and the fact that $C_n$ does not contain acyclic bonds we get that $C_n$ is a subset of the square lattice. Moreover, $C_n$ is repulsion-free, see Remark \ref{rem: repulsionsfree}. Therefore, our energy on ground-state configurations coincides (up to distributing alternating charge or neglecting it)  with the one considered in \cite{Mainini-Piovano}. The statement then follows from \cite[Theorem 8.1]{Mainini-Piovano}.
\end{proof}

\begin{proof}[Proof of Theorem \ref{theorem : GeometryGroundstatewithout}(b)]   By Theorem \ref{TheoremGroundstates} and  Lemma \ref{lemma:flags} we know that the bond graph contains at most one flag and  no other acyclic bonds. If it contains a flag, one can remove the flag and the remaining configuration is still a ground state   by  Lemma \ref{LemmaPropertiesbeta}(i). We can therefore assume that there are no acyclic bonds in the bond graph of $C_n$.  
The statement  follows from Theorem \ref{theorem:deviation}. In fact, observe that Theorem \ref{theorem:deviation} also implies that  the diameter of a ground  state  is of order $\sqrt{n}$. 
\end{proof}

\subsection{Characterization of the net charge}\label{subsection:netcharge}

We now prove Theorem \ref{theorem : charge}. Part (b) of the statement has already been addressed by an explicit construction in  Subsection  \ref{subsection:sqare rectangle}.  Thus, it remains to prove part  (a). \EEE 

\begin{definition}[Line segment] A tuple $(x_0,\ldots,x_m)\subset \mathbb{Z}^2$ is called a \textit{line segment} if there exists $v \in \{(1,0),(0,1)\}$ such that $x_{k+1}=x_k + v$ for all $k \in \{0,\ldots,m-1\}$.
\end{definition}

\begin{proposition}[Convexity of ground states]\label{prop:convexity} Let  $n \in \mathbb{N}$     and let $C_n$ be a ground state with no acyclic bonds. Then each line segment $(x_0,\ldots,x_m)$ with $x_0,x_m \in X_n$ satisfies $x_i \in X_n$ for $i=0,\ldots,m$.  
\end{proposition}

\begin{proof}
 For the proof we refer to \cite[Proposition 6.3]{Mainini-Piovano}, where this property is called  \emph{convexity by rows and columns}. The result in  \cite[Proposition 6.3]{Mainini-Piovano} is applicable due to the fact that our energy on ground-state configurations coincides   with the one considered in \cite{Mainini-Piovano}, see the proof of Theorem \ref{theorem:deviation}.  \EEE
\end{proof}

\begin{proof}[Proof of Theorem \ref{theorem : charge}(a)]   In view of Theorem \ref{TheoremGroundstates} and Lemma \ref{lemma:flags}, it suffices to treat the case that the bond graph of $C_n$ does not contain any acyclic bonds. We apply Theorem \ref{theorem:deviation} to find two squares $S_{k_1^2} \subset X_n \subset S_{k_2^2}$.   By Proposition \ref{prop:convexity} it is elementary to see that $X_{n} \setminus S_{k_1^2}$ can be written as the union of at most $4(k_2-k_1)$ line segments. Recall that $C_n$ has alternating charge distribution and therefore the net charge of each line segment is in $\{-1,0,1\}$. Also recall from  Subsection  \ref{subsection:special subsets} that squares have charge in $\{-1,0,1\}$. This implies that the net charge of the configuration $C_n$ satisfies
\begin{align*}
|\mathcal{Q}(C_n)|\leq 4(k_2-k_1) +1.
\end{align*}
The statement follows from the fact that $k_2-k_1 \leq cn^{1/4}$, see Theorem \ref{theorem:deviation}.
\end{proof}

 \section{Characterization of optimal configurations for Prescribed charge}\label{sec: prescribed}

 In this section we prove  Proposition \ref{theorem : qsat}, Theorem \ref{theorem: min-en2}, and Theorem \ref{theorem : GeometryGroundstate}.

\subsection{Energy of optimal \EEE configurations}

 In this short subsection we prove  Theorem \ref{theorem: min-en2}. 
 
\begin{proof}[Proof of Theorem \ref{theorem: min-en2}:]  Without restriction we may suppose that $q_{\rm net} \ge 0$ since the proof for $q_\mathrm{net}\leq 0$ follows analogously.   We proceed in two steps. First, we prove the statement for $q_\mathrm{net}\geq q_\mathrm{sat}^n$ and then   for $ 0 \le q_\mathrm{net} < q_\mathrm{sat}^n$.  The proof is performed by an induction argument.

\noindent\emph{Step 1: $q_\mathrm{net} \geq q_\mathrm{sat}^n$.}   Our goal is to show
\begin{align}\label{eq: first-indu}
\mathcal{E}^n_\mathrm{min}(q_\mathrm{net}) =  -2n +2q_\mathrm{net} \ \ \ \text{for all } \  q_\mathrm{net} \geq q_\mathrm{sat}^n. 
\end{align}
For $q_{\rm net} = q_\mathrm{sat}^n$ the statement is clearly true, see \eqref{eq: qsat-def}.  Suppose that the statement holds for $q_\mathrm{net}  \ge  q_\mathrm{sat}^n  $ with $q_\mathrm{net} \leq n-2$. We show \eqref{eq: first-indu} for $q_\mathrm{net} + 2$. \EEE By Lemma \ref{lemma : Charge Energybound}(a), for all configurations $C_n$  satisfying  $\mathcal{Q}(C_n)=q_\mathrm{net}+2$ there holds $\mathcal{E}(C_n) \geq -2n + 2(q_\mathrm{net}+2)$.

It thus  suffices to construct a configuration  $\tilde{C}_n$ with $\mathcal{Q}(\tilde{C}_n)=q_\mathrm{net}+2$ and  $ \mathcal{E}(\tilde{C}_n) = -2n +2(q_\mathrm{net}+2)$.   To this end, let $C_n$ be  a configuration  with $\mathcal{Q}(C_n) = q_\mathrm{net}$ and $\mathcal{E}(C_n) =-2n +2q_\mathrm{net}$,  which exists by the induction hypothesis.

Choose $x_i \in C_n$  with  $q_i =-1$ and modify $C_n$ as follows:  remove $x_i$   and add a new atom $\tilde{x}_i$ with  charge  $+1$   to the configuration such that $|\tilde{x}_i-x_j| \geq \sqrt{2}$ for all $j=1,\ldots,n$. Denote this configuration by $\tilde{C}_n$.  We have $\mathcal{Q}( \tilde{C}_n \EEE ) = q_\mathrm{net}+2$. By Lemma \ref{lemma : Charge Energybound}(a)(i), the  relocated  atom was $4$-bonded  and the bond lengths were of unit length.  Hence, we obtain by [ii] 
 \begin{align*}
  \mathcal{E}(\tilde{C}_n)  = \mathcal{E}(C_n)+4 = -2n +2q_\mathrm{net} +4 = -2n +2(q_\mathrm{net}+2).
 \end{align*}
Thus, \eqref{eq: first-indu}   is proven for $ q_{\rm net } + 2$.

\noindent\emph{Step 2:  $ 0 \le \EEE q_\mathrm{net}  < \EEE  q_\mathrm{sat}^n$.}   As before, the lower bound in \eqref{eq: charge energy} follows from Lemma \ref{lemma : Charge Energybound}(a). Therefore, it remains to show that   
\begin{align}\label{eq: first-indu2}
\mathcal{E}^n_\mathrm{min}(q_\mathrm{net}) \leq -2n + 4q_\mathrm{sat}^n-2q_\mathrm{net}  \ \ \ \  \text{for all }  0 \le q_\mathrm{net} \le q_\mathrm{sat}^n.  
\end{align}
 For $q_{\rm net} = q_\mathrm{sat}^n$ the statement is clearly true,  see \eqref{eq: qsat-def}. \EEE  Assume that \eqref{eq: first-indu2} holds for $2\leq q_\mathrm{net} \le q_\mathrm{sat}^n$.  We show \eqref{eq: first-indu2} for $q_\mathrm{net} - 2$. To this end, we have to construct a configuration  $\tilde{C}_n$ with   $\mathcal{Q}(\tilde{C}_n) = q_\mathrm{net}-2$ and  $ \mathcal{E}(\tilde{C}_n) \le  -2n+ 4q_\mathrm{sat}^n-2(q_\mathrm{net}-2)$. 

Let $C_n$ be a configuration  with  $\mathcal{Q}(C_n) = q_\mathrm{net}$ and  $\mathcal{E}(C_n) = \mathcal{E}^n_\mathrm{min}(q_\mathrm{net})$.   Choose \EEE an atom $x_i$  with  $q_i=1$ and modify $C_n$ in the following way: remove $x_i$   and add a new atom $\tilde{x}_i$ with  charge $-1$    to the configuration such that $|\tilde{x}_i-x_j| \geq \sqrt{2}$ for all $j=1,\ldots,n$. Denote this configuration by $\tilde{C}_n$. First, note that $\mathcal{Q}( \tilde{C}_n  ) =q_\mathrm{net}-2$.  By [ii] and   $\#\mathcal{N}(x_i) \leq 4$, see \eqref{NeighbourhoodCard}, we get 
\begin{align*}
\mathcal{E}(\tilde{C}_n)  \leq \mathcal{E}(C_n) +4.
\end{align*}
As \eqref{eq: first-indu2} holds for $q_\mathrm{net}$ by induction hypothesis,  we obtain
\begin{align*}
\mathcal{E}_\mathrm{min}^n(q_\mathrm{net}-2) \leq  \mathcal{E}(\tilde{C}_n) \le  \mathcal{E}(C_n)  +4  \leq -2n +  4q_\mathrm{sat}^n-2q_\mathrm{net} + 4  =  -2n +   4q_\mathrm{sat}^n -2(q_\mathrm{net}-2).
\end{align*}
This shows \eqref{eq: first-indu2} for $q_{\rm net}-2$ and  concludes the proof.
\end{proof}

\subsection{Boundary and interior net charge}\label{sec: boundary charge} 
In the following, we will consider without restriction configurations which satisfy $\mathcal{Q}(C_n) \ge 0$, i.e., the $+1$ phase is the majority phase. This can indeed always be achieved by interchanging the roles of the positive and negative charges. We define the notion of a {bridging atom}. An atom $x_0$ is called \emph{bridging atom} if it is $2$-bonded and not contained in any simple cycle, see Fig.~\ref{Fig : bridging atom}.  Denote by $C_m$ and $C_{n-m-1}$ the two connected components of $C_n \setminus \{x_0,q_0\}$. We set 
 \begin{align}\label{eq: bridgi}
C_{m+1} = C_{m} \cup \{(x_0,q_0)\}, \ \ \ \  C_{n-m} = C_{n-m-1} \cup \{(x_0,q_0)\}.
\end{align} 
 In the following we say that $C_m$ and $C_{m-n}$ are \emph{connected through a bridging atom}. Recall the definition of acyclic bonds in Subsection \ref{sec: notions}.  In a similar fashion, we say that an atom $x_i$ is an \emph{acyclic atom} if it is not contained in any cycle. We denote  the union of the  acyclic atoms by $I_{\rm ac}$.  
 \EEE

\begin{figure}[htp]
\begin{tikzpicture}[scale=0.5]

\begin{scope}[rotate=90]
\draw[ultra thin](0,2)--++(90:2);
\draw[ultra thin](0,2)++(90:1)++(180:1)--++(0:2);
\draw[ultra thin](0,2)++(90:1)++(180:1)++(0:2)--++(90:2);
\draw[ultra thin](0,2)++(90:2)--++(0:2);

\foreach \j in {0,...,2}{
\draw[ultra thin](\j,-2+\j)--++(0,4-2*\j);
\draw[ultra thin](-\j,-2+\j)--++(0,4-2*\j);
\draw[ultra thin](\j-2,\j)--++(4-2*\j,0);
\draw[ultra thin](\j-2,-\j)--++(4-2*\j,0);
\draw[fill=black](0,2)++(90:1)++(90*\j:1)circle(.07);
\draw[fill=black](0,2)++(90:2)++(0:1)++(90*\j:1)circle(.07);
}

\draw[fill=white](0,2)++(90:1) circle(.07);
\draw[fill=white](0,2)++(90:2)++(0:1)circle(.07);

\draw[fill=black,black](0,2)circle(.1);

\foreach \j in {0,2}{
\pgfmathsetmacro\y{\j-1}
\foreach \l in {0,...,\y}{
\foreach \k in {0,...,3}{
\begin{scope}[rotate = \k*90]
\draw[fill=black](0,0)++(\j,0)++(-\l,\l) circle(.07);
\end{scope}
}
}}

\foreach \j in {1}{
\pgfmathsetmacro\y{\j-1}
\foreach \l in {0,...,\y}{
\foreach \k in {0,...,3}{
\begin{scope}[rotate = \k*90]
\draw[fill=white](0,0)++(\j,0)++(-\l,\l) circle(.07);
\end{scope}
}
}}

\end{scope}
\begin{scope}[shift={(7,0)}]
\draw(0,0) node[anchor =north]{$x_j$};
\draw(-.75,0) node[anchor =north east]{$x_i$};
\draw(.8,0) node[anchor =north west]{$x_k$};
\foreach \j in {0,...,3}{
\draw[fill=black](1,0)++(90*\j:1) circle(.07);
\draw[ultra thin](1,0)--++(90*\j:1);
\draw[fill=black](-1,0)++(90*\j:1) circle(.07);
\draw[ultra thin](-1,0)--++(90*\j:1);
\draw[fill=black](0,1)++(90*\j:1) circle(.07);
\draw[ultra thin](0,1)--++(90*\j:1);
}
\draw[fill=white](1,0) circle(.07);
\draw[fill=white](0,1) circle(.07);
\draw[fill=white](-1,0) circle(.07);
\end{scope}
\end{tikzpicture}
\caption{On the left: A configuration with a bridging atom illustrated bold. On the right: A configuration with $I_\mathrm{ac} \cap X_n^-=\emptyset$,  where the atoms $X_n^-$ defined in \eqref{def : Cnpm} are illustrated in white.  }
\label{Fig : bridging atom}
\end{figure}
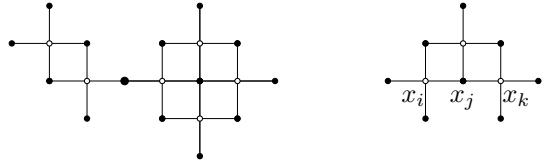

\begin{lemma}[Bridging atom]\label{lemma: bridging}
 Consider  $C_n$ with   $0 \le \mathcal{Q}(C_n)$   and $\mathcal{E}(C_n)  = -2n + 2\mathcal{Q}(C_n)$.  

\noindent (a) If there exists a bridging atom connecting two configurations $C_m$ and $C_{n-m}$, there holds  
\begin{align*}
\mathcal{Q}(C_m) \geq q_\mathrm{sat}^{m+1}-1 \ \ \  \text{ and } \ \ \  \mathcal{Q}(C_{n-m}) \geq q_\mathrm{sat}^{n-m}.
\end{align*} 
Moreover, $C_{n-m}$ is an optimal configuration  and there holds  $m, n-m \ge 4$. \EEE

\noindent (b)  If $n \geq 6$ and  $C_n$ is connected without bridging atoms,  then $I_\mathrm{ac}\cap X_n^- =\emptyset$  and $ I_\mathrm{ac}\cap X_n^+$ consists of $1$-bonded atoms. 
\end{lemma}

\begin{proof}
(a) Suppose that a bridging atom  $x_0$ exists. As $\mathcal{E}(C_n) = -2n + 2\mathcal{Q}(C_n)$, Lemma \ref{lemma : Charge Energybound}(a)(i) implies that $x_0$ has charge $q_0 = +1$. This clearly gives
\begin{align}\label{eq : charge added}
\mathcal{Q}(C_{m+1}) = \mathcal{Q}(C_m) +1.
\end{align}
Define by $C_{n+1}$  a configuration given by the union of the two connected components $C_{n-m}$ and $C_{m+1}$, where $C_{m+1}$ is translated in such a way that   $\mathrm{dist}(X_{n-m},X_{m+1})\geq \sqrt{2}$. We observe that the atoms  of charge $-1$ are $4$-bonded  by construction since $C_n$ satisfies this property. In a similar fashion, $C_{n+1}$ is repulsion-free and its bonds have unit length by Lemma \ref{lemma : Charge Energybound}(a)(i),(ii).    Therefore, in view of  Lemma \ref{lemma : Charge Energybound}(b), for $C_{n+1}$ there holds the equality  $\mathcal{E}(C_{n+1}) = -2(n+1) + 2\mathcal{Q}(C_{n+1})$. Then applying Lemma \ref{lemma : Charge Energybound}(a)(vi) on the two connected components $C_{m+1}$ and $C_{n-m}$ we find  \EEE $ \mathcal{Q}(C_{m+1}) \geq q_\mathrm{sat}^{m+1}$ and   $\mathcal{Q}(C_{n-m}) \geq q_\mathrm{sat}^{n-m}$. The first part of the claim now follows from (\ref{eq : charge added}). 

In view of  Lemma \ref{lemma : Charge Energybound}(b), for $C_{n-m}$ there holds the equality  $\mathcal{E}(C_{n-m}) = -2(n-m) + 2\mathcal{Q}(C_{n-m})$. Thus, $C_{n-m}$ is an optimal configuration  by \eqref{ineq : charge}.   Finally, since atoms of charge $-1$ are $4$-bonded, each of the two components contains at least four atoms, i.e.,  $m,n-m \ge 4$. \EEE

(b) Now assume that  $n \geq 6$ and that  $C_n$ is connected without bridging atoms.  By $\mathcal{E}(C_n)  = -2n + 2\mathcal{Q}(C_n)$ and  \eqref{ineq : charge} we have that $C_n$ is optimal. \EEE  First, \EEE suppose  by contradiction that there exists some  $x_i \in I_\mathrm{ac}\cap X_n^- \neq \emptyset$.   Since $C_n$ is connected with alternating charge distribution (see  Lemma \ref{LemmaNeighborhood}) \EEE and $n \ge 6$, there exist $x_k \in X_n^-$, $x_k \neq x_i$, and $x_j \in \mathcal{N}(x_i) \cap \mathcal{N}(x_k)$.   Since $C_n$ does not contain bridging atoms  and $x_i \in I_{\rm ac}$,  $x_j$ is at least $3$-bonded.  By Lemma \ref{lemma : Charge Energybound}(a)(iv),(v)  all bond angles of $x_j$ are integer multiples of $\frac{\pi}{2}$.  Thus, as sketched  in  Fig.~\ref{Fig : bridging atom}, $x_i$ is contained in a square.  This contradicts $x_i \in I_\mathrm{ac}$. In a similar fashion, Fig.~\ref{Fig : bridging atom} shows that some $x_j \in X_n^+$ which is $3$-bonded or $4$-bonded is contained in a square. This along with \eqref{NeighbourhoodCard} and the fact that  by assumption \EEE there are no  bridging atoms, i.e., no \EEE  $2$-bonded atoms in $I_{\rm ac}$, shows that $I_{\rm ac} \cap X_n^+$ consists of $1$-bonded atoms only.  \EEE
\end{proof}

\textbf{Sub-configurations:} Similar to the definition of $I_{\rm ac}$, we say that an atom $x_i$ is an \emph{exterior acyclic atom} if it is not contained in any cycle and not contained in the interior region of any cycle. We denote the union of the exterior acyclic atoms by $I_{\rm ac}^{\rm ext}$. Clearly, there holds  $I_{\rm ac}^{\rm ext} \subset I_{\rm ac}$.

Let  $n\geq 6$  and let $C_n$ be a connected,  optimal configuration without bridging atoms satisfying $\mathcal{Q}(C_n) \ge q_{\rm sat}^n$. In particular, there holds $\mathcal{E}(C_n)  = -2n + 2\mathcal{Q}(C_n)$ by  Theorem \ref{theorem: min-en2}. \EEE  We denote by $C_n^a$ the configuration without the exterior acyclic  atoms and their charges. \EEE We observe that  $C_n^a$ is still connected. This follows from Lemma \ref{lemma: bridging}(b). By $\partial C_n^{a}$ we indicate the maximal polygon of $C_n^a$, i.e,  the simple cycle which delimits the bond graph, see also Subsection \ref{section:boundary energy}. \EEE   The cardinality of   $\partial C_n^{a}$ is denoted by $d$.  Furthermore,  for $k=2,3,4$, \EEE  we set
\begin{align}\label{def: Ik}
I_k = \{x_i \in \partial X_n^a : \ \  \#(\mathcal{N}(x_i) \cap X_n^a)=k\}.
\end{align}
 Finally, we set $C_n^{a,\mathrm{bulk}} = C_n^a \setminus \partial C_n^{a}$. In a similar fashion, we denote by $X_n^a$, $\partial X_n^{a}$, and $X_n^{a,\mathrm{bulk}}$ the atomic positions of the sub-configurations.  Recall the definition of $\eta$ in \eqref{Excess}. \EEE

 \EEE

\begin{lemma}[Cardinality of  $n-d$\EEE] \label{lemma : card positively charged shell} 
Let  $n\geq 6$ and let  $C_n$ be a connected, optimal configuration which satisfies $\mathcal{Q}(C_n)\geq q_\mathrm{sat}^n$ and   does not contain a bridging atom. Then
\begin{align*}
n- d  \in \EEE n + 4-4  \mathcal{Q}(C_n)  + \eta+ 2\# I_{\rm ac}^{\rm ext}  + 2\mathbb{N}_0. \EEE  
\end{align*}
\end{lemma}
\begin{proof}
Let $f_j$ be the number elementary $j$-gons in the bond graph and let $f$ be the number of elementary polygons.    There holds
\begin{align*}
\sum_{j\geq 4} j f_j =   2b-d-2  b_\mathrm{ac},  
\end{align*} 
where $b_\mathrm{ac}$ denotes the cardinality of the acyclic bonds. In fact,  by the summation on the left the bonds of the maximal polygon are  counted once, the  acyclic bonds  are not counted, and all other bonds are counted twice.  By Lemma \ref{RemarkPolygon}(a)      we obtain $4f + \eta = 2b-d - 2 b_{\rm ac}$. This along with Euler's formula $n-b+f=1$ (omitting the exterior face)  and the fact that $b_{\rm ac} \ge \# I_{\rm ac}^{\rm ext}$ \EEE yields 
\begin{align*}
n-d  \in \EEE  2b+4 -  3n \EEE +\eta+ 2\# I_\mathrm{ac}^{\rm ext}  +2\mathbb{N}_0. \EEE
\end{align*}
As $b=2n-2  \mathcal{Q}(C_n)  $, see   Remark \ref{rem: repulsionsfree}, Theorem \ref{theorem: min-en2}, and Lemma \ref{lemma : Charge Energybound}(a)(i),(ii),   the claim follows. \EEE 
\end{proof}

Recall the definition of the non-equilibrated atoms $\mathcal{A}(X_n)$ in  Subsection  \ref{sec: notions}. The following estimate for the  net \EEE charge of the boundary and the interior configuration will be instrumental for our analysis. Its proof will be given in Appendix \ref{appendix}, along with the proof of the boundary energy estimate (Lemma \ref{LemmaBoundaryEnergy}).


\begin{lemma}[Net charge of the interior and boundary]\label{lemma : charge interior} Let $n \geq 6$ and let $C_n$ be a connected $q_\mathrm{sat}^n$-optimal configuration without bridging atoms. 

 \noindent (a) Then there exists $m \geq 4$ such that
\begin{align}\label{eq :I2I3I3}
\#I_2 +2\#I_3+3\#I_4=2d-m. 
\end{align}
If  $\#\mathcal{A}(X_n) \geq \#\mathcal{A}(X_n^{\mathrm{a},\mathrm{bulk}})+2$, then $m \geq 6$.

\noindent (b) Suppose that \eqref{eq :I2I3I3} holds. \EEE Then $\#I_{\mathrm{ac}}^{\mathrm{ext}} \geq m$ and there exists an optimal configuration $C_{n-d-m}$ satisfying $\mathcal{Q}(C_{n-d-m}) \geq q_\mathrm{sat}^{n-d-m}$ and 
\begin{align}\label{ineq : charge-new}
\mathcal{Q}(C_{n}) = \mathcal{Q}(C_{n-d-m})+m.
\end{align}
 Moreover, $X_{n-d-m}$ is a subset of $X_n$ up to $0$-bonded atoms and \EEE $\mathcal{A}(X_{n-d-m}) = \mathcal{A}(X_n^{\mathrm{a},\mathrm{bulk}} )$.
\end{lemma}
\EEE
%
%

Roughly speaking, the configurations  $C_{n-d-m}$ \EEE are constructed by removing  $\partial X_n^{a}$  and $ I_\mathrm{ac}^{\rm ext}$ \EEE  from $C_n$. In this sense,   $\mathcal{Q}(C_{n-d-m})$   can be regarded as the net charge of the interior.  The statement then shows that the boundary net charge can be controlled from below by $4$, and  by at least 6 \EEE if  the number of non-equilibrated atoms decreases by at least 2 when the boundary is removed. \EEE 


\subsection{Characterization  of $q_\mathrm{sat}^n$}\label{subsection : calculation qsat}

 The goal of this subsection is to prove  Proposition  \ref{theorem : qsat}. In view of the construction in Subsection \ref{subsection : upper bound qsat}, see Proposition \ref{proposition : upper bound}, it remains to show the lower bound $q_\mathrm{sat}^n \geq \phi(n)$, where  the function $\phi$ is defined in \eqref{def: phi}.  As a preparation, we first  provide an equivalent representation of $\phi$, and state some  monotonicity and subadditivity properties.  

\begin{lemma}[Representation of $\phi$] For $n \ge 2$ there holds
\begin{align}\label{eq: different repr}
\phi(n) = \begin{cases}
2 +  2k  + 0  & \ \ \ \text{if } n=1+2k^2+2k+m, \ \ \  1\leq m \leq 2k+2, \ m \text{ odd},\\
2+ 2k +  1 &\ \ \ \text{if } n=1+2k^2+2k+m,  \ \ \  1\leq m \leq 2k+2, \ m \text{ even},\\
3+ 2k + 1  &\ \ \ \text{if } n=1+2k^2+2k+m, \ \ \   2k+3\leq m \leq 4k+4, \  m \text{ odd},\\
3+ 2k + 0 &\ \ \ \text{if } n=1+2k^2+2k+m, \ \ \  2k+3\leq m \leq 4k+4, \ m \text{ even},\\
\end{cases}
\end{align}
where $k \in \mathbb{N}_0$.
\end{lemma}
\EEE

\begin{proof}
We divide the proof into two steps. First, we prove the statement in the case that $n$ is even and then in the case that $n$ is odd.

\noindent \emph{$n$ even.} Let $n =1 + 2k^2+2k+m$ with $n$ even, i.e., $1 \leq m \leq 4k+4$ and $m$ odd. By (\ref{def: phi}) we have 
\begin{align}\label{eq: first-phi}
\phi(n) = 2 + 2 \lfloor \tfrac{1}{2}\sqrt{2n-4}  \rfloor = 2 + 2 \lfloor \tfrac{1}{2}\sqrt{4k^2+4k+2m-2} \rfloor.
\end{align}
It is elementary to check that  
\begin{align*}
\lfloor \tfrac{1}{2}\sqrt{4k^2+4k+2m-2} \rfloor & = k  \  \ \ \ \ \ \ \   \text{ for all }  1 \leq m \leq 2k+2, \text{ $m$ odd},\\
\lfloor \tfrac{1}{2}\sqrt{4k^2+4k+2m-2} \rfloor & = k+1 \  \ \ \text{ for all }  2k+3 \leq m \leq 4k+4,  \text{ $m$ odd}.
\end{align*}
This along with  \eqref{eq: first-phi} shows  the desired equality of  $\phi(n)$ with the expression  on the right hand side of (\ref{eq: different repr}), in the case that $n$ is even. \EEE

%
%
%
%

\noindent \emph{$n$ odd.} Let $n =1 + 2k^2+2k+m$ with $n$ odd, i.e., $1 \leq m \leq 4k+4$ and $m$ even. By (\ref{def: phi}) we have 
\begin{align}\label{eq: first-phi2}
\phi(n) = 3 + 2 \lfloor-\tfrac{1}{2}+ \tfrac{1}{2}\sqrt{2n-5}  \rfloor = 3 + 2 \lfloor-\tfrac{1}{2}+ \tfrac{1}{2}\sqrt{4k^2+4k+2m-3}\rfloor.
\end{align}
It is again  elementary to see that 
\begin{align*}
\lfloor-\tfrac{1}{2}+ \tfrac{1}{2}\sqrt{4k^2+4k+2m-3}\rfloor& = k  \  \ \ \ \ \ \ \   \text{ for all }  1 \leq m \leq 4k+4,  \text{ $m$ even}.
\end{align*}
This together with  \eqref{eq: first-phi2} shows  the desired equality of  $\phi(n)$ with the expression  on the right hand side of (\ref{eq: different repr}), in the case that $n$ is odd. \EEE
\end{proof}

\begin{lemma}[Properties of $\phi$] \label{lemma : properties phi}
The following  properties  hold true:
\begin{itemize}
\item[(i)] $\phi(n) \leq \phi(n-1)+1$ and $\phi(n) \leq \phi(n+1)+1$ for all $n \in \mathbb{N}$,
\item[(ii)] $\phi(n)\leq  \phi(n+2)$ for all $n \in \mathbb{N}$,
\item[(iii)] $\phi(n) \leq \phi(n-3)+1$ for all $n \geq 5$,
\item[(iv)]   $\phi(n) \leq \phi(n-5)+1$ for all $n \geq 11$, 
\item[(v)]   $\phi(n_1-t) + \phi(n_2+t) \le \phi(n_1) + \phi(n_2) +  6  $ for all $t \le n_1 \le n_2$. 
\item[(vi)] $\phi(n)+2 \leq \phi(m) +\phi(n-m)$ for all $3\leq m,n-m\leq n,  \ \  n\ge11$, 
\item[(vii)] There exists $m_0 \in \mathbb{N}$ such that $\phi(n)+   52  \leq \phi(m) +\phi(n-m)$ for all $m_0 \leq m,n-m \leq n$.
\end{itemize}
\end{lemma}

\begin{proof}
Properties (i)-(v) are elementary and can be checked by using \eqref{eq: different repr}: for (ii), we use the monotonicity of $\phi$ when restricted to even and odd numbers, respectively. Property (i) follows by looking closely at \eqref{eq: different repr}. To see  (iii)-(iv),  we  denote by $k(n) \in \mathbb{N}_0$ and $1 \le m(n)  \le 4k(n)+4$ the numbers such that $n = 1 + 2(k(n))^2 + 2k(n) + m(n)$. We also set  $n_3 = 7$ and $n_5=15$. \EEE It is elementary to check that for $t \in \lbrace 3, 5 \rbrace$ and for all $n \ge n_t$ there holds    either   $k(n) = k(n-t)$ or
 $$\text{$k(n) = k(n-t)+1$, \ \ \ \ \  $m(n) \le 2k(n)+2$, \ \ \ \ \  $m(n-t) \ge 2k(n-t)+3$}.$$
This along with a careful inspection of  \eqref{eq: different repr} implies (iii)-(iv) for $n\geq 7$ and $n \geq 15$ respectively. The remaining cases can be checked by Table~\ref{table2}.  To see (v),   we first note  that for $0 \leq m \leq n$ there holds $\phi(n)-\phi(m) \leq 2(k(n)-k(m))+2$  and  $\phi(m)-\phi(n) \leq 2(k(m)-k(n))+2$. This can be seen by careful inspection of (\ref{eq: different repr}). Additionally, we use that for $t\leq n_1\leq n_2$  there holds  $k(n_1-t) - k(n_1) + k(n_2 + t) - k(n_2)\leq 1 $.    
 
  We proceed with (vi). The  case   $m=3$ follows from (iii) and $\phi(3) = 3$, see Table \ref{table2} below. The case $m=4$ follows by taking additionally (i)  and $\phi(4)= 4$ into account.  For $m=5$, we use (iv)  and the fact that $\phi(5) = 3$.   Now suppose that $n \ge 11 $ and  $6\leq m,n-m\leq n$.

We  will use the following property: let $k_1,k_2,m_0 \in \mathbb{N}$, with $m_0 \ge k_1 \ge k_2$. Since the function $x \mapsto \sqrt{x}$ is concave and increasing, the function
\begin{align}\label{eq: 4timesprop}
m \mapsto \frac{1}{2}\sqrt{2m-k_1}+\frac{1}{2}\sqrt{2(n-m)-k_2}, \quad m_0\leq m,n-m \leq n
\end{align}
attains its minimum for $m=m_0$. Moreover, we will use that $\lfloor a+ b \rfloor \le \lfloor a  \rfloor + \lfloor b \rfloor  +1$. We will work directly with the definition of $\phi$, see \eqref{def: phi}.  Consulting  Table~\ref{table2}, we note that the cases $11 \le n \le 32$ can be checked directly by comparing the values of $\phi(n-m)+\phi(m)$ and $\phi(n)$. We can therefore assume that $n \geq 33$.  There are four cases to consider: (a) $n$ even, $m$ even. (b) $n$ even, $m$ odd. (c) $n$ odd, $m$ even. (d) $n$ odd, $m$ odd.

\noindent\emph{(a) $n$ even, $m$ even.}  By \eqref{eq: 4timesprop} for $k_1 = k_2 = 4$ and $m_0 = 6$, we get
\begin{align*}
\phi(m)+\phi(n-m) &= 2 \lfloor \tfrac{1}{2}\sqrt{2m-4} \rfloor + 2 \lfloor \tfrac{1}{2}\sqrt{2(n-m)-4} \rfloor +4 \\&\geq 2  \lfloor \tfrac{1}{2}\sqrt{2m-4}+ \tfrac{1}{2}\sqrt{2(n-m)-4} \rfloor +2 \\&
\geq 2  \lfloor \tfrac{1}{2}\sqrt{8}+ \tfrac{1}{2}\sqrt{2n-16} \rfloor +2 \\& = 2\left\lfloor\frac{1}{2}\sqrt{2n-4} +\sqrt{2}-\frac{6}{\sqrt{2n-4}+\sqrt{2n-16}}\right\rfloor +2.
\end{align*}
Now for $n \geq 32$ we have that  $\sqrt{2}- 6(\sqrt{2n-4}+\sqrt{2n-16})^{-1} \geq 1$, which indeed yields $\phi(m)+\phi(n-m) \ge 2\lfloor\frac{1}{2}\sqrt{2n-4}\rfloor +2 +2= \phi(n)+2$.


\noindent\emph{(b) $n$ even, $m$ odd.}  Observe that in this case we have $m,n-m\ge7$ and thus $n \ge 14$.  By \eqref{eq: 4timesprop} for $k_1 = k_2 = 5$ and $m_0 = 7$, we obtain 
\begin{align*}
\phi(m)+\phi(n-m) &= 2 \lfloor-\tfrac{1}{2}+ \tfrac{1}{2}\sqrt{2m-5} \rfloor + 2 \lfloor -\tfrac{1}{2}+\tfrac{1}{2}\sqrt{2(n-m)-5} \rfloor +6 \\&\geq 2  \lfloor \tfrac{1}{2}\sqrt{2m-5}+ \tfrac{1}{2}\sqrt{2(n-m)-5}-1 \rfloor +4 \\&
\geq 2  \lfloor \tfrac{1}{2}\sqrt{9}+ \tfrac{1}{2}\sqrt{2n-19}-1 \rfloor +4 \\&
= 2\left\lfloor\frac{1}{2}\sqrt{2n-4} +  \frac{1}{2}  -\frac{1}{2}\frac{15}{\sqrt{2n-4}+\sqrt{2n-19}}\right\rfloor +4.
\end{align*}
We check that for all   $n \ge  34$ (note $n$ is even)  there holds  $1 - 15(\sqrt{2n-4}+\sqrt{2n-19})^{-1} \geq 0$,  and thus $\phi(m)+\phi(n-m)  \ge \phi(n)+2$. 


\noindent\emph{(c) $n$ odd, $m$ even.}  By \eqref{eq: 4timesprop} for $k_1 = 5$, $k_2 = 4$ and $m_0 = 6$, we obtain 
\begin{align*}
\phi(m)+\phi(n-m) &= 2 \lfloor\tfrac{1}{2}\sqrt{2m-4} \rfloor + 2 \lfloor -\tfrac{1}{2}+\tfrac{1}{2}\sqrt{2(n-m)-5}\rfloor +5 \\&\geq 2  \lfloor \tfrac{1}{2}\sqrt{2m-4}+ \tfrac{1}{2}\sqrt{2(n-m)-5}-\tfrac{1}{2} \rfloor +3 \\&
\geq 2 \lfloor \tfrac{1}{2}\sqrt{ 8  }+ \tfrac{1}{2}\sqrt{2n- 17  }-\tfrac{1}{2} \rfloor +3 \\& =  2\left\lfloor\frac{1}{2}\sqrt{2n-5}-\frac{1}{2} +  \sqrt{2} \EEE - \frac{ 6 \EEE }{\sqrt{2n-5}+\sqrt{2n- 17  }}\right\rfloor +3.
\end{align*}
For  $n \geq 33$  we have that   $\sqrt{2}-6(\sqrt{2n-5}+\sqrt{2n-17})^{-1} \geq 1$,   and thus $\phi(m)+\phi(n-m)  \ge \phi(n)+2$.  


\noindent\emph{(d) $n$ odd, $m$ odd.} We proceed as in (c) by interchanging the roles of $m$ and $n-m$.

 Finally, to see  (vii), \EEE one may follow the lines of the proof of   (vi). \EEE  We sketch only the case where $n$ and $m$   are   even. By repeating the argument in (a) for general $m_0$ we find 
$$\phi(m)+\phi(n-m) \ge 2\left\lfloor\frac{1}{2}\sqrt{2n-4} + \tfrac{1}{2}\sqrt{2m_0-4}-\frac{m_0}{\sqrt{2n-4}+\sqrt{2n-2m_0-4}}\right\rfloor +2. $$
 One can check that for $m_0 =  3945 \EEE $ we have $\phi(m)+\phi(n-m) \ge \phi(n) + 52$ for all $n \ge 2m_0$. 
\end{proof}
%

{
\begin{table}[h]\centering
\begin{tabular}{|c||c|c|c|c|c|c|c|c|c|c|c|c|c|c|c|c|c|c|c|c|c|c|c|c|c|c|c|c|c|c|c|c|c|}
\hline
$n$ & 1 & 2 & 3 & 4 & 5 & 6 & 7 & 8 & 9 & 10 & 11 & 12 & 13 & 14 &15 & 16     \\  \hline 
$\phi(n)$& 1 & 2 & 3 & 4&3&4&5&4& 5&6&5&6&5&6&7 &6\\
\hline 
\hline
$n$ &17 &18 &19 &20 &21 & 22 & 23 & 24 &25 & 26 &27 &28 &29&30 &31&32 \\
\hline
$\phi(n)$  &7&6&7 & 8 &7&8 &7 &8&7&8&9&8&9&8   & 9 & 8   \\
\hline 
  \end{tabular}
\vspace{8mm}
\caption{The function  $\phi(n)$ for $1 \leq n \leq 32$.}  \label{table2}
\end{table}
}

We now are in a position to establish the lower bound for  $q_\mathrm{sat}^n$.  This lower bound together with the upper bound in  Proposition \ref{proposition : upper bound} shows    Proposition \ref{theorem : qsat}.

\begin{proposition}[Lower bound for $q_\mathrm{sat}^n$]\label{proposition : lower bound} Let $n \in \mathbb{N}$.  Then $q_\mathrm{sat}^n \geq \phi(n)$.
\end{proposition}
\begin{proof} 
Let $C_n$ be a $q_\mathrm{sat}^n$-optimal configuration,  i.e.,  $\mathcal{Q}(C_n)=q_\mathrm{sat}^n $ and  $\mathcal{E}(C_n) =\mathcal{E}_\mathrm{min}^n(q_\mathrm{sat}^n)=-2n+2q_\mathrm{sat}^n$. We proceed to show that $q_\mathrm{sat}^n=\mathcal{Q}(C_n) \geq \phi(n)$.
We prove the claim by induction.   In view of Table \ref{table2},  it is elementary to check that for $1 \leq n \leq 10 $ the claim holds.  Indeed, for $1 \leq n \leq 4$ the cardinality of the minority phase is zero, for $5 \leq n \leq 7$ the cardinality of the minority phase is at most one, and for $8 \leq n \leq 10$ the cardinality of the minority phase is at most two.  This is due to the fact  that by Lemma \ref{lemma : Charge Energybound}(a)(i),(iv) each atom of the minority phase is $4$-bonded with bond angles $\pi/2$ and thus  two atoms  of the minority phase can share at most two neighbors.  Finally, three such atoms can only have one neighbor in common. 

 Let $n \ge 11$.  We assume that the statement holds for all $m \in \mathbb{N}$  with  $1\leq m < n$, and prove the statement for $n$. We proceed in three steps. First, we show that the claim holds true if  $C_n$ is not connected. Then, we  treat  the case where  $C_n$   is connected and  contains a bridging atom. \EEE Finally, we  address     the case of connected  $C_n$  without any bridging atoms. \EEE  
 
\noindent\emph{Step 1: $C_n$ is not connected.} We assume that $C_n$ is not connected. Denote by $C_m$ and $C_{n-m}$, $1 \leq m,n-m \leq n$, two sub-configurations consisting of $m$ and $n-m$ atoms, respectively,  with no bonds between them.  Lemma \ref{lemma : Charge Energybound}(a)(vi) implies that $\mathcal{Q}(C_m) \geq q_\mathrm{sat}^m$ and $\mathcal{Q}(C_{n-m}) \geq q_\mathrm{sat}^{n-m}$.  Suppose that $m \leq 2$ or $n-m \le 2$, without restriction say $m \leq 2$.  We can apply the induction  hypothesis,  Table \ref{table2},  and Lemma \ref{lemma : properties phi}(i) ($m$ times)  to obtain
\begin{align*}
\mathcal{Q}(C_n)   = \mathcal{Q}(C_m) + \mathcal{Q}(C_{n-m}) \ge  q_\mathrm{sat}^{m} + q_\mathrm{sat}^{n-m} \geq \phi(m) +\phi(n-m)    = m + \phi(n-m) \geq \phi(n).
\end{align*}
  On the other hand,   if  $m,n-m \geq 3$, the induction   hypothesis and    Lemma \ref{lemma : properties phi}(vi)   (recall $n\geq 11$)   yield
\begin{align*}
 q_\mathrm{sat}^n =  \mathcal{Q}(C_n) = \mathcal{Q}(C_m) + \mathcal{Q}(C_{n-m}) \geq \phi(m) +\phi(n-m) \geq 2+ \phi(n)  \ge \phi(n). 
\end{align*}
This yields the claim in the case  that  $C_n$ is \EEE not connected.

\noindent\emph{Step 2:    $C_n$ is connected and contains   bridging atoms.} Assume that the bond graph contains a  bridging atom. Denote the two configurations that a connected through are   bridging atom by $C_{m}$ and $C_{n-m}$, see   the definition before \eqref{eq: bridgi}.  By Lemma \ref{lemma: bridging}(a) and the induction hypothesis we have  $n \ge m, n-m \ge 4$ and 
\begin{align*}
\mathcal{Q}(C_m) \geq q_\mathrm{sat}^{m+1} -1\geq \phi(m+1)-1 \ \ \ \text{ and } \ \ \ \mathcal{Q}(  C_{n-m} ) \geq q_\mathrm{sat}^{n-m}\geq \phi(n-m).
\end{align*}
By using Lemma \ref{lemma : properties phi}(i),(vi)    we conclude
\begin{align*}
 q_\mathrm{sat}^n = \mathcal{Q}(C_n) =  \mathcal{Q}(C_m) + \mathcal{Q}( C_{n-m}  )\geq \phi(m+1)-1+\phi(n-m) \geq \phi(n+1)+1 \geq \phi(n).
\end{align*}

\noindent\emph{Step 3:   $C_n$ is connected and does not contain any bridging atoms.} Since $C_n$ does not contain any bridging atoms, Lemma \ref{lemma : charge interior}   \EEE  along with  the induction hypothesis  yields 
 \begin{align*}
q_\mathrm{sat}^n=\mathcal{Q}(C_n) \geq q_\mathrm{sat}^{n-d-m}+m  \ge   \phi(n-d-m) +m
\end{align*}
for some $m \ge 4$. By Lemma \ref{lemma : card positively charged shell}  and Lemma \ref{lemma : properties phi}(i),(ii)   we then obtain
\begin{align*}
 q_\mathrm{sat}^n \geq  \phi(n-4q_\mathrm{sat}^n +\eta+2 \# I_{\rm ac}^{\rm ext}) +4. 
\end{align*}
By Lemma \ref{lemma : properties phi}(ii),  $\# I_{\rm ac}^{\rm ext} \ge m \EEE \geq 4$,  and the fact that $\eta$ is even  (see Lemma \ref{RemarkPolygon}(a))  we get
\begin{align}\label{eq: NNN}
q_\mathrm{sat}^n \geq   \phi(n-4q_\mathrm{sat}^n + \eta +2 \# I_{\rm ac}^{\rm ext}) +4 \geq \phi(n+8-4q_\mathrm{sat}^n)+4.
\end{align}
 Suppose first that $n+8-4q_\mathrm{sat}^n \le -1$. Then by  Table~\ref{table2},  \eqref{def: phi},  and the fact that $n- q_\mathrm{sat}^n$ is even  it is elementary to check that $q_\mathrm{sat}^n \ge \lceil (n+9)/4 \rceil \ge \phi(n)-1 $. Otherwise, the claim follows from \eqref{eq: NNN} and Lemma \ref{lemma : computation of qsatn}  below for $x=q_\mathrm{sat}^n$,  where we  again use  that $n- q_\mathrm{sat}^n$ is always even. 
\end{proof}

 In the previous proof we have used the following lemma. 

\begin{lemma} \label{lemma : computation of qsatn} Let $n\geq 7$. Let $x \in \mathbb{N}$ with $x \le n/4 +2$,  $n-x$ even,       and  
\begin{align} \label{ineq : q recursive}
x \ge  4+\phi(n+ 8-4x).
\end{align}
Then there holds $x \ge \phi(n)$.
\end{lemma}

\begin{proof}
If $4+\phi\left(8+n-4x\right)  \geq \phi(n)$, the statement follows from  (\ref{ineq : q recursive}) and if $x \geq \phi(n)$ there is nothing to prove. We  now  assume by contradiction that 
\begin{align}\label{ineq : strict phin}
{\rm (i)} \ \ 4+\phi\left(8+n-4x\right) < \phi(n)\ \ \  \text{ and }  \ \  \ {\rm (ii)} \ \ x <\phi(n).
\end{align}
 Since $\phi(8+n-4x)-\phi(n)$ is even  by (\ref{def: phi}),   and  $x-\phi(n)$  is even by assumption and (\ref{def: phi}), \EEE   (\ref{ineq : strict phin}) yields
\begin{align}\label{ineq : strict phin2}
{\rm (i)} \ \ 6+\phi\left(8+n-4x\right) \leq \phi(n) \ \ \  \text{ and } \EEE \ \  \   {\rm (ii)} \ \ x  \leq \phi(n)-2.
\end{align}
Let $n=1+2k^2+2k+m$, $m \in \mathbb{N}$, $1 \leq m \leq 4k+4$. By (\ref{ineq : strict phin2})(i) there holds $\phi(n) \geq 6$. Using Table~\ref{table2}, we observe that   $n \geq 10$ and thus also $k \ge 1$.  We distinguish three cases:  (a) $2\leq m \leq 4k+4$ and $m$ even. (b) $1\leq m \leq 2k+1$, $m$ odd. (c) $2k+3\leq m \leq 4k+3$, $m$ odd.

\noindent\emph{(a) $2 \leq m \leq 4k+4$ and $m$ even.}  By (\ref{eq: different repr})  there holds $\phi(n) = 3+2k$,  and (\ref{ineq : strict phin2})(ii) thus     implies
\begin{align*}
8+n-4x &\geq 1 + 2k^2 +2k +m  -4(\phi(n)-2) +8  =  1 + 2k^2 +2k +m  -4(2k+1) +8\\&= 1+2(k-2)^2 +2(k-2) + m  \geq 1+2(k-2)^2 +2(k-2) +2.
\end{align*}
 One can also check that the difference of the first and the last expression is even.   By Lemma \ref{lemma : properties phi}(ii), (\ref{ineq : strict phin2})(i),  and  $\phi(n) = 3+2k$  we then  obtain
\begin{align*}
\phi\big(1+2(k-2)^2 +2(k-2) +2\big)+6 \leq \phi\left(8+n-4x\right) +6 \leq \phi(n)=3+2k.
\end{align*}
  For $n\geq 14$ and thus $k\geq 2$,  there holds $\phi(1+2(k-2)^2 +2(k-2) +2)=3+2(k-2)$ by (\ref{eq: different repr}).  This yields the contradiction $5 +2k \le 3 + 2k$. A contradiction in the cases $n=11,13$, i.e., $k=1$, can be obtained by noting $1+2(k-2)^2 +2(k-2) +2=3$ and   $\phi(3)=3$, cf. Table~\ref{table2}.

\noindent\emph{(b) $1 \leq m \leq 2k+1$ and $m$ odd.}  By (\ref{eq: different repr}) there holds $\phi(n) = 2+2k$ and (\ref{ineq : strict phin2})(ii) thus yields
\begin{align*}
8+n-4x &\geq 1 + 2k^2 +2k +m  -4(\phi(n)-2) +8  =  1 + 2k^2 +2k +m  -8k +8\\&= 1+2(k-2)^2 +2(k-2) + m +4 \geq 1+2(k-2)^2 +2(k-2) +1.
\end{align*}
 The difference of the first and the last expression is even.  By Lemma \ref{lemma : properties phi}(ii),   (\ref{ineq : strict phin2})(i),  and $\phi(n) = 2+2k$  we obtain
\begin{align*}
\phi\big(1+2(k-2)^2 +2(k-2) +1\big)+6 \leq \phi\left(8+n-4x\right) +6 \leq \phi(n)=2+2k.
\end{align*}
  For $n\geq 14$ and thus $k\geq 2$, \EEE there holds $\phi(1+2(k-2)^2 +2(k-2)+1)=2+2(k-2)$ by (\ref{eq: different repr}). This yields the contradiction $4 +2k \le 2 + 2k$. A contradiction  in the cases $n=10,12$, i.e., $k=1$, can be obtained by noting $1+2(k-2)^2 +2(k-2) +1=2$ and   $\phi(2)=2$, cf. Table~\ref{table2}.

\noindent\emph{(c) $2k+3 \leq m \leq 4k+3$ and $m$ odd.} By (\ref{eq: different repr}) there holds $\phi(n) = 4+2k$ and (\ref{ineq : strict phin2})(ii) thus  implies
\begin{align*}
8+n-4x &\geq 1 + 2k^2 +2k +m  -4(\phi(n)-2) +8  =   1 + 2k^2 +2k +m  -4(2k+2) +8\\&= 1+2(k-2)^2 +2(k-2) + m -4 \geq 1+2(k-2)^2 +2(k-2) +2(k-2)+3.
\end{align*}
 The difference of the first and the last expression is even.  By Lemma \ref{lemma : properties phi}(ii), (\ref{ineq : strict phin2})(i),  and $\phi(n) = 4+2k$   we obtain
\begin{align*}
\phi\big(1+2(k-2)^2 +2(k-2) +2(k-2) +3\big)+6 \leq \phi\left(8+n-4x\right) +6 \leq \phi(n)=4+2k.
\end{align*}
  For $n\geq 14$ and thus $k\geq 2$,  there holds $\phi(1+2(k-2)^2 +2(k-2)+2(k-2) +3)=4+2(k-2)$ by (\ref{eq: different repr}).  This yields the contradiction $6 +2k \le 4 + 2k$. A contradiction  in the cases $n=10,12$, i.e., $k=1$, can be obtained by noting $1+2(k-2)^2 +2(k-2) +2(k-2)+3=2$ and   $\phi(2)=2$, cf. Table~\ref{table2}.
\end{proof}

\subsection{Crystallization result for \EEE $q_\mathrm{sat}^n$-optimal configurations} \label{subsection : geometry qsat}

This subsection is devoted to the proof of  Theorem \ref{theorem : GeometryGroundstate}(a). \EEE In Lemma \ref{lemma: removi} we first show that $q^{n}_{\rm sat}$-optimal configurations are  connected and do not contain bridging atoms   after removing a finite number of atoms independently of $n$. Afterwards, we control the number of non-equilibrated atoms (Lemma \ref{lemma : smallness of A}). This then allows us to show   Theorem \ref{theorem : GeometryGroundstate}(a).

 Recall that in Proposition \ref{theorem : qsat} we have shown that $\phi(n) = q_{\rm sat}^{n}$ for all $n \in \mathbb{N}$. In the following, we will use this equality without further notice.  As before, it is not restrictive to consider configurations $C_n$ with $\mathcal{Q}(C_n) \ge 0$.

\begin{lemma}[Connectedness, bridging atoms]\label{lemma: removi}
 Let \EEE $C_n$ be a $q_\mathrm{sat}^n$-optimal configuration. 
 
\noindent (a) If $n\geq 11$ and  $C_n$ is not connected, we can remove  $m \in \lbrace 1,2\rbrace$  $0$-bonded \EEE atoms from  $C_n$ to obtain a  connected  $q_{\rm sat}^{n-m}$-optimal configuration.

\noindent (b)  Let $n\geq 26$  and \EEE let $m_0$ be the constant from Lemma \ref{lemma : properties phi}(vii).   If $C_n$ contains a bridging atom, there holds $\mathcal{Q}(C_n) \ge \phi(n- 26 \EEE ) + 2$. We can remove $m \in \lbrace 1,\ldots,m_0\rbrace$ atoms from $C_n$ to obtain a configuration $C_{n-m}$ which  is $q_{\rm sat}^{n-m}$-optimal and  does not contain bridging atoms.  
\end{lemma}

 We defer the proof and continue with the next ingredient for the proof of Theorem \ref{theorem : GeometryGroundstate}(a). \EEE We  show that the number of non-equilibrated atoms $\mathcal{A}(X_n)$ can be controlled. We remark that the bound on $\mathcal{A}(X_n)$ is not sharp and could be improved at the expense of more elaborated methods. As our focus lies on a qualitative description of the geometry of optimal configurations, we refrain from entering into finer estimates.

\begin{lemma}[Control on $\#\mathcal{A}(X_n)$] \label{lemma : smallness of A} Let $C_n$ be a   $q_\mathrm{sat}^n$-optimal configuration  without \EEE bridging atoms. Then $\#\mathcal{A}(X_n) \leq  50  $. 
\end{lemma}

 We again defer the proof and proceed to show Theorem \ref{theorem : GeometryGroundstate}(a). \EEE

%
%
%
%
%
%
%
%
%
%
%

\begin{proof}[Proof of Theorem \ref{theorem : GeometryGroundstate}(a)]    Let $m_0$ be the constant from Lemma \ref{lemma : properties phi}(vii) and define $n_0 =  52   m_0+2$. Let $C_n$ be a $q_\mathrm{sat}^n$-optimal configuration. If $n < n_0$, the statement is trivial, we therefore suppose that $n \ge n_0$.    The goal is to prove that,  after removing at most  $2m_0+2$ atoms, \EEE the remaining configurations   is a subset of the square lattice.  In view of Lemma \ref{lemma: removi}, we can remove $m \le m_0 +2$  atoms from $C_n$ to obtain a connected configuration $C_{n-m}$ without bridging atoms which is a $q_\mathrm{sat}^{n-m}$-optimal configuration. Consequently, it suffices to  consider a   $q_\mathrm{sat}^n$-configuration \EEE $C_n$, $n \ge   51 \EEE  m_0$,  which is connected without bridging atoms and to show that, after removing at most $m_0$ atoms,  it is a subset of the square lattice. 

\noindent \emph{Step 1: Proof for connected configurations without bridging atoms.} We introduce \EEE
\begin{align*}
\mathcal{X}(X_n) := \left\{X_m \subset X_n : X_m \text{ connected and  $\mathcal{A}(X_m) = \emptyset$\EEE}\right\}.
\end{align*}
Note that, if $X_m \in \mathcal{X}(X_n)$, then up to isometry $X_m \subset \mathbb{Z}^2$. Choose $X_n^{\mathrm{max}} \in \mathcal{X}(X_n)$ as a maximal element, that is $\#X_m \leq \# X_n^\mathrm{max}$ for all $X_m \in \mathcal{X}(X_n)$. Denote its cardinality by $n_\mathrm{max} \leq n$. As $C_n$ is connected and $\#\mathcal{A}(C_n)\leq  50 \EEE$, it is elementary to see that $C_n$ consists of at most $ 51 \EEE $ sub-configuration each of which subset of a (different) square lattice. Since $n \ge  51 \EEE  m_0$, this  implies $n_{\rm max} \ge m_0$.       Additionally, we set $C_n^\mathrm{max}=\{(x_i,q_i) : x_i \in X_n^\mathrm{max}\}.$
Our goal is now to prove that
\begin{align}\label{ineq : Xn setminus Xnmax}
\#\left(X_n \setminus X_n^\mathrm{max}\right)=n-n_\mathrm{max} \le m_0.
\end{align}
 The main ingredient for the  proof \EEE is the estimate \EEE
\begin{align}\label{ineq : chargesXnmax}
{\rm (i)} \ \ \mathcal{Q}(C_n^\mathrm{max}) \geq \phi(n_\mathrm{max}), \ \ \ \ \ \ \ \ {\rm (ii)} \ \      \mathcal{Q}(C_n \setminus C_n^\mathrm{max}) \geq \phi(n-n_\mathrm{max}+ 50 \EEE )-  50 \EEE  \EEE.
\end{align}
 We defer the proof of \eqref{ineq : chargesXnmax} to Step 2 below and first show \eqref{ineq : Xn setminus Xnmax}. Assume by contradiction that  (\ref{ineq : Xn setminus Xnmax}) \EEE does not hold true, i.e., $n-n_\mathrm{max} > m_0$. By (\ref{ineq : chargesXnmax})   we obtain
\begin{align*}
\mathcal{Q}(C_n) = \mathcal{Q}(C_n^\mathrm{max}) + \mathcal{Q}(C_n \setminus C_n^\mathrm{max}) \geq \phi(n_\mathrm{max}) + \phi(n-n_\mathrm{max}+  50 \EEE  )-  50 \EEE  \EEE.
\end{align*}
Since also $n_{\rm max} \ge m_0$, we then derive   by   Lemma \ref{lemma : properties phi}(ii),(vii) 
\begin{align*}
\mathcal{Q}(C_n) \geq \phi(n_\mathrm{max}) + \phi(n-n_\mathrm{max}+  50 \EEE  )-  50 \EEE   \geq   \phi(n+  50 \EEE )+2 \geq \phi(n)+2.
\end{align*}
This  yields $\mathcal{Q}(C_n)  \ge q_\mathrm{sat}^n+2$  and  contradicts \EEE  $\mathcal{Q}(C_n)=q_\mathrm{sat}^n$.

\noindent \emph{Step 2: Proof of (\ref{ineq : chargesXnmax}).} To conclude the proof, it remains to confirm (\ref{ineq : chargesXnmax}). As a preparation, \EEE we  introduce
\begin{align*}
\mathrm{d}X_n^\mathrm{max}:=\{x_i \in X_n^\mathrm{max} : \mathcal{N}(x_i) \setminus X_n^\mathrm{max} \neq \emptyset\},
\end{align*}
 where $ \mathcal{N}(x_i) $ (see \eqref{eq: neighborhood}) is defined with respect to $X_n$. \EEE First, we prove that 
 \begin{align}\label{subset : dXnmax} 
 \mathrm{d}X_n^\mathrm{max} \subset X_n^+ \cap \mathcal{A}(X_n).
 \end{align}
To see this, let $x_i \in \mathrm{d}X_n^\mathrm{max}$ and   choose $x_j \in \mathcal{N}(x_i) \setminus X_n^\mathrm{max}$.  Suppose by contradiction that $x_i  \notin \mathcal{A}(X_n)$ or $x_i \in X_n^-$. In both cases, this implies $x_i  \notin \mathcal{A}(X_n)$ by Lemma \ref{lemma : Charge Energybound}(a)(i),(iv). Then  all bond angles at $x_i$ are integer multiples of $\frac{\pi}{2}$, and thus the configuration $X_n^\mathrm{max} \cup \{x_j\} $ is such that  $\mathcal{A}(X_n^\mathrm{max} \cup \{x_j\} )=\emptyset$\EEE. This implies $X_n^\mathrm{max} \cup \{x_j\} \in \mathcal{X}(X_n)$  and \EEE contradicts the maximality of $\# X_n^{\rm max}$.

 We now show (\ref{ineq : chargesXnmax}) and begin with (i). \EEE By (\ref{subset : dXnmax}) and Lemma \ref{lemma : Charge Energybound}(a)(i) there holds $\#(\mathcal{N}(x_i) \cap X_n^\mathrm{max}) = \#\mathcal{N}(x_i) =4$  for all $x_i \in X_n^{\rm max}$ with $q_i = -1$, i.e., all negatively charged atoms of $X_n^\mathrm{max}$ are $4$-bonded.  Moreover, since $X_n^\mathrm{max} \subset X_n$, $X_n^\mathrm{max}$ is repulsion-free and all bonds have unit length by Lemma \ref{lemma : Charge Energybound}(a)(i),(ii). By Lemma  \ref{lemma : Charge Energybound}(b) this implies the equality $\mathcal{E}(C_n^\mathrm{max}) =   -2n_{\rm max} + 2\mathcal{Q}(C_n^\mathrm{max})$. Then  Lemma \ref{lemma : Charge Energybound}(a)(iii) yields      $\mathcal{Q}(C_n^\mathrm{max}) \geq \phi(n_\mathrm{max})$, as desired.

We now prove \eqref{ineq : chargesXnmax}(ii). We set  $\tilde{X}_n := (X_n \setminus X_n^\mathrm{max}) \cup \mathrm{d}X_n^\mathrm{max}$ and  \EEE
\begin{align*}
 \tilde{C}_n   :=  \EEE \{(x_i,q_i) : x_i \in (X_n \setminus X_n^\mathrm{max}) \cup \mathrm{d}X_n^\mathrm{max}\}.
\end{align*}
Since  $\mathrm{d}X_n^\mathrm{max} \subset X_n^+$  by \eqref{subset : dXnmax}, \EEE we obtain
\begin{align}\label{eq : charge Xnmaxc}
\mathcal{Q}( \tilde{C}_n) = \mathcal{Q}(C_n \setminus C_n^\mathrm{max}) + \#\mathrm{d}X_n^\mathrm{max}.
\end{align}
We now proceed   as in the proof of \eqref{ineq : chargesXnmax}(i): let $x_i \in    \tilde{X}_n^-\EEE$. By (\ref{subset : dXnmax}) and Lemma \ref{lemma : Charge Energybound}(a)(i) there holds $\#(\mathcal{N}(x_i) \cap \tilde{X}_n) = \#\mathcal{N}(x_i) =4$. Since $\tilde{X}_n \subset X_n$, $\tilde{X}_n$ is repulsion-free and all bonds are of unit length by Lemma \ref{lemma : Charge Energybound}(a)(i),(ii).  By Lemma  \ref{lemma : Charge Energybound}(b) this implies the equality $\mathcal{E}(\tilde{C}_n) =   -2\#\tilde{X}_n + 2\mathcal{Q}(\tilde{C}_n)$. Then   Lemma \ref{lemma : Charge Energybound}(a)(iii) implies  $\mathcal{Q}(\tilde{C}_n) \geq \phi(\#\tilde{X}_n)$. \EEE

By (\ref{subset : dXnmax}) and Lemma \ref{lemma : smallness of A}  we obtain $\#\mathrm{d}X_n^\mathrm{max} \leq \#\mathcal{A}(X_n)\leq  50 \EEE $.
By using (\ref{eq : charge Xnmaxc}), $\#\tilde{X}_n = n-n_\mathrm{max}+\#\mathrm{d}X_n^\mathrm{max} $, and applying Lemma \ref{lemma : properties phi}(i) ($  50 \EEE  -\#\mathrm{d}X_n^\mathrm{max}$)-times, we obtain
\begin{align*}
\mathcal{Q}(C_n \setminus C_n^\mathrm{max}) &= \mathcal{Q}(  \tilde{C}_n \EEE ) - \#\mathrm{d}X_n^\mathrm{max} \geq   \phi(\#\tilde{X}_n) - \#\mathrm{d}X_n^\mathrm{max}  \\&= \phi(n-n_\mathrm{max}+\#\mathrm{d}X_n^\mathrm{max} ) - \#\mathrm{d}X_n^\mathrm{max} \geq \phi(n-n_\mathrm{max}+ 50 \EEE )-  50. \EEE
\end{align*}
This concludes the proof \EEE of \eqref{ineq : chargesXnmax}(ii). 
\end{proof}

\begin{remark}[Maximal component]\label{remark: bridging-helpi}
For later purposes, we observe that the configuration $C_n^{\mathrm{max}}$ identified in the proof is a connected subset of the square lattice and that it is saturated, i.e., the atoms with charge $-1$ are $4$-bonded. We also note that $\mathcal{Q}(C_n^{\mathrm{max}}) \le \phi(n_{\mathrm{max}}) + 49$. In fact, otherwise by  (\ref{ineq : chargesXnmax})(ii) and  Lemma \ref{lemma : properties phi}(ii),(vi)  we  would get the contradiction  $\mathcal{Q}(C_n) \ge \phi(n_{\rm max}) + 50 +\phi(n-n_\mathrm{max}+ 50 )-  50 \ge \phi(n)+2$.
\end{remark}

\EEE

We proceed \EEE  with the proofs of Lemma \ref{lemma: removi} and Lemma \ref{lemma : smallness of A}. \EEE
 
\begin{proof}[Proof of Lemma \ref{lemma: removi}]
(a)   Let $n\geq 11$ and  assume  by contradiction that $C_n$ is not connected with two connected components of at least three atoms.  Denote by $C_m$ and $C_{n-m}$, $3 \leq m,n-m \leq n$, two  sub-configurations consisting of $m$ and $n-m$ atoms, respectively,  with no bonds between them. By Lemma \ref{lemma : Charge Energybound}(a)(vi)   we get $\mathcal{Q}(C_m)\geq \phi(m)$ and $\mathcal{Q}(C_{n-m})\geq\phi(n-m)$. By   Lemma  \ref{lemma : properties phi}(vi)   \EEE we obtain
\begin{align*}
\mathcal{Q}(C_n) = \mathcal{Q}(C_m) + \mathcal{Q}(C_{n-m}) \geq \phi(m) +\phi(n-m) \geq 2+ \phi(n) =  2 + q_{\rm sat}^n. 
\end{align*}
This contradicts $\mathcal{Q}(C_n)=q_\mathrm{sat}^n$.  Observe that two single atoms may indeed exist, see Fig.~\ref{fig : construction qsat}.  We also note that these (at most two)  atoms are $0$-bonded since they have the same charge $+1$ and $C_n$ has alternating charge distribution.

Thus, we can remove $m \in \lbrace 1,2\rbrace$ atoms of charge $+1$ to obtain a connected configuration $C_{n-m}$ with $\mathcal{Q}(C_{n-m}) = \mathcal{Q}(C_n)-m = q_{\rm sat}^n - m$. By Lemma \ref{lemma : properties phi}(i) this implies $\mathcal{Q}(C_{n-m}) \le q_{\rm sat}^{n-m}$. On the other hand, by construction, $C_{n-m}$ satisfies
$$\mathcal{E}(C_{n-m}) = \mathcal{E}(C_n) = -2n +2\mathcal{Q}(C_n) = -2(n-m) + 2\mathcal{Q}(C_{n-m}).$$ 
Therefore, we get $\mathcal{Q}(C_{n-m}) \ge q_{\rm sat}^{n-m}$, see   Lemma \ref{lemma : Charge Energybound}(a)(iii). Thus, $C_{n-m}$ is $q_{\rm sat}^{n-m}$-optimal.

 \noindent (b) \emph{Step 1: There exists at most one bridging atom.} Assume by contradiction that there exist two bridging atoms.  Denote the two components connected through the first bridging atom  by $C_{m_1},C_{n-m_1}$, and denote the two components of $C_{n-m_1}$ connected through the second bridging atom by  $C_{m_2}$ \EEE and $C_{n-m_1-m_2}$, see \eqref{eq: bridgi}.    By Lemma \ref{lemma: bridging}(a)  we get  $n \geq n-m_1-m_2,m_1,m_2 \ge 4$ and 
\begin{align*}
\mathcal{Q}(C_{m_1}) \geq \phi(m_1+1)-1, \ \ \ \  \mathcal{Q}(C_{m_2}) \geq \phi(m_2+1)-1, \ \ \ \  \mathcal{Q}(C_{n-m_1-m_2}) \geq \phi(n-m_1-m_2).
\end{align*}
 By using  Lemma \ref{lemma : properties phi}(ii) and  Lemma \ref{lemma : properties phi}(vi) twice we derive  
\begin{align}\label{eq: repeati1}
\mathcal{Q}(C_n) \geq \phi(m_1+1)+\phi(m_2+1)+\phi(n-m_1-m_2)-2 &\geq \phi(m_1+m_2+2) + \phi(n-m_1-m_2)\notag \\&\geq \phi(n+2)+2 \geq \phi(n)+2 =  2 + q_{\rm sat}^n.
\end{align}
 This contradicts $\mathcal{Q}(C_n)=q_\mathrm{sat}^n$. 
 
\noindent \emph{Step 2: If there exists a bridging atom, then one of the components contains at most $m_0$ atoms.} 
 Assume  by contradiction  that $C_n$ consists of two components $C_m$, $C_{n-m}$ connected through  a bridging atom  with \EEE $m,n-m \geq m_0+1$.     By Lemma \ref{lemma: bridging}(a) we get \EEE $\mathcal{Q}(C_{m}) \geq \phi(m+1)-1$ and $\mathcal{Q}(C_{n-m}) \geq \phi(n-m)$. Then Lemma \ref{lemma : properties phi}(i),(vii) yield
 \begin{align}\label{eq: repeati2}
 \mathcal{Q}(C_n) =\mathcal{Q}(C_{m}) + \mathcal{Q}(C_{n-m}) \geq \phi(m+1)-1 + \phi(n-m) \geq \phi(n+1) +   51 \EEE &\geq \phi(n)+   50. \EEE
 \end{align}
  Thus, $\mathcal{Q}(C_n) > q_\mathrm{sat}^n$.  \EEE  This contradicts $\mathcal{Q}(C_n)=q_\mathrm{sat}^n$.

\noindent \emph{Step 3: Conclusion.} In view of Step 1 and Step 2, we can remove  $m \in \lbrace 4, \ldots, m_0\rbrace$,  $n-m\ge 4$, \EEE atoms   to obtain a configuration $C_{n-m}$ without bridging atoms. We denote the sub-configuration consisting of the removed atoms by $C_{m}$. By  Lemma \ref{lemma: bridging}(a) we get that $C_{n-m}$ is an optimal configuration and $\mathcal{Q}(C_{n-m}) \ge q_{\rm sat}^{n-m}$.    Now suppose by contradiction that $\mathcal{Q}(C_{n-m}) > q_{\rm sat}^{n-m}$, i.e., $\mathcal{Q}(C_{n-m}) \ge q_{\rm sat}^{n-m} + 2$. Recall by Lemma \ref{lemma: bridging}(a) that  $\mathcal{Q}(C_{m}) \geq \phi(m+1)-1$. Therefore,   since \EEE  $n \geq m, n-m \ge 4$, \EEE by Lemma \ref{lemma : properties phi}(i),(vi)   we derive 
\begin{align*}
\mathcal{Q}(C_n) = \mathcal{Q}(C_m) + \mathcal{Q}(C_{n-m}) \geq \phi(m+1) - 1 +\phi(n-m)+ 2 \geq \phi(n+1) + 3 \ge \phi(n) +2.
\end{align*}
 This contradicts $\mathcal{Q}(C_n)=q_\mathrm{sat}^n$ and shows $\mathcal{Q}(C_{n-m}) = q_{\rm sat}^{n-m}$. It remains to prove $\mathcal{Q}(  C_n \EEE ) \ge \phi(n- 26 \EEE ) + 2$. To see this, we recall by Lemma \ref{lemma: bridging}(a) that 
\begin{align}\label{eq: repeat}
\mathcal{Q}(C_n)  \geq \phi(m+1) - 1 +\phi(n-m).
\end{align} 
 In view of Lemma  \ref{lemma : properties phi}(i),(ii), for $4 \le m \le  26 \EEE $    we derive  $ \mathcal{Q}(C_n)  \ge \phi(m+1) - 1 + \phi(n- 26 \EEE ) - m \, {\rm mod}\, 2$.       Then the result follows by checking Table \ref{table2}. If $m \ge  27 \EEE $, we suppose without restriction that $m+1 \le n-m$, and \EEE we use Lemma  \ref{lemma : properties phi}(v) for   $n_1 = m+1$, $n_2 = n-m$, $t = n_1 -  27 \EEE $  to \EEE  get 
$$\phi(m+1) - 1 +\phi(n-m) \ge \phi(m+1-t) + \phi(n-m+t) -  7 \EEE = \phi( 27 \EEE ) + \phi(n- 26 \EEE ) -   7 \EEE = \phi(n- 26 \EEE  ) + 2,$$
where the last step follows from   $\phi(27) = 9$, \EEE see Table \ref{table2}. In view of \eqref{eq: repeat}, this concludes the proof. 
\end{proof}
\begin{proof}[Proof of Lemma \ref{lemma : smallness of A}] We \EEE prove the statement by induction. For $n \leq 50$ the statement is clearly true.   Now let $n\geq 51$.   We assume  that the statement holds for all $m<n$ and we proceed to prove the statement for $n$. Let $C_n$ be a $q_\mathrm{sat}^n$-optimal configuration without bridging atoms. If $C_n$ is not connected, we can apply Lemma \ref{lemma: removi}(a) to remove $k \in \{1,2\}$ $0$-bonded atoms to obtain a $q_\mathrm{sat}^{n-k}$-optimal configuration. By the induction hypothesis, this new configuration contains at most $50$ non-equilibrated atoms and therefore also $\#\mathcal{A}(X_n) \leq 50$.  Therefore, it is not restrictive to assume that $C_n$ is connected. \EEE We divide the proof into several steps.  Recall the definition of  $I_{\rm ac}$ below \eqref{eq: bridgi}.

%

\noindent \emph{Step 1: $\mathcal{A}(X_n) \cap I_\mathrm{ac}=\emptyset$ and $\#\mathcal{A}(X_n) \leq 3\eta/2$.}  We first prove that $\mathcal{A}(X_n) \cap I_{\rm ac}=\emptyset$.  By \EEE Lemma \ref{lemma: bridging}(b) there  holds $I_\mathrm{ac}\cap X_n^- =\emptyset$  and $ I_\mathrm{ac}\cap X_n^+$ consists of $1$-bonded atoms only.  Since $1$-bonded atoms do not have bond angles, this implies $I_\mathrm{ac} \cap \mathcal{A}(X_n) = \emptyset$.

Next, we prove $\#\mathcal{A}(X_n) \leq  3\eta/2$. By   Theorem \ref{theorem: min-en2}, Remark \ref{rem: repulsionsfree}, Lemma  \ref{RemarkPolygon}(b), and Lemma \ref{lemma : Charge Energybound}(a)(i),(ii)       all squares are regular. Thus, if    $x \in  \mathcal{A}(X_n) \EEE$ is contained in a polygon, this  polygon cannot be a square. Moreover, $x \in X_n^+$ by Lemma \ref{lemma : Charge Energybound}(a)(i),(iv).  Denote now by $f_j$ the number of  elementary  $j$-gons in the bond graph and recall that $f_j = \emptyset$ for $j$ odd or $j \le 3$,  see Lemma \ref{RemarkPolygon}(a).   As $ \mathcal{A}(X_n) \EEE  \cap I_{\rm ac} = \emptyset$,  each $x \in  \mathcal{A}(X_n) \EEE$ is contained in at least one $j$-gon, $j \geq 6$. We also observe that,  due to alternating charge distribution, there can be at most $j/2$ atoms of positive charge in each $j$-gon. This implies  by  \eqref{Excess} \EEE
\begin{align*}
 \#  \mathcal{A}(X_n) \EEE \leq \frac{1}{2}\sum_{j \ge 6} jf_j  \leq \frac{3}{2}  \sum_{j \ge 6} (j-4)f_j =  \frac{3}{2}  \sum_{j \geq 4} (j-4)f_j = \frac{3}{2}\eta.
\end{align*}

 \noindent \emph{Step 2: Preliminaries.} \EEE  From now on, we suppose by contradiction that $\#\mathcal{A}(X_n) \geq 51$. This implies $\eta \ge 34$ by Step 1.    Note by Lemma \ref{lemma : charge interior} \EEE that  $\#I_2+2\#I_3+3\#I_4=2d-m$  and \EEE $\#I_\mathrm{ac}^\mathrm{ext}\geq m$  for some $m \geq 4$. We \EEE will use the following fact several times: the estimate
\begin{align}\label{charge: counter}
\mathcal{Q}(C_n) \geq  \phi(n-d-26 - m)+6 
\end{align}
 along with $\eta \geq 34$ leads to a contradiction: in fact,  by \EEE Lemma \ref{lemma : card positively charged shell} and Lemma \ref{lemma : charge interior} there holds $n-d \in n+4-4\mathcal{Q}(C_n) + \eta +2\#I_\mathrm{ac}^\mathrm{ext} + 2\mathbb{N}_0$ with  $\eta \geq 34$ and  $\#I_\mathrm{ac}^\mathrm{ext}\geq m \ge 4$. \EEE Using   Lemma \ref{lemma : properties phi}(ii) \EEE we obtain
\begin{align*}
\mathcal{Q}(C_n) \geq \phi(n+8 -4(\mathcal{Q}(C_n)-2)) +6.
\end{align*}
 In view of Lemma \ref{lemma : computation of qsatn} applied for $x=\mathcal{Q}(C_n)-2$, we get  the contradiction $\phi(n)-2=\mathcal{Q}(C_n)-2 \ge \phi(n)$, \EEE  see \eqref{eq: NNN} for a similar argument. \EEE

  We  denote by $C_{n-d-m}$ the  optimal \EEE configuration obtained from  Lemma \ref{lemma : charge interior}(b) \EEE satisfying  $\mathcal{A}(X_{n-d-m}) =\mathcal{A}(X_n^{\mathrm{a},\mathrm{bulk}})$,   $X_{n-d-m} \subset X_n$ up to $0$-bonded atoms, \EEE and
\begin{align}\label{charge: qn-d-m}
\mathcal{Q}(C_n) =\mathcal{Q}(C_{n-d-m})+m  \ge \phi(n-d-m) + m. \EEE
\end{align}
 Note by   $n-4\phi(n) \ge -16$ for $n \ge 51$, \EEE  (\ref{def: phi}), Lemma \ref{lemma : card positively charged shell}, and Lemma \ref{lemma : charge interior} that $n-d-m \geq n-4\phi(n) +4 +\eta +m \geq 26$.\EEE

\noindent \emph{Step 3: $C_{n-d-m}$ satisfies $\mathcal{Q}(C_{n-d-m})=\phi(n-d-m)$.}  \EEE  Assume by contradiction that  $\mathcal{Q}(C_{n-d-m}) \geq \phi(n-d-m)+2$. This together with \eqref{charge: qn-d-m}, $m \ge 4$, and Lemma \ref{lemma : properties phi}(ii) \EEE leads to $\mathcal{Q}(C_n) \geq  \phi(n-d-m)+6 \ge \phi(n-d-26-m)+6$ which yields a contradiction by \eqref{charge: counter}.  In view of Lemma \ref {lemma: removi}(a),  up to removing at most two $0$-bonded atoms, we can thus assume that $C_{n-d-m}$ is also connected, and that \EEE $X_{n-d-m} \subset X_n$.

\noindent \emph{Step 4: $\#\mathcal{A}(X_n) \leq \#\mathcal{A}(X_n^{\mathrm{a},\mathrm{bulk}})+1$,  in  particular $\mathcal{A}(X_n) \cap \partial X_n^\mathrm{a} = \emptyset$.}  Recall the definition of  $\partial X_n^{a}$ and $X_n^{a,\mathrm{bulk}}$ below \eqref{def: Ik}. \EEE Suppose by contradiction $\#\mathcal{A}(X_n) \geq \#\mathcal{A}(X_n^{\mathrm{a},\mathrm{bulk}})+2$.  By   Lemma \ref{lemma : charge interior}(a) we have $m \ge 6$. Thus, by \eqref{charge: qn-d-m} \EEE and Lemma  \ref{lemma : properties phi}(ii) we get $\mathcal{Q}(C_n) \geq \phi(n-d-m)+6\ge \phi(n-d-26-m)+6$ which yields a contradiction by \eqref{charge: counter}. In particular, this also shows $\mathcal{A}(X_n) \cap \partial X_n^\mathrm{a} = \emptyset$.  Indeed, \EEE otherwise we would get $\#(\mathcal{A}(X_n) \cap \partial X_n^\mathrm{a}) \geq 2$ by Lemma \ref{RemarkPolygon} which contradicts  \EEE $\#\mathcal{A}(X_n) \leq \#\mathcal{A}(X_n^{\mathrm{a},\mathrm{bulk}})+1$. 

\noindent \emph{Step 5: $C_{n-d-m}$ contains a bridging atom.}  Suppose by contradiction that $C_{n-d-m}$ does not contain a bridging atom. \EEE By Step 3 and Step 4 we get that \EEE
\begin{align}\label{ineq : A}
 51 \le \EEE \#\mathcal{A}(X_n) \leq \#\mathcal{A}(X_n^{\mathrm{a},\mathrm{bulk}})+1 =  \# \mathcal{A}(X_{n-d-m}) + 1 \EEE
\end{align}
 and $\mathcal{Q}(C_{n-d-m}) =\phi(n-d-m)$.  (The equality in \eqref{ineq : A} follows from Lemma \ref{lemma : charge interior}(b).)  \EEE By the induction hypothesis applied on $C_{n-d-m}$, we  get  $\#\mathcal{A}(X_{n-d-m}) \le 50$. This implies \EEE 
  equality in \EEE \eqref{ineq : A}, i.e., ~$50  = \EEE \#\mathcal{A}(X_n)-1    = \#\mathcal{A}(X_{n-d-m}) \EEE $.  Note that this together with Step 1  yields $\eta(C_{n-d-m}) \geq 34$.  By Step  4 (applied for   $X_{n-d-m}$ in place of $X_n$; note that $C_{n-d-m}$ is connected, $q_{\rm sat}^{n-d-m}$-optimal, without bridging atoms, and $\eta(C_{n-d-m}) \geq 34$) \EEE we also get $\mathcal{A}(X_{n-d-m})\cap \partial X_{n-d-m}^\mathrm{a} = \emptyset$.  
  
  Consider the unique  \EEE $x \in \mathcal{A}(X_n) \setminus \mathcal{A}(X_n^{\mathrm{a},\mathrm{bulk}})$. Since $X_n$ does not contain bridging atoms and all elements in $\mathcal{A}(X_n)$ are at least $2$-bonded, $x$ is contained in an elementary polygon $P$. By Lemma \ref{RemarkPolygon} we find  some \EEE $y \in \mathcal{A}(X_n) \cap (P \setminus \{x\})$. We now show that $y$ is a bridging atom in $X_{n-d-m}$.   Note first that  $y \in \mathcal{A}(X_{n-d-m})$ \EEE since otherwise there would hold strict inequality in \eqref{ineq : A}. Observe by Lemma \ref{lemma : Charge Energybound}(a)(iv),(v) that  $x$  is $2$-bonded.  Then, it is also clear that $P\cap X_{n-d-m}$ is contained in the boundary of the exterior face  of $X_{n-d-m}$.  Therefore, $y$ has to be an acyclic atom  since otherwise we obtain the contradiction $\mathcal{A}(X_{n-d-m})\cap \partial X_{n-d-m}^\mathrm{a} \neq \emptyset$.  As also $y$ is is $2$-bonded, see Lemma \ref{lemma : Charge Energybound}(a)(iv),(v), this \EEE shows that the configuration $C_{n-d-m}$ contains the bridging atom $y$.  This contradicts our assumption that $C_{n-d-m}$ does not contain bridging atoms. Summarizing, $C_{n-d-m}$ needs to contain a bridging atom. \EEE

\noindent \emph{Step 6: $\#\mathcal{A}(X_n) \leq 50$.}  Recall that we have supposed by contradiction that $\#\mathcal{A}(X_n) \geq 51$ which implies $\eta \geq 34$ by Step 1.  By Lemma  \ref{lemma: removi}(b) we obtain $\mathcal{Q}(C_{n-d-m}) \geq \phi(n-d-m-26)+2$. This along with \eqref{charge: qn-d-m}  and $m \ge 4$ \EEE  gives
$$
\mathcal{Q}(C_n) \geq \phi(n-d-26-m)+m+2 \geq \phi(n-d-26-m)+6,
$$ 
which leads to a contradiction by \eqref{charge: counter}. 
\end{proof}

\begin{remark}[Theorem \ref{theorem : GeometryGroundstate}(a) without bridging atoms]\label{remark: bridging} By increasing $n_0$ independently  of  $n$, the configuration in Theorem \ref{theorem : GeometryGroundstate}(a) can be constructed in such a way that it is  a  connected, saturated subset of the square lattice without  bridging atoms: the configuration $C_n^{\mathrm{max}}$  in Remark \ref{remark: bridging-helpi}  may contain at most $ 49  $  bridging atoms. In fact, otherwise  one could apply Lemma \ref{lemma: bridging}(a)  and Lemma \ref{lemma : properties phi}(vi)   $50$ times  to obtain an  estimate of the form (see \eqref{eq: repeati1} for a similar argument)
\begin{align}\label{eq: repeati3}
\mathcal{Q}(C_n^{\mathrm{max}}) \geq \phi(n_{\mathrm{max}}+50) + 50 \geq \phi(n_\mathrm{max}) +50,
\end{align}
which contradicts $\mathcal{Q}(C_n^{\mathrm{max}}) \le \phi(n_{\mathrm{max}}) + 49$. In a similar fashion, repeating the arguments in \eqref{eq: repeati2}, one of the components connected through a bridging atom has at most $m_0$ atoms as otherwise \eqref{eq: repeati3} holds. Thus, it suffices to remove at most $49m_0$ atoms to get the desired configuration without bridging atoms.  
\end{remark}

\subsection{Wulff-shape emergence for  $q_\mathrm{sat}^n$-optimal configurations} \label{subsection : geometry qsat2}

This subsection is devoted to the proof of Theorem \ref{theorem : GeometryGroundstate}(b).  Our strategy is to apply an anisotropic quantitative isoperimetric inequality on   suitable interpolations of the  configurations. (A similar idea can be found in \cite{cicalese}.) \EEE  We first state two lemmas which relate the $\infty$-perimeter of the interpolation to the net charge and then prove the main result. \EEE 
 Recall that configurations are called saturated if the atoms of the minority phase are $4$-bonded. We also recall \eqref{Excess}.

\begin{lemma}[Control on $\eta$]\label{lemma:etacontrol} 
 Let $N_1 \in \mathbb{N}$. \EEE  Let $C_n$ be a  connected and \EEE saturated sub-configuration of the square lattice \EEE without bridging atoms satisfying
\begin{align}\label{ineq:chargeCneta}
0\leq \mathcal{Q}(C_n) \leq \phi(n)+N_1.
\end{align}
Then there exists $N_2=N_2(N_1) >0$ such that we can add  $  k \EEE \le N_2$ atoms to $C_n$ to obtain a  connected, \EEE saturated sub-configuration $C_{n+  k \EEE }$ of the square lattice   without bridging atoms satisfying \EEE   $\mathcal{Q}(C_{n+  k \EEE }) \ge 0$ and $\eta(C_{n+  k \EEE })=0$. \EEE
\end{lemma}

 We defer the proof to the end of the subsection. \EEE A key ingredient for the proof of Theorem \ref{theorem : GeometryGroundstate}(b) is that,  for (special) saturated configurations, the net charge coincides with the $\infty$-perimeter of a suitable interpolation of the  configuration.  We define  $D(x) := \left\{|x-\cdot|_1 \leq 1\right\}$ for $x \in \mathbb{Z}^2$. For a saturated configuration $C_n$ with $\mathcal{Q}(C_n) \ge 0$ we  introduce the interpolation
\begin{align}\label{def : A(Cn)}
A(C_n) :=  \bigcup_{q_i=-1} D(x_i)
\end{align}
and define the   $\infty$-perimeter by
\begin{align}\label{ineq : infty charge}
\mathrm{Per}_\infty(A(C_n))   :=  \int_{\partial A(C_n)}|\nu_{\partial A(C_n)}|_\infty \, \mathrm{d}\mathcal{H}^1 =  \tfrac{1}{\sqrt{2}} \mathcal{H}^1(\partial A(C_n)).
\end{align}
 Here, $\nu_{\partial A(C_n)}$ is the  outer \EEE  unit normal to the set $A(C_n)$. The second identity in \eqref{ineq : infty charge} is elementary since by   definition  \EEE  there holds $|(\nu_{\partial A(C_n)})_1(x)|=|(\nu_{\partial A(C_n)})_2(x)|=1/\sqrt{2}$ for   $\mathcal{H}^1$-almost every point $x \in \partial A(C_n)$.

\begin{lemma}\label{lemma: extra}
 Let $C_n$ be a   saturated,  connected sub-configuration of the square lattice which does not contain bridging atoms and satisfies $\eta=0$  as well as $\mathcal{Q}(C_n) \ge 0$. \EEE Then  $\mathrm{Per}_\infty(A(C_n)) =  2\mathcal{Q}(C_n)-2$.
 \end{lemma}

 We again defer the proof and proceed to show Theorem \ref{theorem : GeometryGroundstate}(b). \EEE

\begin{proof}[Proof of Theorem \ref{theorem : GeometryGroundstate}(b)]  Our strategy is to apply an anisotropic quantitative isoperimetric inequality on $A(C_n)$, defined in \eqref{def : A(Cn)}.  To this end, we will modify a given $q_{\rm sat}^n$-optimal configuration such that Lemma \ref{lemma: extra} is applicable.  We split the proof into three steps. In the following, $c$ denotes a universal constant which may vary from line to line, but is independent of $n$.

\noindent\emph{Step 1: Some elementary facts.}  For $n \in \mathbb{N}$ we   set
\begin{align*}
\underline{n} = \max \{  N\leq n:  \,  N=1+2k^2+2k, \,  k\in \mathbb{N} \}.
\end{align*}
\EEE
A careful inspection of  \eqref{eq: different repr} yields  for all   $n \in \mathbb{N}$  
\begin{align}\label{charge:diamon}
{\rm (i)} \ \ \phi(\underline{n}) \le \phi(n) \le \phi(\underline{n}) + 3, \ \ \ \ \ \ {\rm (ii)} \ \ \underline{n}- \phi(\underline{n}) \le {n}- \phi({n})  \le \underline{n}- \phi(\underline{n}) + c\sqrt{n}. 
\end{align}
Let $C'_n$ be a saturated configuration with $\mathcal{Q}(C_n') \ge 0$. We find $\mathcal{L}^2(A(C'_n)) = 2 \# \lbrace q_i =-1 \rbrace$ since $\mathcal{L}^2(D(x)) = 2$ for all $x \in \mathbb{Z}^2$. Thus, $\mathcal{L}^2(A(C'_n)) = n-\mathcal{Q}(C_n')$.  We let  $W_r : =\{|\cdot|_1 \leq r\}$ and note that 
\begin{align}\label{eq: strangeW}
\mathcal{L}^2\big(W_{r(C'_n)}\big) = 2r(C'_n)^2 = \mathcal{L}^2(A(C'_n)), \ \ \ \ \ \text{where} \ \  r(C'_n) := \sqrt{\tfrac{1}{2}\EEE(
n-\mathcal{Q}(C_n'))}.
\end{align}

\noindent\emph{Step 2: Modification of  $q_{\rm sat}^n$-optimal configurations.} Let $C_n$ be a  $q_{\rm sat}^n$-optimal configuration.   By Remark \ref{remark: bridging},  after removing at most \EEE $n_0$ atoms, we  get  a saturated, connected sub-configuration  of the square lattice $C_{n-m}$, $0 \le m \le n_0$,  whose  \EEE bond graph does not contain any bridging atoms.  By  \EEE Lemma \ref{lemma : properties phi}(i)  there holds \EEE
\begin{align*}
 0\leq\mathcal{Q}(C_{n-m}) \leq \phi(n ) + m \leq \phi(n-m)+2m \le \phi(n-m) + 2n_0. \EEE
\end{align*}
By Lemma \ref{lemma:etacontrol} we can add at most  $N_2(2n_0)$  atoms to obtain a   saturated, connected sub-configuration of the square lattice without bridging atoms satisfying $\eta = 0$. For simplicity, we denote the configuration again by $C_n$ and observe that \EEE
\begin{align}\label{charge:QCn} 
|\mathcal{Q}(C_n)-\phi(n)|\leq c.
\end{align}
 In Step 3 below we will apply Lemma \ref{lemma: extra} on $C_n$. \EEE The newly constructed configuration and the original configuration differ in cardinality and net charge by a finite constant independent of $n$, i.e., it suffices to prove the statement for  this \EEE new configuration.

\noindent\emph{Step 3:  Application of the quantitative isoperimetric inequality.} Recall that the diamond $D_{\underline{n}}$  with alternating charge distribution \EEE is a $q_{\mathrm{sat}}^{\underline{n}}$-optimal configuration, see \eqref{def :Dn} and  Subsection \ref{subsection : upper bound qsat}.  By \eqref{charge:diamon}(ii), \eqref{eq: strangeW}, and \eqref{charge:QCn}  this implies $r(D_{\underline{n}})^2 - c \le r(C_n)^2 \le r(D_{\underline{n}})^2 +c\sqrt{n}$. Therefore, we get 
\begin{align}\label{ineq:perinftydiamond}
{\rm (i)} \ \ \mathrm{Per}_\infty(A(D_{\underline{n}})) \leq \mathrm{Per}_\infty(W_{r_n}) + c, \ \ \ \ \ {\rm (ii)} \ \  \mathcal{L}^2(W_{r_n} \triangle A(D_{\underline{n}})) \leq c\sqrt{n},
\end{align}
where for brevity we have set $r_n = r(C_n)$.  Using the definition $r_n = \sqrt{\tfrac{1}{2}\EEE(n-\mathcal{Q}(C_n))}$,  \eqref{def: phi}, and \eqref{charge:QCn}  one can also  check  that \EEE for $n$ sufficiently large there holds
\begin{align}\label{ineq:perinftydiamond2}
 r_n \in (\tfrac{1}{2}\sqrt{n/2},\sqrt{n/2}). \EEE
\end{align}
By the quantitative anisotropic isoperimetric inequality \cite[Theorem 1.1]{Figalli-Maggi-Pratelli}   (applied for the convex set $K = W_1$)  there exists  a translation $\tau \in \mathbb{R}^2$ such that   
$$
 (\mathcal{L}^2(W_{r_n}))^{-2} \mathcal{L}^2(W_{r_n}\triangle (A(C_n)+\tau))^2 \le c\, \mathrm{Per}_\infty(W_{r_n})^{-1}  \big(\mathrm{Per}_\infty(A(C_n))-\mathrm{Per}_\infty(W_{r_n})\big).
$$
As by \eqref{ineq : infty charge}, \eqref{eq: strangeW}, and \eqref{ineq:perinftydiamond2} there holds $\mathcal{L}^2(W_{r_n}) = 2r_n^2 \le  n\EEE $ and $\mathrm{Per}_\infty(W_{r_n}) = 4r_n \ge  2\sqrt{n/2} \EEE $, we get  
\begin{align}\label{ineq : volume perimeter}
\mathcal{L}^2(W_{r_n}\triangle (A(C_n)+\tau))^2  \le  c n^{3/2} \, \big(\mathrm{Per}_\infty(A(C_n))-\mathrm{Per}_\infty(W_{r_n})\big).
\end{align}
To estimate the right hand side, we apply  Lemma \ref{lemma: extra} (on both $C_n$ and $D_{\underline{n}}$),  \eqref{charge:diamon}(i), \eqref{charge:QCn}, \eqref{ineq:perinftydiamond}(i), and $\phi(\underline{n}) = \mathcal{Q}(D_{\underline{n}})$   to get
\begin{align}\label{ineq:comparisonAW}
\mathrm{Per}_\infty(A(C_n))-\mathrm{Per}_\infty(W_{r_n}) &\leq  \mathrm{Per}_\infty(A(C_n)) -\mathrm{Per}_\infty(A(D_{\underline{n}})) + c =  2\mathcal{Q}(C_n)-2\mathcal{Q}(D_{\underline{n}}) + c\notag \\&  \le  2 (\phi(n)- \phi(\underline{n})) +c\le c.  
\end{align}
 After translation  we may suppose that $X_n + \tau$ and $D_{\underline{n}}$ are both contained in $\mathbb{Z}^2$.  \EEE By the fact that $\mathcal{L}^2(D(x))=2$ for all $x \in \mathbb{Z}^2$,  (\ref{def : A(Cn)}),  \eqref{ineq:perinftydiamond}(ii),  (\ref{ineq : volume perimeter}),  and (\ref{ineq:comparisonAW}) we get
\begin{align*}
2\,\#\big(D_{\underline{n}}^- \triangle (X^-_n+\tau)\big) =\mathcal{L}^2(A(D_{\underline{n}}) \triangle (A(C_n)+\tau)) \leq\mathcal{L}^2(W_{r_n} \triangle (A(C_n)+\tau)) + \mathcal{L}^2(W_{r_n} \triangle A(D_{\underline{n}})) \leq Cn^{3/4}.
\end{align*}
It remains to note that $\#(D_{\underline{n}} \triangle (X_n+\tau)) \leq 5 \, \#(D_{\underline{n}}^- \triangle (X^-_n+\tau))$  since $C_n$ is saturated. \EEE 
  \end{proof}

 It remains to prove Lemma \ref{lemma:etacontrol} and Lemma \ref{lemma: extra}. 

\begin{proof}[Proof of Lemma \ref{lemma:etacontrol}]

 We first construct the configuration $C_{n+k}$. \EEE  By $X_{n}^{\mathrm{int}}$ we denote the union of the points  $\mathbb{Z}^2 \setminus X_n$ which are contained in the interior of a simple cycle of the bond graph of $X_n$. Let  \EEE 
$C_n^{\mathrm{int}}=(X_{n}^{\mathrm{int}},Q^{\mathrm{int}})$, where $Q^{\mathrm{int}}$ is chosen such that $C_n \cup C_n^{\mathrm{int}}$ has alternating charge  distribution. \EEE This is possible  since $X_n \cup X_n^{\mathrm{int}} \subset \mathbb{Z}^2$.  Define $k:= \#X_{n}^{\mathrm{int}}$ and  $C_{n+k} := C_n \cup C_n^{\mathrm{int}}$.  Note that $C_{n+k}$ is still connected, saturated and does not contain bridging atoms. Moreover, $C_{n+k}$ has alternating charge distribution, is subset of the square lattice,  and  $\eta(C_{n+k})=0$.  Therefore, the configuration satisfies Lemma \ref{lemma : Charge Energybound}(a)(i),(ii). \EEE Hence, Lemma \ref{lemma : Charge Energybound}(b),(a)(iii) implies \EEE
\begin{align}\label{ineq:QCn+k}
\mathcal{Q}(C_{n+k}) \geq \phi(n+k).
\end{align}
 It remains to control $k$.  Choose $N_2 \in \mathbb{N}$  depending only on $N_1$  such that 
\begin{align}\label{eq:choiceC2}
\phi(  N_2 + 1  ) \geq N_1+3.
\end{align}
 We now show $k \le N_2$. Assume  by contradiction that $k \ge N_2+1$.   \EEE  We claim that
\begin{align}\label{eq:chargeCnint}
\mathcal{Q}(C_{n}^{\mathrm{int}}) \leq -\phi(k).
\end{align}
 We postpone the proof of this estimate to the end and proceed to establish a contradiction.  Using (\ref{ineq:chargeCneta}) and \eqref{ineq:QCn+k}-(\ref{eq:chargeCnint}) we obtain
\begin{align}\label{ineq:lastcharge}
\phi(n+k) \leq \mathcal{Q}(C_{n+k})=\mathcal{Q}(C_n) + \mathcal{Q}(C_{n}^{\mathrm{int}}) \leq \phi(n)+N_1-\phi(k).
\end{align}
\EEE Now by  $k \ge N_2+1$, Lemma \ref{lemma : properties phi}(i),(ii), and    \eqref{eq:choiceC2} there holds $\phi(k)\geq \phi(N_2+1)-1 \geq N_1+2$. This together with (\ref{ineq:lastcharge}) leads to $\phi(n+k) \leq \phi(n)-2$  which contradicts  Lemma \ref{lemma : properties phi}(i),(ii). \EEE

 It remains to prove \eqref{eq:chargeCnint}. \EEE By construction there holds $X_{n}^{\mathrm{int}} \subset \mathbb{Z}^2$. In particular, the configuration is repulsion-free and all bonds have unit length. In view of Lemma \ref{lemma : Charge Energybound}(b),(a)(iii) \EEE (note that it is applicable even though the configuration has negative   net \EEE charge: consider $ (X_{n}^{\mathrm{int}},-Q^{\mathrm{int}})$), it remains to show that each  $x \in (X_{n}^{\mathrm{int}})^+$  has four neighbors in $X_{n}^{\mathrm{int}}$. \EEE 
If there existed some $x \in (X_{n}^{\mathrm{int}})^+$  which is at most $3$-bonded, \EEE  then there would exist $y \in  X^-_n \EEE$ such that $|x-y|=1$. Since every point on the square lattices has at most four neighboring points on the square lattice, this would imply that  $y$ is not $4$-bonded   in $X_n$ \EEE and thus \EEE $X_n$ would not have been a saturated configuration:  a contradiction.
\end{proof}
\begin{proof}[Proof of Lemma \ref{lemma: extra}]
 We prove the claim by induction over $k=\#X_n^-$.  Let $k=1$. As $C_n$ is connected, this implies $n=5$ and therefore \EEE $\mathcal{Q}(C_n)=3$ as well as $\mathrm{Per}_\infty(A(C_n))=4$. Thus, the claim is true for $k=1$.  Now assume that we  have \EEE proved the statement for all $l <k$. We proceed to prove the claim for $k$. Consider  the leftmost of the uppermost \EEE negatively charged atoms and denote it by $x_0$. There are three cases to consider (see Fig.~\ref{fig : infty charge proof}):

\begin{itemize}
\item[(a)] $D(x_0)$ shares  exactly one \EEE side with another square $D(x_1)$, $x_1 \in X^-_n$.
\item[(b)] $D(x_0)$ shares exactly two adjacent sides with  two \EEE other squares $D(x_1),D(x_2)$, $x_1,x_2 \in X^-_n$.  The atom which is contained in both of these sides is contained in four squares.  \EEE
\item[(c)] $D(x_0)$ shares  exactly two adjacent sides with  two \EEE other squares $D(x_1),D(x_2)$, $x_1,x_2 \in X^-_n$.  The atom which is contained in both of these sides is contained in only three squares.  \EEE
\end{itemize}


This is a consequence of the choice of $x_0$, the absence of bridging atoms, the connectedness of $C_n$, and $\eta=0$. In fact, by  the \EEE choice of the point $x_0$, the top corner cannot be contained in another square. Next, at least  one side of $D(x_0)$ has to be shared with another square \EEE since otherwise the configuration would not be connected  or there would be a 2-bonded atom. The latter in turn would be a bridging atom or would lead to  $\eta \ge 2$. Finally, if $D(x_0)$ shares exactly two sides with other squares, the bottom corner is either contained in three or four squares.  
 \EEE

\begin{figure}[htp]
\begin{tikzpicture}[scale=.6]
\draw[dashed, ultra thin] (0,-1)--++(-135:.5);
\draw[dashed, ultra thin] (1,-2)--++(-135:.5);

\draw[dashed, ultra thin] (2,-1)--++(-45:.5);
\draw[dashed, ultra thin] (1,-2)--++(-45:.5);
\draw[fill=white](0,0)circle(.07);
\draw(0,0) node[anchor=south]{$x_0$};
\foreach \j in {0,1,2,3}{
\draw[fill=black](90*\j:1) circle(.07);
\draw[ultra thin](90*\j:1)--(90*\j+90:1);
}
\begin{scope}[shift={(1,-1)}]
\draw[fill=white](0,0)circle(.07);
\foreach \j in {0,1,2,3}{
\draw[fill=black](90*\j:1) circle(.07);
\draw[ultra thin](90*\j:1)--(90*\j+90:1);
}
\end{scope}

\begin{scope}[shift={(12,0)}]
\draw(0,0) node[anchor=south]{$x_0$};
\draw[dashed, ultra thin] (-1,-2)--++(-135:.5);
\draw[dashed, ultra thin] (-2,-1)--++(-135:.5);
\draw[dashed, ultra thin] (2,-1)--++(-45:.5);
\draw[dashed, ultra thin] (1,-2)--++(-45:.5);
\draw[fill=white](0,0)circle(.07);
\foreach \j in {0,1,2,3}{
\draw[fill=black](90*\j:1) circle(.07);
\draw[very thin,gray](90*\j:1)--(90*\j+90:1);
}
\begin{scope}[shift={(1,-1)}]
\draw[fill=white](0,0)circle(.07);
\foreach \j in {0,1,2,3}{
\draw[fill=black](90*\j:1) circle(.07);
\draw[very thin,gray](90*\j:1)--(90*\j+90:1);
}
\end{scope}

\begin{scope}[shift={(-1,-1)}]
\draw[fill=white](0,0)circle(.07);
\foreach \j in {0,1,2,3}{
\draw[fill=black](90*\j:1) circle(.07);
\draw[very thin,gray](90*\j:1)--(90*\j+90:1);
}
\end{scope}

\end{scope}

\begin{scope}[shift={(6,.5)}]
\draw(0,0) node[anchor=south]{$x_0$};
\draw[dashed, ultra thin] (-1,-2)--++(-135:.5);
\draw[dashed, ultra thin] (0,-3)--++(-135:.5);
\draw[dashed, ultra thin] (0,-3)--++(-45:.5);
\draw[dashed, ultra thin] (-2,-1)--++(-135:.5);
\draw[dashed, ultra thin] (2,-1)--++(-45:.5);
\draw[dashed, ultra thin] (1,-2)--++(-45:.5);
\draw[fill=white](0,0)circle(.07);
\foreach \j in {0,1,2,3}{
\draw[fill=black](90*\j:1) circle(.07);
\draw[very thin,gray](90*\j:1)--(90*\j+90:1);
}
\begin{scope}[shift={(1,-1)}]
\draw[fill=white](0,0)circle(.07);
\foreach \j in {0,1,2,3}{
\draw[fill=black](90*\j:1) circle(.07);
\draw[very thin,gray](90*\j:1)--(90*\j+90:1);
}
\end{scope}

\begin{scope}[shift={(-1,-1)}]
\draw[fill=white](0,0)circle(.07);
\foreach \j in {0,1,2,3}{
\draw[fill=black](90*\j:1) circle(.07);
\draw[very thin,gray](90*\j:1)--(90*\j+90:1);
}
\end{scope}

\begin{scope}[shift={(0,-2)}]
\draw[fill=white](0,0)circle(.07);
\foreach \j in {0,1,2,3}{
\draw[fill=black](90*\j:1) circle(.07);
\draw[very thin,gray](90*\j:1)--(90*\j+90:1);
}
\end{scope}

\end{scope}
\end{tikzpicture}
\caption{The three cases  (a)-(c). \EEE   Case (a) holds up to reflection.  }
\label{fig : infty charge proof}
\end{figure}
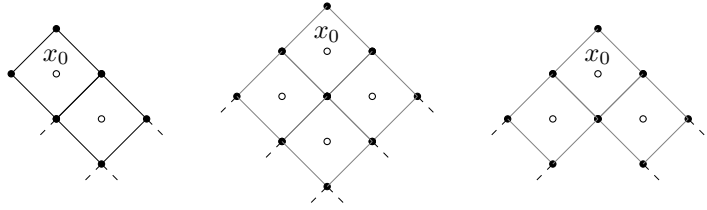
\EEE
\noindent\emph{Proof of (a):} We remove the two corners of positive charge not contained in any other square as well as $x_0$ to obtain a configuration  $C_{n-3}$ such that $\# X_{n-3}^- = k-1$. \EEE Note that
\begin{align}\label{eq : charge induction}
\mathcal{Q}(C_{n-3}) = \mathcal{Q}(C_n) -1. 
\end{align}
Since we have removed three sides of the square $D(x_0)$ but count a side of the square $D(x_1)$ to the perimeter,   we obtain
\begin{align} \label{eq : Per induction}
\mathrm{Per}_\infty(A(C_{n-3})) = \mathrm{Per}_\infty(A(C_n))-2.
\end{align}
 It is elementary to check that  the new configuration is still saturated, connected, satisfies   $\eta(C_{n-3})=0$, and does not contain any bridging atom. \EEE Therefore, we can apply the induction  hypothesis \EEE together with (\ref{eq : charge induction}) and (\ref{eq : Per induction}) to obtain
\begin{align*}
\mathrm{Per}_\infty(A(C_n)) = \mathrm{Per}_\infty(A(C_{n-3}))+2= 2\mathcal{Q}(C_{n-3})=2\mathcal{Q}(C_n)-2.
\end{align*}

\noindent\emph{Proof of (b):} We remove the  top corner of positive charge as well as $x_0$ to obtain a configuration   $C_{n-2}$ such that $\# X_{n-2}^- = k -1$. \EEE Again, the new configuration is still saturated, connected, satisfies  $\eta(C_{n-2})=0$, and does not contain any bridging atom.  We get $\mathcal{Q}(C_{n-2}) = \mathcal{Q}(C_n)$ and, since we have removed two sides of the square $D(x_0)$ but  count one side of each square $D(x_1)$, $D(x_2)$ \EEE to the perimeter, we obtain $\mathrm{Per}_\infty(A(C_{n-2})) = \mathrm{Per}_\infty(A(C_n))$. This along with the induction  hypothesis \EEE yields 
\begin{align*}
\mathrm{Per}_\infty(A(C_n)) = \mathrm{Per}_\infty(A(C_{n-2}))= 2\mathcal{Q}(C_{n-2})-2=2\mathcal{Q}(C_n)  - 2. 
\end{align*}

\noindent\emph{Proof of (c):} We remove the  top  corner of positive charge as well as $x_0$ to obtain two sub-configurations $C_{m_1}$,$C_{m_2}$  which share only the bottom atom of $D(x_0)$. This implies $m_1+m_2=n-1$ and  
\begin{align}\label{eq : charge induction3}
\mathcal{Q}(C_{m_1})+\mathcal{Q}(C_{m_2}) = \mathcal{Q}(C_n) +1.
\end{align}
We again observe that the new configurations are saturated, connected, satisfy 
  $\eta(C_{m_1})=\eta(C_{m_2})=0$,  and do not contain any bridging atom. (Note that the configuration obtained by removing the top corner and $x_0$ would contain one bridging atom.) We have removed two sides of the square $D(x_0)$, but count one side of $D(x_1)$ to the perimeter of $C_{m_1}$ and count one side of $D(x_2)$ to the perimeter of $C_{m_2}$. This implies    \EEE  
\begin{align}\label{eq : Per induction4}
\mathrm{Per}_\infty(A(C_{m_1}))+\mathrm{Per}_\infty(A({C}_{m_2})) =  \mathrm{Per}_\infty(A(C_n)). \EEE
\end{align}
Note that $\#X_{m_1}^-,\# X_{m_2}^- <k$.  Thus, \EEE  we can use the induction  hypothesis \EEE along with (\ref{eq : charge induction3})-(\ref{eq : Per induction4}) to obtain
\begin{align*}
\mathrm{Per}_\infty(A(C_n)) &= \mathrm{Per}_\infty(A(C_{m_1}))+\mathrm{Per}_\infty(A({C}_{m_2})) =
 2\mathcal{Q}(C_{m_1})-2+ 2\mathcal{Q}({C}_{m_2})-2 =2\mathcal{Q}(C_n)  - 2. 
\end{align*}
This concludes the proof.
\end{proof}

\appendix 
\section{Proof of the boundary energy and  net \EEE  charge estimates \EEE }\label{appendix}

This section is devoted to the proofs of  Lemma \ref{LemmaBoundaryEnergy} and Lemma \ref{lemma : charge interior}.

\begin{proof}[Proof of  Lemma \ref{LemmaBoundaryEnergy}]
For convenience, we decompose the proof into three steps.

\noindent \emph{Step 1: Preliminary estimate.} \EEE Assume that $C_n$ is a connected ground state with no acyclic bonds.
 Suppose that  $\{x_1,\ldots,x_d\}$ are ordered such that \EEE  $x_i\in \mathcal{N}(x_{i+1})$, $i=1,\ldots,d$.  Here and in the following, we use \EEE the identification $x_{d+1} = x_1$ and $x_{0} = x_d$. For a 3-bonded atom $x_i$, denote
by $x_i^b \in X_n \setminus \lbrace x_{i-1},x_{i+1}\rbrace$ the atom that is connected to $x_i$ with the third bond. In a similar fashion, for  a 4-bonded atom $x_i$, denote
by $x_i^{b,1}, x_i^{b,2} \in X_n \setminus \lbrace x_{i-1},x_{i+1}\rbrace$ the two atoms which are bonded to $x_i$.  Recall \eqref{def:cell energy}-\eqref{eq: bnd2} and Remark \ref{rem: boundary energy}.  We first observe  that the boundary energy can be estimated by
\begin{align}\label{eq: boundary energy}
\mathcal{R}^{\mathrm{bnd}}(C_n) &\geq \sum_{i=1}^d \Big(\frac{1}{2}\big( V_{\mathrm{a}}(|x_i-x_{i+1}|)+V_{\mathrm{a}}(|x_i-x_{i-1}|)\big)+\frac{1}{4} V_{\mathrm{r}}(|x_{i+1}-x_{i-1}|) \Big)\notag \\&\quad+ \sum_{x_i \in I_3}\Big(  V_{\mathrm{a}}(|x_i-x_i^b|) + \frac{1}{4}\left( V_{\mathrm{r}}(|x_{i-1}-x_i^b|) + V_{\mathrm{r}}(|x_{i+1}-x_i^b|)  \right) \Big)\notag \\& \quad  + \sum_{x_i \in I_4}\Big(  V_{\mathrm{a}}(|x_i-x_i^{b,1}|) + V_{\mathrm{a}}(|x_i-x_i^{b,2}|)\notag \\&\qquad \ \ \ \ \ \ \ \ + \frac{1}{4}\big( V_{\mathrm{r}}(|x_{i-1}-x_i^{b,1}|) + V_{\mathrm{r}}(|x_{i}^{b,1}-x_i^{b,2}|) + V_{\mathrm{r}}(|x_{i+1}-x_i^{b,2}|) \big) \Big).
\end{align}
For a $3$-bonded atom $x_i$, denote by $\theta_i^1,\theta_i^2\in [0,2\pi]$ the angle between $x_{i-1},x_i, x_{i}^b$ and the angle between  $x_{i}^b,x_i, x_{i+1}$ respectively. For a $4$-bonded atom $x_i$, denote by $\theta_i^1,\theta_i^2, \theta_i^3\in [0,2\pi]$ the angle between $x_{i-1},x_i, x_{i}^{b,1}$, the angle between $x_{i}^{b,1},x_i, x_{i}^{b,2}$ and the angle between $x_{i}^{b,2},x_i, x_{i+1}$, respectively. Finally, we define $\delta : = (\#I_2+2\#I_3+3\#I_4)/d$ and note that $\delta \in [1,3]$ since $\# I_2 + \# I_3 +\#I_4  =d$. \EEE
We will prove that 
\begin{align}\label{boundarylowerbound-new}
\mathcal{R}^{\mathrm{bnd}}(C_n) & \ge -\delta d + \frac{1}{4}\Big(\sum_{x_i \in I_2} V_{\mathrm{r}}\Big(2\sin\Big(\frac{\theta_i}{2}\Big)\Big) + \hspace{-0.1cm} \sum_{x_i \in I_3}   \sum_{j=1}^2 V_{\mathrm{r}}\Big(2\sin\Big(\frac{\theta_i^j}{2}\Big)\Big)+ \hspace{-0.1cm}  \sum_{x_i \in I_4}   \sum_{j=1}^3 V_{\mathrm{r}}\Big(2\sin\Big(\frac{\theta_i^j}{2}\Big)\Big)  \Big) \EEE \nonumber \\& \geq  - \delta d + \frac{1}{4}\delta d \, V_{\mathrm{r}}\Big(2\sin\Big(\frac{\pi(d-2)}{2\delta d}\Big)\Big),
\end{align}
where the first inequality is strict if not all  lengths of boundary bonds \EEE  are equal to $1$.  We defer the proof of \eqref{boundarylowerbound-new} to Step 3.  At this stage, let us only point out that the specific definition of the boundary energy including also bulk bonds, see Remark \ref{rem: boundary energy} and Fig.~\ref{Fig:ConcaveAngle}, ensures that also $\theta_i^2$, $x_i \in I_4$, gives a contribution.

 \noindent\emph{Step 2: Proof of the statement.} \EEE  We now show that \eqref{boundarylowerbound-new} implies  the statement of the lemma. First,  we introduce \EEE  
\begin{align*}
\alpha(\delta) := \frac{\pi(d-2)}{2\delta d},
\end{align*}
 and observe that \EEE estimate \eqref{boundarylowerbound-new}  can be written as \EEE
\begin{align}\label{Ebndestimate}
\begin{split}
\mathcal{R}^{\mathrm{bnd}}(C_n) &\geq -\delta d +\frac{1}{4} \delta d \, V_{\mathrm{r}}\Big(2\sin\Big(\frac{\pi(d-2)}{2\delta d}\Big)\Big) =  \delta d\Big(\frac{1}{4} V_{\mathrm{r}} \big(2\sin (\alpha(\delta) )\big)-1\Big).
\end{split}
\end{align}
We obtain (\ref{BoundaryEnergyEstimate}) by minimizing the right hand side of (\ref{Ebndestimate}) with respect to $\delta$. To see this, \EEE set $\delta_0 = 2-\frac{4}{d}$. For $ 1 \le \EEE \delta \leq \delta_0$, we have $\alpha(\delta_0) = \frac{\pi}{4} \le \alpha(\delta) \le \frac{\pi}{2}$. \EEE By $[\mathrm{vii}]$  we get $V_{\mathrm{r}}( 2\sin(\alpha(\delta)))  = \EEE 0$ for all $ 1 \le \EEE \delta\leq \delta_0$. Therefore, we find
\begin{align}\label{deltageqdelta1}
\delta d\Big(\frac{1}{4} V_{\mathrm{r}} (2\sin (\alpha(\delta) ))-1\Big) = -\delta d\geq -\delta_0 d=-  (2d-4) \EEE
\end{align}
and we obtain   (\ref{BoundaryEnergyEstimate})  \EEE for $ 1 \le \EEE \delta\leq \delta_0$. Now for $\delta_0 < \delta \le 3$, \EEE we have $\alpha(\delta) < \alpha(\delta_0)$.  By  $[\mathrm{v}]$ \EEE we get
$$V_{\mathrm{r}}\big(2\sin(\alpha(\delta))\big) \geq V_{\mathrm{r}}\big(2\sin(\alpha(\delta_0))\big) + 2V_{\mathrm{r},-}^\prime \big(2\sin(\alpha(\delta_0)\big) \, \big(\sin(\alpha(\delta))-\sin(\alpha(\delta_0)\big).$$
 Then by \EEE the fact that $\sin(\theta)$ is concave for $\theta \in [0,\pi]$,  $V_{\rm r}(\sqrt{2}) = 0$ by [vii], \EEE $V_{\mathrm{r},-}^\prime(\sqrt{2}) < - 16/(\sqrt{2}\pi)<0$  by $[\mathrm{viii}]$, $\alpha(\delta_0) = \frac{\pi}{4}$,  $\sin(\alpha(\delta_0)) = \cos(\alpha(\delta_0))= \frac{1}{2}\sqrt{2}$, \EEE and $\alpha(\delta)-\alpha(\delta_0) < 0$ \EEE we derive
\begin{align}\label{deltageqdelta0}
\begin{split}
V_{\mathrm{r}}(2\sin(\alpha(\delta))) &\geq V_{\mathrm{r}}(\sqrt{2}) + 2V_{\mathrm{r},-}^\prime (\sqrt{2})\cos(\alpha(\delta_0))\big(\alpha(\delta)-\alpha(\delta_0)\big) \\&=V_{\mathrm{r}}(\sqrt{2}) +  \sqrt{2} \, V_{\mathrm{r},-}^\prime(\sqrt{2}) \Big(\frac{\pi(d-2)}{2\delta d}-\frac{\pi}{4}\Big) > -\frac{16}{\pi}\Big(\frac{\pi(d-2)}{2\delta d}-\frac{\pi}{4}\Big) \\&=\frac{4}{\delta d} (\delta d-2d+4).
\end{split}
\end{align}
From the previous calculation and (\ref{Ebndestimate}),  estimate (\ref{BoundaryEnergyEstimate})  \EEE follows  also for  $\delta_0 < \delta \le 3$. \EEE

 We now show that we have strict inequality in (\ref{BoundaryEnergyEstimate})  if one of the conditions \eqref{eq: 1}-\eqref{eq: 3} is violated. First, if  a boundary bond is not of unit length, we have strict inequality in \eqref{boundarylowerbound-new} and thus also in \eqref{Ebndestimate}. If \eqref{eq: 2} is violated, we find $\delta \neq \delta_0$ after a short computation. Then we get strict inequalities from \eqref{deltageqdelta1} or \eqref{deltageqdelta0}, respectively. Finally, let use assume that \eqref{eq: 3} is violated. We can suppose that $\delta=\delta_0$ and \eqref{eq: 1}-\eqref{eq: 2} hold as otherwise the inequality in  (\ref{BoundaryEnergyEstimate}) is strict. If  \EEE equality holds in (\ref{BoundaryEnergyEstimate}), then equality also holds in \eqref{boundarylowerbound-new}.  As $V_{\rm r}(2\sin(\alpha(\delta_0)))= 0$, \EEE this implies  
\begin{align*}
\sum_{ x_i \in I_2} V_{\mathrm{r}}\Big(2\sin\Big(\frac{\theta_i}{2}\Big)\Big) + \sum_{x_i \in I_3}\sum_{j=1}^2 \Big( V_{\mathrm{r}}\Big(2\sin\Big(\frac{\theta_i^j}{2}\Big)\Big) + \sum_{x_i \in I_4} \sum_{j=1}^3 V_{\mathrm{r}}\Big(2\sin\Big(\frac{\theta_i^j}{2}\Big)\Big) \Big) =0.
\end{align*} 
In view of ${\rm [vii]}$, this gives 
\begin{align}\label{eq: anglerange}
\begin{split}
&\theta_i \in [\tfrac{\pi}{2},\tfrac{3\pi}{2}]  \ \text{ for $x_i \in I_2$},\ \ \ \    \theta^1_i, \theta^2_i \in  [\tfrac{\pi}{2},\tfrac{3\pi}{2}] \ \text{ for $x_i \in I_3$},\ \ \ \  \theta^1_i, \theta^2_i,\theta_i^3 \in  [\tfrac{\pi}{2},\tfrac{3\pi}{2}]  \ \text{ for $x_i \in I_4$}.
\end{split}
\end{align}
  Under the assumption that \eqref{eq: 3} is violated,  in view of   \eqref{eq: anglerange},  $\theta_i = \theta_i^1 + \theta_i^2$ for  $x_i \in I_3$ and $\theta_i = \theta_i^1 + \theta_i^2+\theta_i^3$ for  $x_i \in I_4$, we obtain
\begin{align*}
  \pi(d-2) &= \sum_{i=1}^d \theta_i = \sum_{x_i \in I_2}\theta_i + \sum_{x_i \in I_3}( \theta_i^1+\theta_i^2) + \sum_{x_i \in I_4} (\theta_i^1+\theta_i^2+\theta_i^3) \\&> \frac{\pi}{2}\#I_2 + \pi \#I_3 + \frac{3\pi}{2}\#I_4  = \pi(d-2),
\end{align*}
 where the first equality follows from the fact that the maximal polygon has $d$ vertices, and the last equality holds by \eqref{eq: 2}. \EEE This is a contradiction and shows  that strict inequality in \eqref{BoundaryEnergyEstimate} holds \EEE if \eqref{eq: 3} is violated.

\noindent \emph{Step 3: Proof of \eqref{boundarylowerbound-new}.} \EEE
To complete the proof, it remains to show \eqref{boundarylowerbound-new}  and its strict version. \EEE We follow the lines of the proof of \cite[Lemma 5.1]{FriedrichKreutzHexagonal}. \EEE In the case of  a \EEE $2$-bonded $x_i$,  define $r^1_i = |x_i-x_{i-1}|$ and  $r^2_i = |x_i-x_{i+1}|$. In the case of  a \EEE $3$-bonded $x_i$,  define additionally $r^3_i = |x_{i}^b-x_{i}|$. In the case of a \EEE $4$-bonded $x_i$,  define   $r^3_i = |x_i^{b,1}-x_{i}|$ and $r^4_i = |x_{i}^{b,2}-x_{i}|$.
By the cosine rule we obtain  
  \begin{align}\label{xidistance}
\nonumber&|x_{i+1}-x_{i-1}|=  \ell(\theta_i,r_i^1,r_i^2),\text{ for   } x_i \in \partial X_n,  \\&
  |x_i^b-x_{i-1}| = \ell(\theta_i^1,r_i^1,r_i^3), \    |x_i^b-x_{i+1}| =  \ell(\theta_i^2,r_i^2,r_i^3) \text{ for   } x_i \in I_3, \\& \nonumber |x_i^{b,1}-x_{i-1}| = \ell(\theta_i^1,r_i^1,r_i^3), \    |x_i^{b,1}-x_{i}^{b,2}| =  \ell(\theta_i^2,r_i^3,r_i^4), \   |x_i^{b,2}-x_{i+1}| =  \ell(\theta_i^3,r_i^2,r_i^4) \EEE \text{ for   } x_i \in I_4,
\end{align}
where we have used the shorthand  
\begin{align}\label{eq: ell} 
 \ell(\theta,r_1,r_2) = \sqrt{r_1^2 + r_2^2 - 2r_1r_2 \cos(\theta)}.
 \end{align}
\EEE We want to prove that for every boundary atom $x_i$ its contribution to the energy can be controlled by the energy contribution in a modified configuration which has \EEE  the same angles but   unit bond lengths instead of $r_i^1, r_i^2$. Recall that \EEE  by $[\mathrm{viii}]$ we have for all $r \in (1,r_0]$ \EEE
\begin{align}\label{eq: strict inequality}
\frac{2}{r-1}(V_{\mathrm{a}}(1)-V_{\mathrm{a}}(r))  <  V'_{\mathrm{r},+}(1). \EEE
\end{align}
Let $\theta \in  [2\pi/5,8\pi/5]$, \EEE $r_1,r_2 \in [1,r_0]$. Then      $\ell(\theta,r_1,r_2)\geq 1$.  Note that \EEE  $V_{\mathrm{r},+}'(r)\leq 0$ for $r \geq 1$ and $V_{\mathrm{r},+}'$ is monotone increasing \EEE due to the convexity assumption in  $(\mathrm{v})$ \EEE on $V_{\mathrm{r}}$.  Moreover, we have  $\partial_{r_1} \ell(\theta,r_1,r_2)\leq 1$ by an elementary computation. We therefore obtain
\begin{align*}
\frac{2}{r-1}(V_{\mathrm{a}}(1)-V_{\mathrm{a}}(r))  \le  V_{\mathrm{r},+}'(1) \EEE \leq V_{\mathrm{r},+}'(\ell(\theta,s,r_2)) \leq V_{\mathrm{r},+}'(\ell(\theta,s,r_2)) \partial_{ r_1  }\ell(\theta,s,r_2).
\end{align*}
 By integration from $1$ to $r$ in the variable $s$   we get \EEE
\begin{align*}
\frac{1}{2}\left( V_{\mathrm{a}}(1)-V_{\mathrm{a}}(r) \right)\leq \frac{1}{4}\int_{1}^r V_{\mathrm{r},+}'(\ell(\theta,s,r_2)) \partial_{r_1} \EEE \ell(\theta,s,r_2) \mathrm{d}s = \frac{1}{4} V_{\mathrm{r}}(\ell(\theta,r,r_2))- \frac{1}{4} V_{\mathrm{r}}(\ell(\theta,1,r_2)).
\end{align*}
By applying this estimate in the second \EEE as well as in the third \EEE component of $\ell(\theta,r_1,r_2)$ with  $1$ \EEE and $r_2$ respectively, we derive  \EEE
\begin{align}\label{Estimateunitlength}
V_{\mathrm{a}}(1) + \frac{1}{4} V_{\mathrm{r}}(\ell(\theta,1,1)) \leq \frac{1}{2}\left( V_{\mathrm{a}}(r_1) + V_{\mathrm{a}}(r_2)\right)  +\frac{1}{4} V_{\mathrm{r}}(\ell(\theta,r_1,r_2)).
\end{align}
  Due to (\ref{ineq: bondangles}),  we have that all angles in \eqref{xidistance} lie in  $[2\pi/5,8\pi/5]$ \EEE and all bond lengths lie in $[1,r_0]$.    Now for all \EEE $2$-bonded $x_i$,  using (\ref{Estimateunitlength})  with $r_i^1$, $r_i^2$, and $\theta_i$, \EEE we have
\begin{align}\label{Twobondedestimate}
\frac{1}{2}\big( V_{\mathrm{a}}(r_i^1) +  V_{\mathrm{a}}(r_i^2)\big) +  \frac{1}{4} \EEE V_{\mathrm{r}}(\ell(\theta_i,r_i^1,r_i^2)) \geq  V_{\mathrm{a}}(1) + \frac{1}{4} V_{\mathrm{r}}(\ell( \theta_i, \EEE 1,1)).
\end{align}
For  all \EEE $3$-bonded $x_i$,  using (\ref{Estimateunitlength}) twice, we  get \EEE
\begin{align}\label{Threebondedestimate}
\begin{split}
&\frac{1}{2}\left( V_{\mathrm{a}}(r_i^1) +V_{\mathrm{a}}(r_i^3)\right) + \frac{1}{4} V_{\mathrm{r}}(\ell(\theta_i^1,r_i^1,r_i^3)) \geq V_{\mathrm{a}}(1) + \frac{1}{4} V_{\mathrm{r}}(\ell(\theta_i^1,1,1)), \\&
\frac{1}{2} \left(V_{\mathrm{a}}(r_i^2) +  V_{\mathrm{a}}(r_i^3)\right) +  \frac{1}{4} V_{\mathrm{r}}(\ell(\theta_i^2,r_i^2,r_i^3)) \geq V_{\mathrm{a}}(1)  + \frac{1}{4} V_{\mathrm{r}}(\ell(\theta_i^2,1,1)).
\end{split}
\end{align}
Finally, for  all \EEE $4$-bonded $x_i$,  using (\ref{Estimateunitlength}) three times, we have
\begin{align}\label{Fourbondedestimate}
\begin{split}
&\frac{1}{2}\left( V_{\mathrm{a}}(r_i^1) +V_{\mathrm{a}}(r_i^3)\right) + \frac{1}{4} V_{\mathrm{r}}(\ell(\theta_i^1,r_i^1,r_i^3)) \geq V_{\mathrm{a}}(1) + \frac{1}{4} V_{\mathrm{r}}(\ell(\theta_i^1,1,1)), \\&
\frac{1}{2} \left(V_{\mathrm{a}}(r_i^3) +  V_{\mathrm{a}}(r_i^4)\right) +  \frac{1}{4} V_{\mathrm{r}}(\ell(\theta_i^2,r_i^3,r_i^4)) \geq V_{\mathrm{a}}(1)  + \frac{1}{4} V_{\mathrm{r}}(\ell(\theta_i^2,1,1)), \\&
\frac{1}{2} \left(V_{\mathrm{a}}(r_i^2) +  V_{\mathrm{a}}(r_i^4)\right) +  \frac{1}{4} V_{\mathrm{r}}(\ell(\theta_i^3,r_i^2,r_i^4)) \geq V_{\mathrm{a}}(1)  + \frac{1}{4} V_{\mathrm{r}}(\ell(\theta_i^3,1,1)).
\end{split}
\end{align}
 By applying \eqref{eq: boundary energy}, \EEE  (\ref{xidistance}), (\ref{Twobondedestimate})-(\ref{Fourbondedestimate}),  $V_{\rm a}(1) = -1$, \EEE and $V_{\rm r} \ge 0$   we obtain   
 \begin{align}\label{boundaryestimate1}
\mathcal{R}^{\mathrm{bnd}}(C_n) &  \geq \EEE  \sum_{x_i \in I_2}  \Big(\frac{1}{2}\big(V_{\mathrm{a}}(r_i^1) + V_{\mathrm{a}}(r_i^2)\big) +  \frac{1}{4} \EEE V_{\mathrm{r}}(\ell(\theta_i,r_i^1,r_i^2)) \Big)  \EEE \notag \\&\quad+\sum_{x_i \in I_3} \Big(\frac{1}{2}\big(V_{\mathrm{a}}(r_i^1)+V_{\mathrm{a}}(r_i^2) + 2 V_{\mathrm{a}}(r_i^3)\big) +\frac{1}{4}\big( V_{\mathrm{r}}(\ell(\theta_i^1,r_i^1,r_i^3))+  V_{\mathrm{r}}(\ell(\theta_i^2,r_i^2,r_i^3))\big) \Big) \notag\\&\quad+\sum_{x_i \in I_4} \Big(\frac{1}{2}\big(V_{\mathrm{a}}(r_i^1)+V_{\mathrm{a}}(r_i^2) + 2 V_{\mathrm{a}}(r_i^3)+ 2 V_{\mathrm{a}}(r_i^4)\big)\notag \\&\qquad \ \ \ \ \ \ \ \ + \frac{1}{4}\left( V_{\mathrm{r}}(\ell(\theta_i^1,r_i^1,r_i^3))+  V_{\mathrm{r}}(\ell(\theta_i^2,r_i^3,r_i^4)) + V_{\mathrm{r}}(\ell(\theta_i^3,r_i^2,r_i^4))   \right) \Big) \notag\\&\geq -(\#I_2+2\#I_3+3\#I_4)  \notag \\&  \quad +  \frac{1}{4} \EEE \Big( \sum_{  x_i \EEE \in I_2} V_{\mathrm{r}}(\ell(\theta_i,1,1))+ \sum_{  x_i \EEE \in I_3}  \sum_{j=1}^2  V_{\mathrm{r}}(\ell(\theta_i^j,1,1))+ \sum_{  x_i \EEE \in I_4}  \sum_{j=1}^3  V_{\mathrm{r}}(\ell(\theta_i^j,1,1))\Big). 
\end{align}
 For later purposes, we remark that this inequality is strict if one bond does not have \EEE unit length. This follows from the strict inequality in \eqref{eq: strict inequality}. \EEE

 Recall   $\delta= (\#I_2+2\#I_3+3\#I_4)/d$ and note that $\ell(\theta,1,1) = 2\sin(\theta/2)$ by  \eqref{eq: ell}.  By \EEE using    $\theta_i = \theta_i^1 + \theta_i^2$ for $x_i \in I_3$, $\theta_i = \theta_i^1 + \theta_i^2+\theta_i^3$ for $x_i \in I_4$,   and \eqref{boundaryestimate1} \EEE we obtain
\begin{align*} 
\mathcal{R}^{\mathrm{bnd}}(C_n) &\geq -\delta d +  \frac{1}{4} \EEE \Big(\sum_{x_i \in I_2} V_{\mathrm{r}}(\ell(\theta_i,1,1)) +   \sum_{x_i \in I_3}  \sum_{j=1,2}V_{\mathrm{r}}(\ell(\theta_i^j,1,1))  +    \sum_{x_i \in I_4}  \sum_{j=1}^3V_{\mathrm{r}}(\ell(\theta_i^j,1,1)) \Big)  \EEE \notag \\& = -\delta d +  \frac{1}{4} \EEE \Big(\sum_{x_i \in I_2} V_{\mathrm{r}}\Big(2\sin\Big(\frac{\theta_i}{2}\Big)\Big) + \hspace{-0.1cm}\sum_{x_i \in I_3}  \sum_{j=1}^2 V_{\mathrm{r}}\Big(2\sin\Big(\frac{\theta_i^j}{2}\Big)\Big)+ \hspace{-0.1cm} \sum_{x_i \in I_4}  \sum_{j=1}^3 V_{\mathrm{r}}\Big(2\sin\Big(\frac{\theta_i^j}{2}\Big)\Big)\Big). \EEE  
\end{align*}
This yields the first inequality in \eqref{boundarylowerbound-new}. We note that this inequality is strict  if one of the bonds does not have  unit length since then \eqref{boundaryestimate1} is strict. The second inequality in \eqref{boundarylowerbound-new} follows by a convexity argument:  since $V_{\mathrm{r}}$ is convex and non-increasing  by ${\rm [v]}$, \EEE and $\sin(\theta/2)$ is concave for $\theta \in [0,2\pi]$, we  observe that $\theta \mapsto V_{\mathrm{r}}(2\sin(\frac{\theta}{2}))$ is a convex function for $\theta \in [0,2\pi]$, see \cite[Proof of Lemma 5.1]{FriedrichKreutzHexagonal} for details. This along  with \EEE the fact that  $\# I_2 + 2 \# I_3 + 3\#I_4 = \delta d$   and 
$$
\pi(d-2)=\sum_{i=1}^d \theta_i = \sum_{x_i \in I_2} \theta_i + \sum_{x_i \in I_3}(\theta_i^1+\theta_i^2)+ \sum_{x_i \in I_4}(\theta_i^1+\theta_i^2+\theta_i^4)
$$
 yields the second inequality in  \eqref{boundarylowerbound-new}. This \EEE concludes the proof. \EEE
\end{proof}

We now proceed with the proof of Lemma \ref{lemma : charge interior}. 
Let $C_n$ be a connected $q_\mathrm{sat}^n$-optimal configuration without bridging atoms.  Recall the decomposition $X_n= X_n^{a,\mathrm{bulk}} \cup \partial X_n^a \cup I_\mathrm{ac}^{\rm ext}$ introduced in Subsection \ref{sec: boundary charge}. For $2 \le k\le 4$ we define
\begin{align*}
I_k^\pm = \{x_i \in \partial X_n^a : q_i= \pm 1, \ \  \#(\mathcal{N}(x_i) \cap X_n^a)=k\}.
\end{align*}
Moreover,  we set $I^\pm = \bigcup_{k=2}^{4} I_k^\pm$. Note that $I_k = I_k^+ \cup I_k^-$, where $I_k$ is given by \eqref{def: Ik}. \EEE By  Lemma \ref{LemmaNeighborhood}, \EEE $C_n$ has alternating charge distribution. Thus, $\partial C_n^a$ is a cycle of  even length $d \in 2\mathbb{N}$ and,   by \eqref{NeighbourhoodCard} \EEE there holds $\# I^- = \# I^+ = d/2$. We denote the interior angle of the polygon $\partial X_n^a$ at $x_i$ by $\theta_i$. We start with some preliminary properties.

\begin{lemma}\label{eq: preliminary estimates}
Let $n\geq 6$ and let $C_n$ be a connected $q_\mathrm{sat}^n$-optimal configuration. Then
\begin{itemize}
\item[(i)] $\# I_4^+=0$.
\item[(ii)] Let $m \geq 4$ be such that $\#I_2 + 2\#I_3 +3\#I_4 \leq 2d-m$. Then
 $2\#I_2^- + \#I_3^-  \geq m + \#I_3^+$.
\end{itemize}
\end{lemma}

\begin{proof}
\noindent\emph{Proof of (i).} Assume by contradiction that there exists $x_i \in I_4^+$. Since $C_n$ is a $q_\mathrm{sat}^n$-optimal configuration, all atoms of charge $-1$ are $4$-bonded with bond angles $\frac{\pi}{2}$, see Lemma \ref{lemma : Charge Energybound}(a)(i),(iv). Hence, $x_i$ is contained in four squares and thus not part of $\partial X_n^a$. 

\noindent\emph{Proof of (ii).}   As each atom is at most $4$-bonded, see \eqref{NeighbourhoodCard}, we have $\#I_2+ \#I_3+\#I_4 = d$. This along with   (i) and the assumption of (ii) shows
\begin{align*}
\#I_4^-+m = \#I_4+m  \leq \#I_2 = \#I_2^+ + \#I_2^-.
\end{align*}
This implies
\begin{align*}
2\#I_2^- +\#I_3^- &=  (\#I_2^- +\#I_3^- - \#I_2^+) + (\#I_2^++\#I_2^-)  \geq \#I_2^-+\#I_3^- - \#I_2^+ +\#I_4^- +m  \notag \\&  = \#I^-  +m - \#I_2^+   = \# I^- +m - \# I^+  + \#I_3^+ + \#I_4^+.
\end{align*}
Since there holds $\# I^- = \#I^+$ due to alternating charge distribution and $\# I_4^+ = 0$ by property  (i), property (ii) follows. 
\end{proof}

\begin{proof}[Proof of Lemma \ref{lemma : charge interior}(a)]

By Theorem \ref{theorem: min-en2}, Lemma \ref{lemma : Charge Energybound}(a)(i),(ii), and Remark \ref{rem: repulsionsfree} there holds $\mathcal{E}(C_n)=-b$. Therefore,  $\theta_i \geq (k-1)\frac{\pi}{2} \text{ if } x_i  \in I_k$ by Lemma \ref{RemarkPolygon}(b).  This implies 
\begin{align*}
\#I_2 + 2\#I_3 +3\#I_4 \leq \frac{2}{\pi} \sum_{i=1}^d \theta_i = 2d-4,
\end{align*}
where in the last step we used that the sum of the interior angles equals $\pi(d-2)$. This shows \eqref{eq :I2I3I3}. 

 We now show the additional statement. To this end, suppose that $\#\mathcal{A}(X_n) \geq \#\mathcal{A}(X_n^{\mathrm{a},\mathrm{bulk}})+2$. We claim that \EEE 
\begin{align}\label{ineq :I2card}
\#\{x_i \in I_2 : \theta_i \geq 5\pi/6\} \geq 2.
\end{align}
 We defer the proof of \eqref{ineq :I2card} and first conclude the proof of the statement. \EEE The fact that the sum of the interior angles equals $\pi(d-2)$, \eqref{ineq :I2card},  and $\theta_i \geq (k-1)\frac{\pi}{2}$ for $x_i  \in I_k$  \EEE imply \EEE
\begin{align*}
1 + \#I_2 +2\#I_3 + 3\#I_4<\frac{2}{\pi}\frac{5}{6}\pi\cdot 2 + (\#I_2-2) + 2\#I_3 + 3\#I_4 \leq \frac{2}{\pi}\sum_{i=1}^d \theta_i = 2d-4.
\end{align*}
Since the left and the right hand side are integers, we find $2 + \#I_2 +2\#I_3 + 3\#I_4 \le 2d-4$.  This shows that   $m \ge 6$ in \eqref{eq :I2I3I3}. \EEE It remains to prove \eqref{ineq :I2card}.  As a preparation, we \EEE first show that 
 \begin{align}\label{eq: grosser winkel}
 x \in \mathcal{A}(X_n)  \ \ \ \ \Rightarrow \ \ \ \  \text{ $x\in X^+_n$, $x$ is $2$-bonded and bond-angles at $x$ lie in $[5\pi/6, 7\pi/6 ]$.} \EEE
 \end{align}
By \EEE Lemma \ref{lemma : Charge Energybound}(a)(i),(iv),(v) we observe that  $x$ is $2$-bonded and has charge $+1$.  Denote the two neighbors of $ x \EEE \in   \mathcal{A}(X_n)$ by  $z_{1}$ and $z_{2}$. Since $z_1$ and $z_2$ have charge $-1$, they are $4$-bonded. Also note that  \EEE $(\mathcal{N}(z_{1}) \cap \mathcal{N}(z_{2}))\setminus \{x\} = \emptyset$. In fact, if there was an atom $  y \EEE \in (\mathcal{N}(z_{1}) \cap \mathcal{N}(z_{2}))\setminus \{x\}$, the four atoms $z_{1}$, $x$, $z_{2}$, and $y$ would form a square. Due to Lemma  \ref{RemarkPolygon}(b), this square is regular which contradicts the fact that $x \in  \mathcal{A}(X_n)$. Now, if we  had  that a bond-angle at $x$ does not lie in $[5\pi/6, 7\pi/6 ]$, \EEE it is elementary to see that $\mathrm{dist}(\mathcal{N}(z_{1}) \setminus \{x\},\mathcal{N}(z_{2})\setminus \{x\}) < \sqrt{2}$. This would contradict the fact that $C_n$ is repulsion-free.

 We now show that  \eqref{ineq :I2card} holds. Recall the assumption $\#\mathcal{A}(X_n) \geq \#\mathcal{A}(X_n^{\mathrm{a},\mathrm{bulk}})+2$.  \EEE We consider  the two \EEE cases (a)  $\mathcal{A}(X_n) \cap \partial X_n^\mathrm{a} \neq \emptyset$  and \EEE  (b) $\mathcal{A}(X_n) \cap \partial X_n^\mathrm{a} = \emptyset$. 

\noindent \emph{Proof of Case $\mathrm{(a)}$.} Observe that $\#( \partial X_n^a \cap \mathcal{A}(X_n)) \ge 2$ by Lemma \ref{RemarkPolygon}(a). Additionally, by \eqref{eq: grosser winkel} we get  $\partial X_n^a \cap \mathcal{A}(X_n) \subset I_2^+$  and  $\theta_i \geq 5\pi/6$ for $x_i \in \partial X_n^a \cap \mathcal{A}(X_n)$. \EEE Hence, \eqref{ineq :I2card}  holds. \EEE 

\noindent \emph{Proof of Case $\mathrm{(b)}$.}   We first \EEE show that for each $x \in \mathcal{A}(X_n)\setminus \mathcal{A}(X_n^{\mathrm{a},\mathrm{bulk}})$ we can find 
\begin{align}\label{eq: z-def}
z_x \in X_n \ \ \text{ such that $z_x \in I_2^+$, \  $|x-z_x| =\sqrt{2}$, \  and the two bond-angles at $z_x$ equal $\pi$}.
\end{align}
By \eqref{eq: grosser winkel} we get that \EEE $x\in X_n^+$,  that \EEE it is $2$-bonded in $X_n$ and, since $x \in \mathcal{A}(X_n)\setminus \mathcal{A}(X_n^{\mathrm{a},\mathrm{bulk}})$, it is $1$-bonded in $X_n^{\mathrm{a},\mathrm{bulk}} $.  Thus, there exists exactly one atom in $\mathcal{N}(x) \cap \partial X_n^\mathrm{a}$, denoted by $x_i$. By $x_{i-1},x_{i+1}$ we denote the two neighbors of $x_i$ \EEE in $\partial X_n^\mathrm{a}$.  At least one of these two atoms, say without restriction $x_{i-1}$, satisfies $|x_{i-1} - x| = \sqrt{2}$, see   Fig.~\ref{Fig: A.17}. \EEE We will identify $x_{i-1}$ as the atom $z_x$ in \eqref{eq: z-def}. \EEE  Recall the assumption $\partial X_n^a \cap \mathcal{A}(X_n)=\emptyset$, which implies that the bond-angles at $x_{i-1}$ lie in $\{\pi/2,\pi,3\pi/2\}$.  Let $y \in \mathcal{N}(x_{i-1}) \setminus \lbrace x_i\rbrace$. Suppose $y$, $x_i$, and $x_{i-1}$ would form an angle in $\lbrace \pi/2,3\pi/2\rbrace$, i.e., $|y-x_i| = \sqrt{2}$.    Then either (i) $x_{i-1},y,x,x_i$ or (ii) $x_{i-1},y,z,x_i$ form a square, where $z \in \mathcal{N}(x_i)\setminus \lbrace x,x_{i-1}\rbrace$, see  Fig.~\ref{Fig: A.17}. In view of \eqref{eq: grosser winkel}, case (i)  \EEE contradicts $x \in \mathcal{A}(X_n)$, and case (ii) contradicts  $x_{i-1} \in \partial X_n^\mathrm{a}$. This shows that $y,x_{i-1},x_i$ form an angle $\pi$, \EEE  and this also yields that $x_{i-1} \in I^+_2$.   We denote $x_{i-1}$ by $z_x$. \EEE

 Since $\#\mathcal{A}(X_n) \geq \#\mathcal{A}(X_n^{\mathrm{a},\mathrm{bulk}})+2$ by assumption, we find $x\neq  y \in \mathcal{A}(X_n) \setminus \mathcal{A}(X_n^{\mathrm{a},\mathrm{bulk}})$.   We now show that  $z_x \neq z_y$\EEE,  where $z_x$ and $z_y$ \EEE are the atoms identified in \eqref{eq: z-def}. \EEE This will imply \eqref{ineq :I2card}. 
Assume by contradiction that $z_x=z_y$. Denote  $z_x = z_y$  by $x_{i} \in I_2^+$ and recall $\theta_i = \pi$. \EEE Denote its two neighbors by $x_{i-1},x_{i+1}$,  and recall that they are $4$-bonded, in \EEE particular bonded to $x$  or $y$, respectively.    Consider the unique atoms $w_x \in \mathcal{N}(x) \setminus \partial X_n^\mathrm{a}$ and  $w_y \in \mathcal{N}(y) \setminus \partial X_n^\mathrm{a}$. Observe that \EEE they are $4$-bonded. First, note that  $\mathcal{N}(  w_x \EEE ) \cap \mathcal{N}( w_y \EEE )=\emptyset$ since the contrary would imply that $x,y$ would be contained in an octagon with bond angles $\pi,\pi/2$ respectively, contradicting the fact that $x,y \in \mathcal{A}(X_n)$. Since by  \eqref{eq: grosser winkel} the bond angles at $x,y \in \mathcal{A}(X_n)$  lie in $[5\pi/6,7\pi/6]$, \EEE an elementary argument shows that $\mathrm{dist}(\mathcal{N}(  w_x \EEE ), \mathcal{N}( w_y\EEE ))<\sqrt{2}$,  see Fig.~\ref{Fig: A.17}. \EEE This contradicts  the fact that $C_n$ is repulsion-free  and concludes the proof of case (b).
\end{proof}

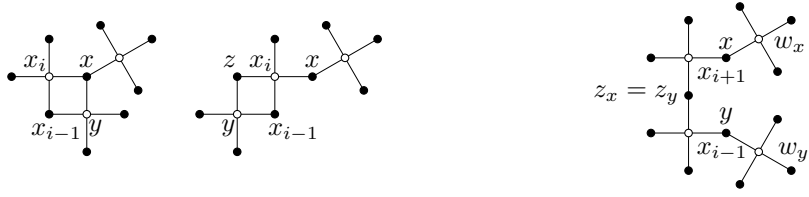
\begin{figure}[htp]
\centering
\begin{tikzpicture}[scale=.5]
\begin{scope}[rotate=90,shift={(-0,-3)}]
\draw[ultra thin](-2,-1)++(-120:1)--++(-120:1);
\draw[ultra thin](-2,-1)++(-120:1)--++(-30:-1);
\draw[ultra thin](-2,-1)--++(-120:1)--++(-30:1);
\draw[ultra thin](0,-1)++(-60:1)--++(30:-1);
\draw[ultra thin](0,-1)++(-60:1)--++(-60:1);
\draw[ultra thin](0,-1)--++(-60:1)--++(30:1);
\draw[ultra thin](0,0)--++(0,1);
\draw[ultra thin](-2,0)--++(0,1);
\draw[ultra thin](-3,0)--(1,0);
\draw[ultra thin](-2,0)--(-2,-1);
\draw[ultra thin](0,0)--(0,-1);
\draw[fill=white](0,0) circle(.1);
\draw[fill=white](-2,0) circle(.1);
\draw[fill=black](-2,1) circle(.1);
\draw[fill=black](0,1) circle(.1);
\draw[fill=black](-3,0) circle(.1);
\draw[fill=black](-1,0) circle(.1);
\draw[fill=black](1,0) circle(.1);
\draw[fill=black](0,-1) circle(.1);
\draw[fill=black](-2,-1) circle(.1);
\draw[fill=white](-2,-1)++(-120:1)circle(.1);
\draw[fill=black](-2,-1)++(-120:1)++(-30:-1)circle(.1);
\draw[fill=black](-2,-1)++(-120:1)++(-120:1)circle(.1);
\draw[fill=black](-2,-1)++(-120:1)++(-30:1)circle(.1);
\draw[fill=white](0,-1)++(-60:1)circle(.1);
\draw[fill=black](0,-1)++(-60:1)++(30:-1)circle(.1);
\draw[fill=black](0,-1)++(-60:1)++(-60:1)circle(.1);
\draw[fill=black](0,-1)++(-60:1)++(30:1)circle(.1);

\draw(-2,-.9) node[anchor=north] { $x_{i-1}$ };
\draw(0,-.9) node[anchor=north] { $x_{i+1}$  };
\draw(0,-1) node[anchor=south] {$x$};
\draw(-0.1,-2.7) node[anchor=south] { $w_x$  };
\draw(-3,-2.8) node[anchor=south] { $w_y$  };
\draw(-2,-1) node[anchor=south] {$y$};
\draw(-1,0) node[anchor=east]{$z_x=z_y$};
\end{scope}

\begin{scope}[rotate=90,shift={(-.5,14)}]

\draw[ultra thin](0,-1)++(-60:1)--++(30:-1);
\draw[ultra thin](0,-1)++(-60:1)--++(-60:1);
\draw[ultra thin](0,-1)--++(-60:1)--++(30:1);
\draw[ultra thin](0,0)--++(0,1);
\draw[ultra thin](-1,0)--(1,0);
\draw[ultra thin](0,0)--(0,-1);
\draw[fill=white](0,0) circle(.1);
\draw[fill=black](0,1) circle(.1);
\draw[fill=black](-1,0) circle(.1);
\draw[fill=black](1,0) circle(.1);
\draw[fill=black](0,-1) circle(.1);

\draw[ultra thin](-1,0)--++(0,-1)--++(1,0);
\draw[ultra thin](-1,0)++(0,-1)--++(-1,0);
\draw[ultra thin](-1,0)++(0,-1)--++(0,-1);
\draw[fill=black](-1,0)++(0,-1)++(-1,0)circle(.1);
\draw[fill=black](-1,0)++(0,-1)++(0,-1)circle(.1);
\draw[fill=white](-1,0)++(0,-1)circle(.1);
\draw[fill=white](0,-1)++(-60:1)circle(.1);
\draw[fill=black](0,-1)++(-60:1)++(30:-1)circle(.1);
\draw[fill=black](0,-1)++(-60:1)++(-60:1)circle(.1);
\draw[fill=black](0,-1)++(-60:1)++(30:1)circle(.1);
\draw(-0.1,-0.2) node[anchor=south east] {$x_i$};
\draw(0,-1) node[anchor=south] {$x$};
\draw(-.9,-.8) node[anchor= north west] {$y$};
\draw(-1,-.2) node[anchor=north]{ $x_{i-1}$};
\end{scope}

\begin{scope}[rotate=90,shift={(-.5,8)}]

\draw[ultra thin](0,-1)++(-60:1)--++(30:-1);
\draw[ultra thin](0,-1)++(-60:1)--++(-60:1);
\draw[ultra thin](0,-1)--++(-60:1)--++(30:1);
\draw[ultra thin](0,0)--++(0,1);
\draw[ultra thin](-1,0)--(1,0);
\draw[ultra thin](0,0)--(0,-1);
\draw[fill=white](0,0) circle(.1);
\draw[fill=black](0,1) circle(.1);
\draw[fill=black](-1,0) circle(.1);
\draw[fill=black](1,0) circle(.1);
\draw[fill=black](0,-1) circle(.1);

\draw[ultra thin](-1,0)--++(0,1)--++(1,0);
\draw[ultra thin](-1,0)++(0,1)--++(-1,0);
\draw[ultra thin](-1,0)++(0,1)--++(0,1);
\draw[fill=black](-1,0)++(0,1)++(-1,0)circle(.1);
\draw[fill=black](-1,0)++(0,1)++(0,1)circle(.1);
\draw[fill=white](-1,0)++(0,1)circle(.1);
\draw[fill=white](0,-1)++(-60:1)circle(.1);
\draw[fill=black](0,-1)++(-60:1)++(30:-1)circle(.1);
\draw[fill=black](0,-1)++(-60:1)++(-60:1)circle(.1);
\draw[fill=black](0,-1)++(-60:1)++(30:1)circle(.1);
\draw(-0.1,-0.2) node[anchor=south east] {$x_i$};
\draw(0,1.2) node[anchor=south] {$z$};
\draw(0,-1) node[anchor=south] {$x$};
\draw(-.9,.8) node[anchor= north east] {$y$};
\draw(-1,-.5) node[anchor=north]{ $x_{i-1}$};
\end{scope}

\end{tikzpicture}
\caption{ The relevant parts of the configurations considered in the proof of Case (b). Left: The two cases (i) and (ii) in the proof of \eqref{eq: z-def}. Right: The situation  $z_x=z_y$. \EEE }
\label{Fig: A.17}
\end{figure}

\vspace{-0.3cm}

%
%

\begin{proof}[Proof of Lemma \ref{lemma : charge interior}(b)]  
We proceed in several steps. We first introduce an auxiliary configuration $C_{n-d}'$ which arises from $C_n$ by removing the boundary and adding again the $3$-bonded boundary atoms of  charge $+1$ (Step 1). We collect some properties about the net charge of $C_{n-d}'$ whose proofs are deferred to Steps 3 and 4. In Step 2 we define the configuration $C_{n-d-m}$ by suitably adding atoms of charge $+1$ to $C_{n-d}'$, and we prove the statement of the lemma.

 \noindent \emph{Step 1: An auxiliary configuration.} We recall the decomposition $X_n= X_n^{a,\mathrm{bulk}} \cup \partial X_n^a \cup I_\mathrm{ac}^{\rm ext}$ introduced in Subsection \ref{sec: boundary charge}. By  Lemma \ref{RemarkPolygon}(a) and Lemma \ref{LemmaNeighborhood}, $\partial C_n^a$ is a cycle of even length $d$ with alternating charge distribution, and therefore $\mathcal{Q}(\partial C_n^a) = 0$. Moreover, all exterior acyclic bonds have charge $+1$  by Lemma \ref{lemma: bridging}(b). \EEE This implies  
\begin{align}\label{eq : chargeCna}
\mathcal{Q}(C_n) =  \mathcal{Q}(C_n^{a,\mathrm{bulk}})+ \mathcal{Q}(\partial C_n^a)+ \mathcal{Q}(I_\mathrm{ac}^{\rm ext}) = \mathcal{Q}(C_n^{a,\mathrm{bulk}})+  \# I_\mathrm{ac}^{\rm ext}.
\end{align}
 We introduce an auxiliary configuration $C_{n-d}'$ by 
\begin{align}\label{eq: strich definition}
C_{n-d}'= C_n^{a,\mathrm{bulk}} \cup \bigcup\nolimits_{x_i \in I_3^+} (x_i,1).
\end{align}
In a similar fashion, we denote the atomic positions by $X_{n-d}'$. It is clear that $\mathcal{Q}(C_{n-d}')=\mathcal{Q}(C_n^{a,\mathrm{bulk}})+\#I_3^+$. By construction there holds $\mathcal{A}(X_n^{\mathrm{a},\mathrm{bulk}} )=\mathcal{A}(X_{n-d}')$ since all the atoms added from $I_3^+$ are $1$-bonded in $C'_{n-d}$,  more precisely, bonded to an  atom which does not lie in $\mathcal{A}(X_n)$. \EEE  By (\ref{eq : chargeCna})   we get 
\begin{align}\label{eq : charge main}
\mathcal{Q}(C_n)= \mathcal{Q}(C_{n-d}') + \# I_{\rm ac}^{\rm ext} -\#I_3^+.
\end{align}
The goal is   to estimate the contributions of the different terms in (\ref{eq : charge main}) separately: we will show that 
\begin{align}
\mathcal{Q}(C_{n-d}') & \geq \phi\big(n-d-(\# I_\mathrm{ac}^{\rm ext}-\#I_3^+)\big), \label{eq: diff-prop1}\\
\# I_\mathrm{ac}^{\rm ext} - \# I_3^+  & \geq  m,     \label{eq: diff-prop2} 
\end{align}
 where $m \ge 4$ is given in \eqref{eq :I2I3I3}. \EEE We defer the proof of \eqref{eq: diff-prop1}--\eqref{eq: diff-prop2} to Steps 3 and 4 below and proceed to show the statement of the lemma.

\noindent \emph{Step 2: Proof of the statement}: First, \eqref{eq: diff-prop2} implies $\# I_\mathrm{ac}^{\rm ext}  \ge  \# I_3^+ + m \ge m$. Next, we construct the configuration $C_{n-d-m}$. A difficulty arises from the fact   that the auxiliary  configuration $C_{n-d}'$ \EEE possibly does not satisfy the net charge equality \eqref{ineq : charge-new} and therefore we  need to add atoms of charge $+1$.

 We set $ k_m:= \# I_\mathrm{ac}^{\rm ext} -\#I_3^+-m$ and observe that     $k_m \geq 0$ by  \eqref{eq: diff-prop2}. We define  
\begin{align*}
C_{n-d-m} := C_{n-d}' \cup   \big\{(z_1,1),\ldots, (z_{k_m},1)\big\}, 
\end{align*}
where the positions are chosen such that $\min\{\mathrm{dist}(z_i,X_n),\mathrm{dist}(z_i,z_j): i\neq j \} \geq \sqrt{2}$.
In view of \eqref{eq: strich definition}, we observe that $C_{n-d-m}$ indeed consists of  
\begin{align*}
n-d-\# I_\mathrm{ac}^{\rm ext} +\#I_3^+ +k_m = n-d-\# I_\mathrm{ac}^{\rm ext}+\#I_3^++\# I_\mathrm{ac}^{\rm ext} -\#I_3^+-m = n-d-m
\end{align*}
atoms. Since the added atoms are all $0$-bonded,  in view of \eqref{eq: strich definition}, $X_{n-d-m}$ is a subset of $X_n$ up to $0$-bonded atoms. Moreover, \EEE we clearly have $ \mathcal{A}(X_{n-d-m}) = \mathcal{A}(X_{n-d}')$ and thus $\mathcal{A}(X_{n-d-m}) = \mathcal{A}(X_n^{\mathrm{a},\mathrm{bulk}}) $.  As we have added $k_m$ $0$-bonded atoms of charge $+1$ to $C_{n-d}'$, $C_{n-d-m}$ still satisfies Lemma \ref{lemma : Charge Energybound}(a)(i),(ii), and thus by Lemma \ref{lemma : Charge Energybound}(b)  there holds equality in  \eqref{ineq : charge}. Therefore, $C_{n-d-m}$ is optimal and satisfies $\mathcal{Q}(C_{n-d-m})\geq q_\mathrm{sat}^{n-d-m}$  by Lemma \ref{lemma : Charge Energybound}(a)(iii). \EEE Moreover, we clearly have  $\mathcal{Q}(C_{n-d-m}) = \mathcal{Q}(C_{n-d}') + k_m$.  By \eqref{eq : charge main} and the definition $k_m =  \# I_\mathrm{ac}^{\rm ext} -\#I_3^+-m$   we observe    
\begin{align*}
\mathcal{Q}(C_n)= \mathcal{Q}(C_{n-d}')+   \# I_\mathrm{ac}^{\rm ext}-\#I_3^+ &= \mathcal{Q}(C_{n-d-m})+ \# I_\mathrm{ac}^{\rm ext} -k_m-\#I_3^+= \mathcal{Q}(C_{n-d-m})+m.
\end{align*}
This shows \eqref{ineq : charge-new} and concludes the proof of  the statement. To conclude the proof, it thus remains to show \eqref{eq: diff-prop1} and \eqref{eq: diff-prop2}.

 \noindent \emph{Step 3: Proof of \eqref{eq: diff-prop1}.} We first recall that  $C_n^{a,\mathrm{bulk}}$ consists of $n-d-\# I_\mathrm{ac}^{\rm ext}$ atoms and thus $C_{n-d}'$ consists of $\# X_{n-d}' :=  n-d-(\# I_\mathrm{ac}^{\rm ext}-\#I^+_3)$ atoms. Thus, in view of  Lemma \ref{lemma : Charge Energybound}(a)(iii)  it suffices to prove 
\begin{align}\label{eq : chargeenergieCn-dprime}
\mathcal{E}(C_{n-d}') = -2 \# X_{n-d}' + 2\mathcal{Q}(C_{n-d}').
\end{align}
By Lemma \ref{lemma : Charge Energybound}(a)(i),(ii) we find that $C_n$ is repulsion-free with unit bond lengths and all atoms of charge $-1$ are $4$-bonded. Since $C_{n-d}' \subset C_n$,   also $C_{n-d}'$ is repulsion-free with unit bond lengths. We now show that each atom of $C_{n-d}'$ with charge $-1$ is $4$-bonded.   Let $x_i \in X_{n-d}'$ with $q_i=-1$. We  show that  $\mathcal{N}(x_i) \cap (X_n \setminus  X_{n-d}' \EEE ) = \emptyset$. In fact, in view of \eqref{eq: strich definition} and   $q_i =-1$, \EEE we have 
\begin{align}\label{eq: emptyempty} 
\mathcal{N}(x_i) \cap (X_n \setminus X_{n-d}') \subset (\partial X_n^a \setminus I_3^+) \cap X_n^+ = I_2^+\cup I_4^+.
\end{align}
We first recall that $\#I_4^+ = \emptyset$, see Lemma \ref{eq: preliminary estimates}(i). Moreover, atoms in $\mathcal{N}(x_i) \cap( X_n \setminus  X_{n-d}' \EEE )$ cannot lie in $I_2$ since atoms in $I_2$ are only bonded to other atoms in $\partial X_n^a$ or to atoms in $I^{\rm ext}_\mathrm{ac}$,  see \eqref{def: Ik}. \EEE Then \eqref{eq: emptyempty} shows $\mathcal{N}(x_i) \cap (X_n \setminus   X_{n-d}' \EEE ) = \emptyset$.

This implies that the neighborhood of $x_i$ in $X_{n-d}'$ coincides with the one in $X_n$, and thus $x_i$ is still $4$-bonded, cf.\ Lemma \ref{lemma : Charge Energybound}(a)(i). Consequently, we have shown that  $C_{n-d}'$ satisfies Lemma \ref{lemma : Charge Energybound}(a)(i),(ii). Then \eqref{eq : chargeenergieCn-dprime} follows from Lemma \ref{lemma : Charge Energybound}(b).

\noindent \emph{Step 4: Proof of \eqref{eq: diff-prop2}}.  
 Since $n\geq 6$ and  $C_n$ does not contain bridging atoms, we can apply Lemma \ref{lemma: bridging}(b). We get that $I_{\rm ac}^{\rm ext} = I_{\rm ac}^{\rm ext} \cap X_n^+$ and that each atom in $I_{\rm ac}^{\rm ext}$ is $1$-bonded.  For all $x_i \in I_3^-$ there exists exactly one $x_j \in \mathcal{N}(x_i) \cap I_\mathrm{ac}^{\rm ext}$, whereas for $x_i \in I_2^-$ there exist exactly two $x_j,x_k \in \mathcal{N}(x_i) \cap I_\mathrm{ac}^{\rm ext}$. As atoms in $I_{\rm ac}^{\rm ext}$ are $1$-bonded, it is clear that no atom in $I_{\rm ac}^{\rm ext}$ is bonded to two different atoms in  $I_2^- \cup I_3^-$. This yields 
\begin{align*}
\# I_{\rm ac}^{\rm ext} \ge 2 \# I_{2}^- + \# I_{3}^-.
\end{align*}
Then  \eqref{eq: diff-prop2} follows from \eqref{eq :I2I3I3} and \EEE Lemma \ref{eq: preliminary estimates}(ii). 
\end{proof}
\EEE

\section*{Acknowledgements} 
This work was funded by the Deutsche Forschungsgemeinschaft (DFG, German Research Foundation) under Germany's Excellence Strategy EXC 2044 -390685587, Mathematics M\"unster: Dynamics--Geometry--Structure.


\begin{thebibliography}{99}
%




\bibitem{Molecular}
 {\sc   N.L. Allinger}.
\newblock {\em Molecular structure: understanding steric and electronic effects from molecular mechanics}.
\newblock John Wiley \& Sons
\newblock (2010).

 

 


\bibitem{Yuen}
 {\sc   Y.~Au Yeung, G.~Friesecke, B.~Schmidt}.
\newblock {\em Minimizing atomic configurations of short range
pair potentials in two dimensions: crystallization in the Wulff-shape}.
\newblock Calc.\ Var.\ Partial Differential Equations
\newblock {\bf 44} (2012), 81--100.



\bibitem{Betermin0}
{\sc L.~B\'etermin, L.~De Luca, M.~Petrache}.
\newblock  {\em Crystallization to the square lattice for a two-body potential}.  
\newblock Preprint at \href{https://arxiv.org/abs/1907.06105}{\tt arXiv:1907.06105}  
\EEE 

 
\bibitem{Betermin}
{\sc L.~B\'etermin, H.~Kn\"upfer, F.~Nolte}.
\newblock  {\em Crystallization of  one-dimensional alternating two-component systems}.  
\newblock Preprint at \href{https://arxiv.org/abs/1804.05743}{\tt arXiv:1804.05743}  
 
 
 \bibitem{Betermin2}
{\sc L.~B\'etermin, H.~Kn\"upfer}.
\newblock  {\em  On Born's conjecture about optimal distribution of charges for an infinite
ionic crystal.}  
\newblock J.\ Nonlinear Sci.\
\newblock {\bf 28} (2018),  1629–-1656.
 
 
 
  

\bibitem{Blanc}
{\sc X.~Blanc, M.~Lewin}.
\newblock {\em The crystallization conjecture: a review}.
\newblock EMS Surv.\ Math.\ Sci.\  
\newblock {\bf 2} (2015), 255--306.

 
\bibitem{B43}
{\sc D.~C.~Brydges, P.~A.~Martin}. 
\newblock {\em Coulomb systems at low   density: A review}.
\newblock J.\ Stat.\ Phys.\
\newblock {\bf 96} (1999),  1163--1330.

  



 

\bibitem{cicalese}
 {\sc M.~Cicalese, G.~P.~Leonardi}. {\em Maximal fluctuations on periodic lattices: an approach via quantitative Wulff inequalities
}.  Commun.\ Math.\ Phys., to appear. \EEE  Preprint at \href{http://cvgmt.sns.it/paper/4195/}{\tt http://cvgmt.sns.it/paper/4195/}
  




\bibitem{Davoli15}
{\sc E.~Davoli, P.~Piovano, U.~Stefanelli}. 
\newblock {\em Wulff shape emergence in graphene}. 
\newblock Math.\ Models Methods Appl.\ Sci.\ 
\newblock {\bf 26} (2016), 12:2277--2310.


\bibitem{Davoli16}
{\sc E.~Davoli, P.~Piovano, U.~Stefanelli}. 
\newblock {\em Sharp $n^{3/4}$ 
 law for the minimizers of the edge-isoperimetric problem in the
 triangular lattice}. 
\newblock J.\  Nonlin.\ Sci.\
\newblock  {\bf 27} (2017),  627--660.



\bibitem{Lucia}
{\sc L.~De Luca, G.~Friesecke}. 
\newblock {\em  Crystallization in two dimensions and a discrete Gauss–Bonnet Theorem}. 
\newblock 
J.\ Nonlinear Sci.\
\newblock {\bf 28} (2017), 69--90.



 
 
 \bibitem{ELi}
 {\sc   W.~E, D.~Li}.
\newblock {\em  On the crystallization of 2D hexagonal lattices}.
\newblock Comm.\ Math.\ Phys.\
\newblock {\bf 286} (2009), 1099--1140.



 

\bibitem{Smereka15}
{\sc B.~Farmer, S.~Esedo\={g}lu, P.~Smereka}. 
\newblock {\em Crystallization for a
Brenner-like potential}.  
\newblock  Comm.\ Math.\ Phys.\
\newblock {\bf 349} (2017),  1029--1061. 
  

\bibitem{Figalli-Maggi-Pratelli}
 {\sc A.~Figalli, F.~Maggi, A.~Pratelli}. 
\newblock {\em A mass transportation approach to quantitative isoperimetric inequalities}. 
\newblock Inventiones mathematicae
\newblock {\bf 182} (2010),167--211.
 
 \bibitem{Flateley1}
{\sc L.~Flatley, M.~Taylor, A.~Tarasov, F~ Theil}. 
\newblock {\em Packing twelve spherical caps to maximize tangencies}. 
\newblock J.\ Comput.\ Appl.\ Math.\
\newblock {\bf 254} (2013),220--225.



 \bibitem{Flateley2}
{\sc L.~Flatley, F.~Theil}. 
\newblock {\em Face-centered cubic crystallization of atomistic configurations}. 
\newblock Arch.\ Ration.\ Mech.\ Anal.\
\newblock {\bf 218} (2015), 363--416.



  

 

 



 \bibitem{FriedrichKreutzHexagonal}
 {\sc  M.~Friedrich, L.~Kreutz}.
 \newblock {\em Crystallization in the hexagonal lattice for ionic dimers}.
\newblock  Math.\ Models Methods Appl.\ Sci.\ 
\newblock {\bf 29} (2019), 1853--1900. \EEE


 



\bibitem{Friesecke-Theil15}
{\sc G.~Friesecke, F.~Theil}. 
\newblock {\em Molecular geometry optimization,
  models}. In the Encyclopedia of Applied and Computational Mathematics,
B. Engquist (Ed.), Springer, 2015.











 



 \bibitem{Gardner}
{\sc C.~S.~Gardner, C.~Radin}. 
\newblock {\em The infinite-volume ground state of the Lennard-Jones potential}. 
\newblock J.\ Stat.\ Phys.\
\newblock {\bf 20} (1979), 719--724.



   



\bibitem{Hamrick}
{\sc G.~C.~Hamrick and C.~Radin}. 
\newblock {\em The symmetry of ground states under perturbation}. 
\newblock J.\ Stat.\ Phys.\
\newblock {\bf 21} (1979), 601--607.
  

\bibitem{HR}
{\sc  R.~Heitman, C.~Radin}. 
\newblock {\em Ground states for sticky disks}. 
\newblock J.\ Stat.\ Phys.\
\newblock {\bf 22} (1980), 3:281--287. 


 


 
 





\bibitem{cronut}
{\sc G.~Lazzaroni, U.~Stefanelli}. {\em Chain-like minimizers in three
  dimensions}.
 Transactions of Mathematics and Its Applications {\bf 2} (2018), tny003. \EEE 


\bibitem{Lewars}
{\sc E.~G.~Lewars}.
\newblock {\em Computational Chemistry}. 2nd edition, Springer, 2011.

 
 

\bibitem{Mainini-Piovano}
{\sc E.~Mainini, P.~Piovano, U.~Stefanelli}. 
\newblock {\em Finite crystallization in the square lattice}. 
\newblock Nonlinearity
\newblock {\bf 27} (2014), 717--737.
 


\bibitem{Mainini-Piovano-schmidt}
{\sc E.~Mainini, P.~Piovano, B.~Schmidt, U.~Stefanelli}. 
\newblock {\em $N^3/4$ law in the cubic lattice}. 
\newblock Preprint at \href{https://arxiv.org/abs/1807.00811}{\tt arXiv:1807.00811}.

 




 
\bibitem{Mainini}
{\sc E.~Mainini, U.~Stefanelli}. 
\newblock {\em Crystallization in carbon nanostructures}. 
\newblock Comm.\ Math.\ Phys.\
\newblock {\bf 328} (2014), 545--571. 



 
 
 

 \bibitem{Pauling}
{\sc L.~Pauling}. 
\newblock {\em The nature of the chemical bond and the structure of molecules and crystals: an introduction to modern structural chemistry}.
\newblock Ithaca, New York: Cornell University Press 1960. 


 
\bibitem{Radin}
{\sc  C.~Radin}. 
\newblock {\em The ground state for soft disks}. 
\newblock J.\ Stat.\ Phys.\
\newblock {\bf 26} (1981), 2:365--373. 




  \bibitem{Radi}
{\sc C.~Radin}. 
\newblock {\em Classical ground states in one dimension}. 
\newblock  J.\ Stat.\ Phys.\
\newblock {\bf 35} (1983), 109--117.
  

 

 

 
 \bibitem{B195}
{\sc C.~Radin}. 
\newblock {\it Crystals and quasicrystals: a continuum
  model}. 
  \newblock Comm.\ Math.\ Phys.\
  \newblock {\bf 105} (1986),  385--390.
  


 
 
 

 

 
 
\bibitem{Schmidt-disk}
{\sc  B.~Schmidt}. 
\newblock {\em Ground states of the 2D sticky disc model: fine properties and $N^{3/4}$ law for the
deviation from the asymptotic Wulff-shape}. 
\newblock J.\ Stat.\ Phys.\
\newblock {\bf 153} (2013), 727--738.

\bibitem{Suto06}
{\sc A.~S\"ut\H{o}}. 
\newblock {\em From bcc to fcc: Interplay between oscillation long-range
and repulsive short
range forces}.  
\newblock Phys.\ Rev.\ B
\newblock  {\bf 74} (2006), 104117.

 
 \bibitem{Theil}
{\sc   F.~Theil}. 
\newblock {\em A proof of crystallization in two dimensions}. 
\newblock Comm.\ Math.\ Phys.\ 
\newblock {\bf 262} (2006), 209--236.

 
 
  
  \bibitem{Ventevogel}
{\sc W.~J.~Ventevogel, B.~R.~A.~Nijboer}. 
\newblock {\em On the configuration of systems of interacting atom with
minimum potential energy per atom}. 
\newblock  Phys.\ A
\newblock {\bf 99} (1979), 565--580.


  

 




 


\bibitem{Wagner83} {\sc H. J. Wagner}. {\em Crystallinity in two dimensions: a note on a paper of C. Radin}.  J.\ Stat.\ Phys.\   {\bf 33},  (1983),  523--526.





 


\end{thebibliography}
\end{document}